\documentclass[12pt, onecolumn]{IEEEtran} 
\usepackage{url}
\usepackage{makecell}
\usepackage{amsfonts}
\usepackage{amssymb}
\usepackage{amsmath}

\usepackage{graphicx, colordvi, psfrag}
\usepackage{calc,pstricks, pgf, xcolor}
\usepackage{calc, pgf, xcolor}
\usepackage{epsfig, cite}

\usepackage{bm} 
\usepackage{bbm} 
\usepackage{enumerate}

\usepackage{color} 
\usepackage[hidelinks]{hyperref} 
\hypersetup{
    colorlinks=false, 
    linktoc=all 
}

\newcommand{\beq}[1]{\begin{equation}\label{#1}}
\newcommand{\eeq}{\end{equation}}

\newcommand{\beqn}[1]{\begin{eqnarray}\label{#1}}
\newcommand{\eeqn}{\end{eqnarray}}

\newtheorem{thmbody}{Theorem}
\newenvironment{thm}{
\begin{thmbody}
	}{
	\end{thmbody} 
	}
\newtheorem{dfnbody}{Definition}

\newtheorem{corbody}{Corollary}
\newenvironment{cor}{
\begin{corbody}
	}{
	\end{corbody} 
	}
\newtheorem{lemmabody}{Lemma}
\newenvironment{lemma}{
\begin{lemmabody}
	}{
	\end{lemmabody} 
	}
\newtheorem{propbody}{Proposition}

\newenvironment{proof}{
	{\it Proof:}
	}{
 $\Box$
	}

\usepackage{dsfont}
\begin{document}
\title{Channel input adaptation\\ via natural type selection}

\author{
Sergey~Tridenski and Ram~Zamir,~\IEEEmembership{Fellow},~\IEEEmembership{IEEE}
\thanks{
The material in this paper was partially presented in ISIT2017 \cite{TridenskiZamir17} and ISIT2018 \cite{TridenskiZamir18}.}
\thanks{This work of S. Tridenski and R. Zamir was partially supported by the Israel Science Foundation (ISF), grant \# 676/15.}
}

\maketitle
\begin{abstract}
We consider a channel-independent decoder which is for i.i.d. random codes
what the maximum mutual-information decoder is for constant composition codes.
We show that this decoder results in exactly the same i.i.d. random coding error exponent and
almost the same correct-decoding exponent
for a given 
codebook distribution as the maximum-likelihood decoder.
We propose an algorithm for computation of the optimal correct-decoding exponent which operates on the corresponding expression for the channel-independent decoder.
The proposed algorithm comes in two versions:
computation at a fixed rate and for a fixed slope.
The fixed-slope version of the algorithm presents an alternative to the Arimoto algorithm
for computation of the random coding exponent function in the correct-decoding regime.
The fixed-rate version of the computation algorithm
translates into a stochastic iterative algorithm for adaptation of the i.i.d. codebook distribution to a discrete memoryless channel
in the limit of large block length.
The adaptation scheme uses i.i.d. random codes with the channel-independent decoder and relies on one bit of feedback per transmitted block.
The communication itself is assumed 
reliable at a constant rate $R$.
In the end of the iterations
the resulting codebook distribution
guarantees reliable communication for all rates below $R + \Delta$ for some predetermined
parameter of decoding confidence
$\Delta > 0$,
provided that $R + \Delta$ is less than the channel capacity.
\end{abstract}


\markboth
{Submitted to the IEEE Trans. on Information Theory}
{Tridenski and Zamir: Channel input adaptation via natural type selection} 

%
%
%
%


\section{Introduction}

Consider a standard information theoretic scenario of communication through a discrete memoryless channel $P(y\,|\,x)$ using block codes.
For this case information theory provides optimal solutions in the form of the channel input distribution ${Q\mathstrut}^{*}(x)$, achieving the Shannon capacity $C$
or achieving the Gallager error exponent $E(R)$ for a given communication rate $R$.
Suppose, however, that the channel stochastic matrix $P(y\,|\,x)$ is slowly, or rarely, changing with time and we would like to sustain reliable communication at a constant rate $R$. For this purpose we assume
using a single bit of feedback, from the receiver to the transmitter, per transmitted block (Fig.~\ref{fig01}).
In our model we further assume that
given this bit of feedback
the codebook is updated
using the last transmitted block only,
i.e. without memory from the previous blocks.
So that potentially the system will follow the changes in the channel more closely.
Our goal of sustaining reliable communication at a constant rate $R$ is legitimate and feasible, of course, only as long as the capacity of the channel $C$ as a function of $P(y\,|\,x)$ stays above the rate $R$.
While the channel capacity may stay well above the rate, the optimal solution ${Q\mathstrut}^{*}(x)$ may drift significantly, as a result of the drift in $P(y\,|\,x)$, and render the initial code unreliable.

In this work the block code is modeled as a random code generated i.i.d. with a distribution $Q$.
The reason for modeling the code as an i.i.d. random code is twofold.
First, random codes
achieve capacity.
The idea is to
choose some positive constant $\Delta > 0$
and, by changing $Q$, to keep the {\em correct-decoding} random coding exponent
for a given $Q$
\cite[eq.~31]{TridenskiZamir17}, \cite{Arimoto76},
``pinned'' to zero at a rate $R' = R + \Delta$
provided that $R + \Delta < C$.
This would mean that the corresponding {\em error} exponent 
for the same $Q$
\cite[eq.~5.6.28]{Gallager} is strictly positive for all rates below $R + \Delta$, 
thus ensuring in particular reliable communication at rate $R$ (Fig.~\ref{graph02}).

Secondly, an i.i.d. distribution in a random code,
as opposed for example to a constant composition codebook,
results in a certain diversity of the codeword types,
which allows us to invoke a mechanism of {\em natural type selection} for updating of the parameter $Q$.
Using this mechanism iteratively, we successively update the codebook distribution $Q$ so that eventually the correct-decoding exponent associated with $Q$ decreases to zero at $R + \Delta$, thus achieving our goal.

The mechanism of natural type selection (NTS) has been originally observed and studied in the lossy source-coding setting \cite{ZamirRose01}, \cite{KochmanZamir02}.
In that setting a discrete memoryless source is mapped into a reproduction codebook, generated i.i.d. according to a distribution $Q$.
In the encoding process
a linear 
search is performed through the codebook until the first reproduction sequence is found, which
satisfies the distortion constraint $D$ with respect
to the source sequence.
Since various types are inherently present in the i.i.d. codebook,
the empirical distribution of the winning reproduction sequence in general is different than $Q$ and
is used for generating the next codebook.
This results in a decrease in the compression rate which after repeated iterations converges to the optimum given by the rate-distortion function $R(D)$.
This last property is guaranteed by the fact that both the conditional type given the source sequence
and the marginal type of the winning sequence with high probability
evolve along two parallel steps of the Blahut algorithm for the rate-distortion function computation \cite{Blahut72}, \cite{CsiszarTusnady84}.

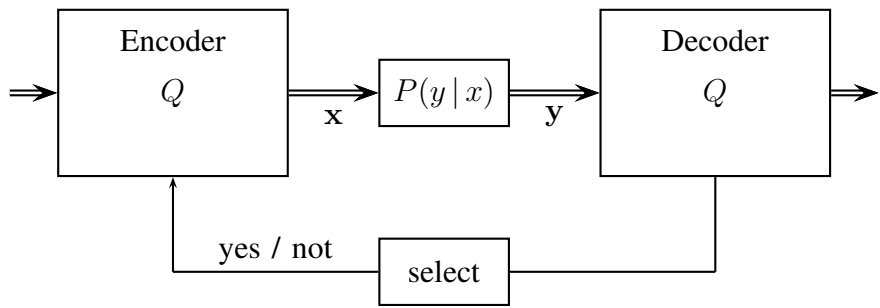
\begin{figure}[!t]
\psset{unit=.7mm}
\begin{center}
\begin{pspicture}(-20, 28)(145, 85)

\psframe(-11, 53)(33, 85)
\rput(11, 79){Encoder}
\rput(11, 69){$Q$}
\psline[doubleline = true]{->}(-20, 69)(-11, 69)
\psline[doubleline = true]{->}(33, 69)(50, 69)

\rput(41.5, 64.5){${\bf x}$}

\psframe(50, 62.5)(75, 75.5) \rput(62.5, 69){$P(y\,|\,x)$}

\rput(83.5, 64.5){${\bf y}$}

\psline[doubleline = true]{->}(136, 69)(145, 69)
\psline[doubleline = true]{->}(75, 69)(92, 69)

\psframe(92, 53)(136, 85)
\rput(114, 79){Decoder}
\rput(114, 69){$Q$}

\psframe(50, 28.5)(75, 41.5)

\rput(62.5, 35){select}

\psline{-}(114, 53)(114, 35)
\psline{-}(75, 35)(114, 35)

\psline{-}(11, 35)(50, 35)
\psline{<-}(11, 53)(11, 35)

\rput(30.5, 39){yes / not}

\end{pspicture}
\end{center}
\caption{DMC with a $1$-bit feedback per block. Each symbol in the random block code is generated i.i.d. according to $Q(x)$.}
\label{fig01}
\end{figure}

\begin{figure}[t]
\centering
\includegraphics[width=.50\textwidth]{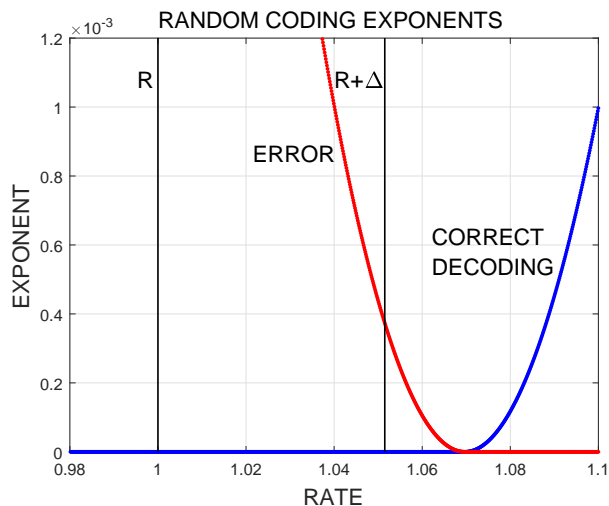}  
\caption{The i.i.d. random-coding error and correct-decoding exponents for a given $Q$.
Both exponents meet zero at the same point which is the mutual information of the joint distribution defined by the codebook distribution $Q(x)$
and the conditional channel distribution $P(y\,|\,x)$. The correct-decoding exponent is zero at $R + \Delta$.}
\label{graph02}
\end{figure}

We propose an analogous scheme for noisy-channel coding, equipped with its own computation algorithm for channels
which is
reminiscent of the Blahut algorithm for sources.
There is a certain analogy between the distortion constraint $D$ in lossy source-coding and the parameter $\Delta$ of the present scheme.
The higher is $D$ -- the poorer is the reproduction fidelity but the smaller is the communication penalty $R(D)$.
In our case, the higher is $\Delta$ -- the wider is the gap to capacity $\,C - R \geq \Delta\,$ but the higher is the communication reliability $E(R)$.
In order to implement
this analogy, we replace the $\log$-channel metric $\,\log P(y \, | \, x)\,$ of the maximum-likelihood (ML) decoder with
a more mathematically suitable channel-independent metric
\begin{equation} \label{eqNewMetric}
\log \frac{{V\!\mathstrut}_{m}(x \, | \, y)}{Q(x)},
\end{equation}
introduced in \cite{Merhav17}, see also \cite[eq.~16]{Merhav13},
where ${V\!\mathstrut}_{m}(x \, | \, y)$ is the conditional type of the codeword for message $m$ given the received block.
The decoder searches for the maximal empirical average of this metric among the codewords in the codebook.
If only a single codeword attains the maximum,
the decoder then compares the difference between the maximal empirical average of (\ref{eqNewMetric}) and the second highest one in the codebook to the parameter $\Delta$.
If the winning codeword wins by more than $\Delta$, its empirical distribution is selected as the new codebook distribution.

We derive expressions for the i.i.d. random-coding error and correct-decoding exponents for a given $Q$
of the decoder (\ref{eqNewMetric}) in a form of minima. They turn out to be equivalent to the corresponding exponents of the ML decoder.
The obtained expression for the i.i.d. correct-decoding exponent for a given $Q$
is used as a vehicle for iterative minimization of this exponent over $Q$ at a fixed rate $R' = R + \Delta$.
This procedure is comparable to the fixed-distortion version of the Blahut algorithm for $R(D)$.
We use the fixed-rate iterative minimization procedure as a basis for our stochastic adaptation scheme.

A fixed-slope version of the same computation is also presented.
This is
comparable to the fixed-slope version of the Blahut algorithm for $R(D)$,
but also presents an alternative to the Arimoto algorithm for fixed-slope computation of the correct-decoding exponent {\em function}
\cite{Arimoto76}, and to a similar recent algorithm \cite{OohamaJitsumatsu15}.

In Section~\ref{MLtoMMI} we introduce the channel-independent metric (\ref{eqNewMetric}).
In Section~\ref{ErrorExp} we derive its i.i.d. random coding error exponent for a given $Q$.
In Section~\ref{CorDecExp} we extend this analysis to the correct-decoding exponent.
In Section~\ref{CordecExpComp} the procedure of iterative minimization of the correct-decoding exponent at fixed $R$ is introduced,
and in Section~\ref{ComptoBlahut} it is compared to the fixed-distortion version of the Blahut algorithm \cite{CsiszarTusnady84}.
In Section~\ref{FixedSlope} we apply the fixed-slope version of the procedure for computation of the correct-decoding exponent function --
as an alternative to the Arimoto algorithm \cite{Arimoto76}.
In Section~\ref{ChInAd} the stochastic adaptation scheme based on the fixed-rate minimization of Section~\ref{CordecExpComp} is presented.
Section~\ref{Conc} summarizes the paper. Some technical details are deferred to the Appendix.

\bigskip

\section{ML replaced with a channel-independent decoder}\label{MLtoMMI}

\bigskip

In order to build up a framework,
we 
replace the optimal maximum-likelihood (ML) decoder with a naturally fitting channel-independent decoder.
This alternative decoder will still be optimal in the exponential sense for i.i.d random codes\footnote{And for constant composition codes as well,
becoming the maximum mutual-information decoder.}.

Let $P(y \, | \, x)$ denote a discrete memoryless channel with letters from finite input and output alphabets, $x\in {\cal X}$ and $y \in {\cal Y}$,
respectively.
Given a positive real rate $R$ and a blocklength $n$, the codebook consists of $\Big\lceil {e\mathstrut}^{nR}\Big\rceil$
codewords of length $n$, generated independently with an i.i.d. distribution $Q(x)$.

Consider now the decoding procedure. Let $T(y)$ denote the type of the received word ${\bf y} \in {\cal Y\mathstrut}^{n}$ from the channel.
For each codeword ${\bf x\mathstrut}_{m} \in {\cal X\mathstrut}^{n}$ in the codebook, where $m = 1, ..., \Big\lceil {e\mathstrut}^{nR}\Big\rceil$,
let ${V\!\mathstrut}_{m}(x \, | \, y)$ denote its conditional type given ${\bf y}$.
ML decoding amounts to evaluating the average
\begin{equation} \label{eqML}
\sum_{x, \, y}T(y){V\!\mathstrut}_{m}(x \, | \, y) \underbrace{\log P(y \, | \, x)}_{\text{metric}}
\; \triangleq \; {\mathbb{E}\mathstrut}_{m} \big[\log P(Y \, | \, X)\big]
\end{equation}
for each $m$ and looking for the maximum over $m$. The logarithm of the channel probability mass function $\log P(y \, | \, x)$
under the average
can be viewed as a decoding metric,
i.e.,
a replaceable function of $x$ and $y$,
which can be replaced by another metric,
resulting in a different (suboptimal) decoder.
We replace the $\log$-channel metric with the metric $\log \frac{{V}_{m}(x \, | \, y)}{Q(x)}$ in (\ref{eqNewMetric}),
which varies with $m$ and resembles the pointwise mutual information. The new decoder evaluates the metric average
\begin{equation}\label{eqMMIlikeDec}
\sum_{x, \, y}T(y){V\!\mathstrut}_{m}(x \, | \, y) \underbrace{\log \frac{{V\!\mathstrut}_{m}(x \, | \, y)}{Q(x)}}_{\text{metric}}
\; \equiv \;
{\mathbb{E}\mathstrut}_{m} \left[\log \frac{{V\!\mathstrut}_{m}(X \, | \, Y)}{Q(X)}\right]
\end{equation}
for each $m$ and chooses the particular $m$ corresponding to the maximal average, as its estimate $\widehat{m}$ of the sent message.
If there is more than one such $m$, then the decoder may break the tie
arbitrarily
or, alternatively, declare an error.
As in the ML decoding, a particular choice between the 
different scenarios
of tie breaking has no effect on the resulting {\em error} exponent.

\bigskip

\section{Error exponent}\label{ErrorExp}

\bigskip

In this section we derive the i.i.d. random coding error exponent in decoding using the metric average (\ref{eqMMIlikeDec})
and a decoding confidence parameter $\Delta \geq 0$ (decoding threshold).
This will correspond to different exponents in our adaptation scheme which we present later in the paper.
Specifically, $\Delta = 0$ will give the error exponent of reliable communication at rate $R$.
The case $\Delta > 0$ will give the probability exponent of the event when the feedback bit is $0$ --
when the codebook distribution $Q$ is not updated. First we derive the exponent in a form of a minimum,
which we call an implicit expression.

\bigskip

\subsection{Implicit expression for the channel-independent metric}\label{ErrorExpA}

\bigskip

Let us define an implicit expression, which will shortly acquire a meaning of an i.i.d. random coding error exponent:
\begin{equation} \label{eqDefImplicit}
{E\mathstrut}_{\!e}(R, Q)
\; \triangleq \;
\min_{\substack{\\T(y), \, V(x\,|\,y)
}}
\left\{
D(T\circ V \, \| \, Q \circ P) \, + \,
\big|D(T\circ V \,\|\, T \times Q) - R\,\big|^{+}
\right\},
\end{equation}
where the notation $|a - b \,|^{+}$ preserves the difference if $\,a - b > 0$, and nullifies it otherwise.
Note that the metric average (\ref{eqMMIlikeDec}) appears in this expression inside the bars $|\,\cdot\,|^{+}$,
compared to the rate $R$,
in the form of the Kullback-Leibler divergence from the product distribution $T(y)Q(x)$,
denoted as $T\times Q$, to the joint distribution $T(y)V(x \, | \, y)$,
denoted as $T\,\circ\, V$.
First we establish the expression (\ref{eqDefImplicit}) as the i.i.d. random coding error exponent of the decoder
described by (\ref{eqMMIlikeDec}).
Afterwards we show, by deriving an explicit expression for (\ref{eqDefImplicit}), that it is equivalent to the optimal i.i.d. random coding error exponent,
achieved by the ML decoder.

The expression (\ref{eqDefImplicit}) -- as the i.i.d. random coding error exponent of the decoder
(\ref{eqMMIlikeDec}) -- can be obtained from \cite[eq.~2.28]{Scarlett14}
when the metric (\ref{eqNewMetric}) is used as a mismatched metric.

In the current paper the expression (\ref{eqDefImplicit}) appears also in a slightly more general setting -- with a decoding confidence parameter $\Delta$.
To this end,
let us further generalize and expand the decoder structure by {\em splitting} the decoding metric. Suppose, 
generally
the decoder uses {\em two} different functions of the joint distribution $T(y)V(x \, | \, y)$, denoted $A(T\circ V)$ and $B(T\circ V)$, such that always
\begin{equation} \label{eqErasureProp}
A(T\circ V) \; \leq \; B(T\circ V), \;\;\;\;\;\; \forall \;\; T\circ V.
\end{equation}
Given a received word ${\bf y}$ of type $T(y)$,
for each message $m = 1, ..., \Big\lceil {e\mathstrut}^{nR}\Big\rceil$ the decoder evaluates both $A(T\circ {V\!\mathstrut}_{m})$ and $B(T\circ {V\!\mathstrut}_{m})$. If there exists a message $m$ such that
\begin{equation} \label{eqDecRule}
A(T\circ {V\!\mathstrut}_{m}) \; > \; B(T\circ {V\!\mathstrut}_{m'}), \;\;\;\;\;\; \forall \;\; m' \, \neq \, m,
\end{equation}
then by the property (\ref{eqErasureProp}) there exists only one such $m$ in the code. The decoder looks for and chooses this winning $m$ as its estimate
$\widehat{m}$ of the sent message.
If there is no such $m$, the decoder declares an erasure.
Let 
${P\mathstrut}_{\!e}$
denote the i.i.d. random code ensemble average probability of erasure and undetected error
(when the winning $\widehat{m}$ exists but is a wrong message)
combined, in this decoding scenario according to the rule (\ref{eqDecRule}).
\bigskip
\begin{thm}\label{thm1}\newline
{\em Let $B(T\circ V) \equiv D(T\circ V \,\|\, T \times Q)$ and let $A(T\circ V)\leq D(T\circ V \,\|\, T \times Q)$ be a continuous function of $T\circ V$ in the support of $Q\circ P$.
Then}
\begin{equation} \label{eqErrorExp1}
\lim_{n\,\rightarrow \,\infty}\,\frac{\log{P}_{e}}{-n}
\;\; = \;\;
\min_{\substack{\\T(y), \, V(x\,|\,y)
}}
\left\{
D(T\circ V \, \| \, Q \circ P) \, + \,
\big|A(T\circ V) - R\,\big|^{+}
\right\}.
\end{equation}
\end{thm}

\bigskip

\begin{cor}[Error exponent of the natural decoder]\label{cor1}\newline
{\em Let $B(T\circ V) \equiv D(T\circ V \,\|\, T \times Q)$ and $A(T\circ V) \equiv D(T\circ V \,\|\, T \times Q) - \Delta$, with $\Delta \geq 0$. Then}
\begin{equation} \label{eqErrorExpConfidence}
\lim_{n\,\rightarrow \,\infty}\,\frac{\log{P}_{e}}{-n}
\;\; = \;\;
{E\mathstrut}_{\!e}(R + \Delta, \,Q),
\end{equation}
{\em where the RHS is defined in (\ref{eqDefImplicit}).}
\end{cor}

\bigskip

The exponent (\ref{eqErrorExpConfidence}) presents an analog of the error exponent of the Forney simplified decoder \cite[eq.~18-20]{TridenskiZamir17}, \cite{WeinbergerMerhav15}, \cite{Forney68}. 
Unlike in the case of the simplified decoder \cite[eq.~18]{TridenskiZamir17}, which uses the ML metric (\ref{eqML}),
in the error exponent (\ref{eqErrorExpConfidence}) which is derived for the metric (\ref{eqMMIlikeDec}) the decoding threshold $\Delta$ results in a simple shift of the exponent as a function of $R$.

As we will see later,
in the adaptation scheme presented in the current work the error exponent (\ref{eqErrorExpConfidence}) with $\Delta > 0$ will correspond to
the probability exponent of the event when the feedback bit is $0$, which signals
to the transmitter not to update the codebook distribution $Q$.

As a special case of (\ref{eqErrorExpConfidence}), with $A(T\circ V) \equiv B(T\circ V) \equiv D(T\circ V \,\|\, T \times Q)$,
the expression (\ref{eqDefImplicit}) now becomes the error exponent in the decoding using the metric (\ref{eqMMIlikeDec}) as introduced in the previous section.
In this case the erasure event is a tie which leads to an error with probability at least $1/2$
(if the tie breaking is performed).
Alternatively, the special case of (\ref{eqErrorExpConfidence}) with $\Delta = 0$ can be obtained from \cite[eq.~2.28]{Scarlett14}
when the metric (\ref{eqNewMetric}) is used as a mismatched metric.

Other examples of $A(T\circ {V\!\mathstrut}_{m})$ satisfying the condition (\ref{eqErasureProp}) of Theorem~\ref{thm1}
with $D(T\circ {V\!\mathstrut}_{m} \,\|\, T \times Q) \equiv \newline B(T\circ {V\!\mathstrut}_{m})$
are obtained by substituting in (\ref{eqMMIlikeDec}) alternative metrics
\begin{equation} \label{eqFamily}
\log \frac{{V\!\mathstrut}_{m}(x \, | \, y)}{{T\!\mathstrut}_{m}(x)},
\;\;\;\;\;\;
\log \frac{\Phi(x \, | \, y)}{Q(x)},
\;\;\;\;\;\;
\log \frac{\Phi(x \, | \, y)}{{T\!\mathstrut}_{m}(x)},
\;\;\;\;\;\;
\log \frac{P(y \, | \, x)}{T(y)},
\end{equation}
where
${T\!\mathstrut}_{m}(x) = T\circ {V\!\mathstrut}_{m}(x)$ is the marginal type of ${\bf x\mathstrut}_{m}$ and
$\Phi(x \, | \, y)$ is some fixed conditional distribution.
Likewise, subtracting $\,\Delta > 0\,$ in $\,A(T\circ {V\!\mathstrut}_{m})\,$ of Theorem~\ref{thm1} results in a simple shift to the left of the exponent (\ref{eqErrorExp1})
as a function of $R$. Only the last metric among the above examples assumes the knowledge of the channel $P(y \, | \, x)$ by the receiver.
This does not lead to the ML decoding of course, because $A(T\circ {V\!\mathstrut}_{m})$ is compared to $B(T\circ {V\!\mathstrut}_{m})$.

To prove Theorem~\ref{thm1} we first need the following auxiliary result.
For any functions $A(T\circ V)$ and $B(T\circ V)$, not necessarily satisfying (\ref{eqErasureProp}),
let ${P\mathstrut}_{\!A \, \leq \, B}$ denote the i.i.d. random code ensemble average probability
that the sent message $m$ does not satisfy the winning condition (\ref{eqDecRule}).
Then for finite blocklength $n$ we have:
\bigskip
\begin{lemma} \label{lemTypes}
\begin{equation} \label{eqLemma}
\frac{\log{P}_{A \, \leq \, B}}{-n}
\;\; = \;\;
\min_{\substack{\\\text{types} \;\; T(y), \, V(x\,|\,y), \, \widehat{V}(x\,|\,y):\\
\\
A(T \,\circ \,V)\;\leq\;
B(T \,\circ \,\widehat{V})
}}
\left\{
D(T\circ V \, \| \, Q \circ P) \, + \,
\big|D(T\circ \widehat{V} \,\|\, T \times Q) - R\,\big|^{+}
\right\} \; + \; o(1),
\end{equation}
{\em where the minimization is over the types, corresponding to the blocklength $n$.}
\end{lemma}
The proof is given in Appendix.

\bigskip

{\em Proof of Theorem~\ref{thm1}:}
By Lemma~\ref{lemTypes}
\begin{align}
\frac{\log{P}_{e}}{-n}
\;\; & = \;\; \min_{\substack{\\\text{types} \;\;  T(y), \, V(x\,|\,y), \, \widehat{V}(x\,|\,y):\\
\\
A(T \,\circ \,V)\;\leq\;
D(T \,\circ \,\widehat{V} \, \| \, T \,\times\, Q)
}}
\left\{
D(T\circ V \, \| \, Q \circ P) \, + \,
\big|D(T\circ \widehat{V} \,\|\, T \times Q) - R\,\big|^{+}
\right\} \; + \; o(1) 
\label{eqStructure} \\
& = \;\;
\;\;\;\;\;\,
\min_{\substack{\\\text{types} \;\; T(y), \, V(x\,|\,y)
}}
\;\;\;\;\;\,
\left\{
D(T\circ V \, \| \, Q \circ P) \, + \,
\big|A(T\circ V) - R\,\big|^{+}
\right\}
\; + \; o(1),
\label{eqErrorExpExact}
\end{align}
where the $\min$'s are over types.
The second equality holds because
$A(T\circ V) \leq D(T\circ V \,\|\, T \times Q)$,
guaranteeing
that the minimizing type $\widehat{V}$ in (\ref{eqStructure}) 
can be chosen such that
$D(T\circ \widehat{V} \,\|\, T \times Q)$ is uniformly close to $A(T\circ V)$.
More precisely
$\,D(T\circ \widehat{V} \,\|\, T \times Q) = \max\,\{0, \,A(T\circ V)\} + o(1)$.
This results in
the two minima over the types being equivalent up to a uniform additive constant $o(1)$.
In the limit, as $n\,\rightarrow\,\infty$, the term $o(1)$ disappears and the minimization is performed over all rational distributions $\,T \circ V$.
Since the objective function in the $\min$ of (\ref{eqErrorExpExact}) is a continuous function of $\,T \circ V$, the infimum over rational distributions equals the minimum over all distributions as intended in the statement of the theorem (\ref{eqErrorExp1}). \ensuremath{\square}

\bigskip

The i.i.d. random coding error exponent ${E\mathstrut}_{\!e}(R, Q)$, (\ref{eqDefImplicit}), can be easily compared to the constant-\newline
composition error exponent:
\begin{align}
{E\mathstrut}_{\!e}(R, Q) \;\;\; = \;\;\;\;
& \min_{\substack{\\T(y), \, V(x\,|\,y)
}}
\,\Big\{
D(T\circ V \, \| \, Q \circ P) \; + \;
\big|D(T\circ V \,\|\, T \times Q) \, - \, R\,\big|^{+}
\Big\}
\nonumber \\
\equiv \;\;\;\; & \!\min_{\substack{\\U(x), \, W(y\,|\,x)
}}
\Big\{
D(U\circ W \, \| \, Q \circ P) \, + \,
\big|\overbrace{I(U\circ W) \, + \, D(U \,\|\, Q)}^{D(T\,\circ\, V \;\|\; T \,\times\, Q)\;=} \, - \, R \,\big|^{+}
\Big\}
\nonumber \\
\overset{U = \, Q}{\leq} \;\; &
\;\;\, \min_{\substack{\\W(y\,|\,x)
}}
\;\;\;
\,\Big\{
D(Q\circ W \, \| \, Q \circ P) \, + \,
\big|\,I(Q\circ W)\, - \, R \,\big|^{+}
\Big\},
\label{eqCC}
\end{align}
where $I(U\circ W)$ denotes the mutual information $I(X; Y)$ of the joint distribution $U(x)W(y\,|\,x)$.
The expression (\ref{eqCC})
can be recognized as the constant-composition error exponent \cite[eq.~5.15]{CsiszarKorner}.
The proof of Theorem~\ref{thm1}~/~Corollary~\ref{cor1} can be repeated for the constant composition codes.
In that case the same metric
$\,\log \frac{{V}_{m}(x \, | \, y)}{Q(x)}\,$ from (\ref{eqMMIlikeDec}) becomes exactly the pointwise mutual information, resulting in the maximum mutual information (MMI) decoder \cite{CsiszarKorner}, and the minimum (\ref{eqCC}) replaces (\ref{eqDefImplicit}).

\bigskip

\subsection{Implicit expressions for the ML metric}\label{AboutML}

\bigskip

As we shall see,
the exponent (\ref{eqDefImplicit}) produced by the metric in (\ref{eqMMIlikeDec}) is optimal for i.i.d. random codes.
However, if we wish to address the decoder with the ML metric $\log P(y \, | \, x)$ directly with the help of Lemma~\ref{lemTypes}, we need to reformulate Theorem~\ref{thm1} as follows:

\bigskip

{\em Theorem 1$'$ (Error exponent of the split decoder)}\newline
{\em Let $A(T\circ V)$ and $B(T\circ V)$ be finite and such that the set of joint distributions over $(x, y, \widehat{x})\in {\cal X}\times {\cal Y}\times {\cal X}$}
\begin{displaymath}
\Big\{ T(y)V(x \, | \, y) \widehat{V}(\widehat{x} \, | \, y): \;\;\; A(T\circ V) \, - \, B(T\circ \widehat{V}) \; \leq \; 0 \Big\}
\end{displaymath}
{\em is the closure of its interior.
Then}
\begin{equation} \label{eqThm1Prime}
\lim_{n\,\rightarrow \,\infty}\,\frac{\log{P}_{A\, \leq \,B}}{-n}
\;\; = \;\;
\min_{\substack{\\T(y), \, V(x\,|\,y), \, \widehat{V}(\widehat{x}\,|\,y):\\
\\
A(T \,\circ \,V)\;\leq\;
B(T \,\circ \,\widehat{V})
}}
\left\{
D(T\circ V \, \| \, Q \circ P) \, + \,
\big|D(T\circ \widehat{V} \,\|\, T \times Q) - R\,\big|^{+}
\right\}.
\end{equation}

\begin{proof}
The minimum over types (\ref{eqLemma}) of Lemma~\ref{lemTypes} converges to (\ref{eqThm1Prime}) as in Sanov's theorem proof \cite{CoverThomas}.
\end{proof}

\bigskip

By this theorem, the i.i.d. error exponent of the ML decoder can be implicitly formulated using the ML metric $\log P(y \, | \, x)$
directly \cite[eq.~2.28]{Scarlett14} as
\begin{align}
& \;\, \min_{\substack{\\T(y), \, V(x\,|\,y), \, \widehat{V}(\widehat{x}\,|\,y):\\
\\
\sum_{x, \, y}T(y) V(x\,|\,y)\widehat{V}(\widehat{x}\,|\,y)\log \frac{P(y\,|\,x)}{P(y\,|\,\widehat{x})}
\;\leq\; 0
}}
\;
\left\{
D(T\circ V \, \| \, Q \circ P) \, + \,
\big|D(T\circ \widehat{V} \,\|\, T \times Q) - R\,\big|^{+}
\right\}
\nonumber \\
\equiv \;\; & \min_{\substack{\\T(y), \, V(x\,|\,y), \, \widehat{V}(\widehat{x}\,|\,x, \, y):\\
\\
\sum_{x, \, y}T(y) V(x\,|\,y)\widehat{V}(\widehat{x}\,|\,x, \, y)\log \frac{P(y\,|\,x)}{P(y\,|\,\widehat{x})}
\;\leq\; 0
}}
\left\{
D(T\circ V \, \| \, Q \circ P) \, + \,
\big|D\big((T\circ V)\circ \widehat{V} \,\|\, (T\circ V) \times Q\big) - R\,\big|^{+}
\right\}.
\nonumber
\end{align}
The second expression allows the i.i.d. {\em decoding error} exponent to be alternatively interpreted as an {\em encoding success} exponent
in the lossy encoding of an effective discrete memoryless source $Q(x)P(y\,|\,x)$ of pairs $(x, y)$
by a reproduction codebook generated i.i.d. according to $Q(\widehat{x})$, where the reproduction letter $\widehat{x}$ and $x$
share the same alphabet \cite{TridenskiZamir17}.
The distortion measure in this case is
$d\big((x, y), \widehat{x}\big) \triangleq \log \frac{P(y\,|\,x)}{P(y\,|\,\widehat{x})}$,
not necessarily positive,
and $0$ in the minimization condition is interpreted as a distortion constraint. The distortion constraint $0$ of the lossy encoding can be generalized to an arbitrary threshold, resulting in a family of exponents of the Forney simplified erasure/list decoder for channels \cite[eq.~18-20]{TridenskiZamir17}, \cite{WeinbergerMerhav15}, \cite{Forney68}.

It can be shown using a version of Theorem~1$'$ for the constant composition codes or \cite[eq.~2.28]{Scarlett14}, that the use of the ML metric $\log P(y \, | \, x)$ for these codes also leads to the exponent (\ref{eqCC}).

\bigskip

\subsection{Explicit expression and the minimizing solutions}\label{ErrorExpB}

\bigskip

Here we prove the identity \cite[eq.~28]{TridenskiZamir17} between the implicit expression (\ref{eqDefImplicit})
and the optimal Gallager expression \cite[eq.~5.6.28]{Gallager},
which we call the explicit expression.
By doing so we also obtain expressions for the minimizing distributions of (\ref{eqDefImplicit}).
The same expressions for the minimizing distributions will be used also in the subsequent sections for the correct-decoding exponent,
in the computation algorithms, and will represent (asymptotically) the types in our stochastic adaptation scheme.

The notation $|\,\cdot\, |^{+}$ in (\ref{eqDefImplicit}) 
can be interpreted as if
the $\min$ there splits into a minimum of two $\min$'s,
subject to the condition whether the divergence $D(T\circ V \,\|\, T \times Q)$ is less or greater than $R$.
The first $\min$ gives the exponent of the error event caused by the conditional types $V$ with $D(T\circ V \,\|\, T \times Q) < R$,
which by virtue of this inequality itself appear in the codebook with high probability. 
The second $\min$ gives the exponent of the error event caused by the conditional types $V$ with $D(T\circ V \,\|\, T \times Q) > R$,
which are exponentially rare in the codebook,
hence the additional positive term $D(T\circ V \,\|\, T \times Q) - R$ in the exponent.
Each one of these two parts is treated separately by the next two lemmas.
The lemmas replace the minima by supporting lines as functions of $R$ and also follow after the minimizing distributions.

\bigskip

\begin{lemma} [For $\,D(T\circ V \,\|\, T \times Q)\leq R\,$ part] \label{lemLeq}\newline
{\em For any $\rho \geq 0$}
\begin{equation} \label{eqPartLeq}
\min_{\substack{\\T(y), \, V(x\,|\,y):\\
\\
D(T \,\circ \,V \, \| \, T \,\times\, Q)\;\leq\;
R
}}
D(T\circ V \, \| \, Q \circ P)
\;\; \geq \;\;
\min_{\substack{\\T(y), \, V(x\,|\,y)}}
\Big\{
D(T\circ V \, \| \, Q \circ P)
\,
+ \,
\rho \big[D(T \circ V \, \| \, T \times Q) \,-\, R\big]
\Big\}.
\end{equation}
{\em
In the case of equality in (\ref{eqPartLeq}),
any minimizing solution of the LHS is also a minimizing solution of the RHS $\,{T\!\mathstrut}_{\rho} \circ {V\!\!\mathstrut}_{\rho}$
such that $R = D\big({T\!\mathstrut}_{\rho} \circ {V\!\!\mathstrut}_{\rho} \, \| \, {T\!\mathstrut}_{\rho} \times Q\big)$,
or $R \geq  D\big({T\mathstrut}_{0} \circ {V\!\mathstrut}_{0} \, \| \, {T\mathstrut}_{0} \times Q\big)$
for $\rho = 0$.
Conversely, if there exists such solution ${T\!\mathstrut}_{\rho} \circ {V\!\!\mathstrut}_{\rho}$ minimizing the RHS,
then it is also a minimizing solution of the LHS and there is equality in (\ref{eqPartLeq}).}
\end{lemma}

\bigskip

\begin{proof}
\begin{align}
&
\;\;\;\;\;\;\;\;\;\;\;\;\;\;\,
\min_{\substack{\\T(y), \, V(x\,|\,y):\\
\\
D(T \,\circ \,V \, \| \, T \,\times\, Q)\;\leq\;
R
}}
\;\;\;\;\;\;\;\;\;\,
\Big\{ D(T\circ V \, \| \, Q \circ P) \Big\}
\label{eqMinFirst} \\
&
\!\! \overset{< \, \infty}{=}
\;\;\;\;\;\;\;\;\;
\;\;\;\;\;\;\;\;\;\;\;\;\;\;\;\;\;\;\;\;\;\;\;\;\;\;\;\;
\;\;\;\;\;\;\;\,
D\big({T\!\mathstrut}_{R}\circ {V\!\!\mathstrut}_{R} \; \| \; Q \circ P\big)
\label{eqSolution} \\
&
\!\! \overset{\rho \, \geq \, 0}{\geq}
\;\;\;\;\;\;\;
\;\;\;\;\;\;\;\;\;\;\;\;\;\;\;\;\;\;\;\;\;\;\;\;\;\;\;\;
\;\;\;\;\;\;\;\;\;
D\big({T\!\mathstrut}_{R}\circ {V\!\!\mathstrut}_{R} \; \| \; Q \circ P\big)
\, + \,
\rho \underbrace{\big[D\big({T\!\mathstrut}_{R}\circ {V\!\!\mathstrut}_{R} \; \| \; {T\!\mathstrut}_{R} \times Q\big) - R\big]}_{\leq \, 0}
\label{eqExpressionR} \\
& \geq
\;\;\;\;\;\;\;\;\;\;\;\;\;\;\;
\min_{\substack{\\T(y), \, V(x\,|\,y)}}
\;\;\;\;\;\;\;\;\;\;\;\;\;\;\,
\Big\{
D(T\circ V \, \| \, Q \circ P)
\;\,
+ \;\,
\rho \,\big[D(T \circ V \, \| \, T \times Q) - R\big]
\Big\}
\label{eqMinSecond} \\
& =
\;\;\;\;\;\;\;\;\;\;\;\;\;\;\;\;\;\;\;\;\;\;\;\;\;\;
\;\;\;\;\;\;\;\;\;\;\;\;\;\;\;\;\;\;\;\;\;
D\big({T\!\mathstrut}_{\rho}\circ {V\!\!\mathstrut}_{\rho} \, \| \, Q \circ P\big)
\; + \;
\rho \,\big[D\big({T\!\mathstrut}_{\rho} \circ {V\!\!\mathstrut}_{\rho} \, \| \, {T\!\mathstrut}_{\rho} \times Q\big) - R\big]
\label{eqExpressionRho} \\
& \geq \,
\min_{\substack{\\T(y), \, V(x\,|\,y):\\ \\ D(T \,\circ \,V \, \| \, T \,\times\, Q) \;\leq\;
D({T\!\mathstrut}_{\rho} \,\circ \,{V\!\!\mathstrut}_{\rho} \, \| \, {T\!\mathstrut}_{\rho} \,\times\, Q)}}
\Big\{
D(T\circ V \, \| \, Q \circ P)\Big\}
\, + \,
\rho\,\big[D\big({T\!\mathstrut}_{\rho} \circ {V\!\!\mathstrut}_{\rho} \, \| \, {T\!\mathstrut}_{\rho} \times Q\big) - R\big].
\label{eqSandwich}
\end{align}
The first equality holds when the minimum (\ref{eqMinFirst})
is finite, and ${T\!\mathstrut}_{R}\circ {V\!\!\mathstrut}_{R}$ is a minimizing solution for a given $R$. 
In the second equality, ${T\!\mathstrut}_{\rho} \,\circ \,{V\!\!\mathstrut}_{\rho}$ is a minimizing solution of the minimum
(\ref{eqMinSecond}) for a given $\rho$.

Observe that when the minima (\ref{eqMinFirst}) and (\ref{eqMinSecond}) are equal,
then also there is equality between the expression (\ref{eqExpressionR}) and (\ref{eqMinSecond}).
Consequently, the minimizing distribution ${T\!\mathstrut}_{R}\circ {V\!\!\mathstrut}_{R}$ of (\ref{eqMinFirst})
is also a minimizing distribution of (\ref{eqMinSecond}) for a given $\rho$ in this case.
From the equality between (\ref{eqSolution}) and (\ref{eqExpressionR}) we conclude
that such solution must satisfy
$R = D\big({T\!\mathstrut}_{R} \circ {V\!\!\mathstrut}_{R} \; \| \; {T\!\mathstrut}_{R} \times Q\big)$ for $\rho > 0$,
or $R \geq D\big({T\!\mathstrut}_{R} \circ {V\!\!\mathstrut}_{R} \; \| \; {T\!\mathstrut}_{R} \times Q\big)$ for $\rho = 0$.

Conversely, the equality for $R = D\big({T\!\mathstrut}_{\rho} \circ {V\!\!\mathstrut}_{\rho} \, \| \, {T\!\mathstrut}_{\rho} \times Q\big)$,
or $R \geq D\big({T\mathstrut}_{0} \circ {V\!\mathstrut}_{0} \, \| \, {T\mathstrut}_{0} \times Q\big)$ for $\rho = 0$
follows from (\ref{eqSandwich}) by a sandwich proof.
In this case the second term in (\ref{eqExpressionRho}) becomes zero, while the difference in the square brackets is non-positive, and $D\big({T\!\mathstrut}_{\rho}\circ {V\!\!\mathstrut}_{\rho} \, \| \, Q \circ P\big)$ is equal to the minimum (\ref{eqMinFirst}).
Therefore ${T\!\mathstrut}_{\rho} \circ {V\!\!\mathstrut}_{\rho}$ is also a minimizing solution of (\ref{eqMinFirst}) for such $R$.
\end{proof}

\bigskip

\begin{lemma} [For $\,D(T\circ V \,\|\, T \times Q)\geq R\,$ part] \label{lemGeq}\newline
{\em For any $\rho \leq 1$}
\begin{align}
&
\;\;\;\;\,
\min_{\substack{\\T(y), \, V(x\,|\,y):\\
\\
D(T \,\circ \,V \, \| \, T \,\times\, Q)\;\geq\;
R
}}
\Big\{D(T\circ V \, \| \, Q \circ P)
\, + \,
D(T \circ V \, \| \, T \times Q) \, - \, R
\Big\}
\nonumber \\
&
\geq
\;\;\;\;\;
\min_{\substack{\\T(y), \, V(x\,|\,y)
}}
\;\;\;\;\;
\Big\{D(T\circ V \, \| \, Q \circ P)
\, + \,
\rho\big[D(T \circ V \, \| \, T \times Q) \, - \, R\big]
\Big\}.
\label{eqPartGeq}
\end{align}
{\em In the case of equality in (\ref{eqPartGeq}), any
minimizing solution of the LHS
is also a minimizing solution of the RHS
$\,{T\!\mathstrut}_{\rho} \circ {V\!\!\mathstrut}_{\rho}$
such that
$R = D\big({T\!\mathstrut}_{\rho} \circ {V\!\!\mathstrut}_{\rho} \, \| \, {T\!\mathstrut}_{\rho} \times Q\big)$,
or $R \leq D\big({T\mathstrut}_{1} \circ {V\!\mathstrut}_{1} \, \| \, {T\mathstrut}_{1} \times Q\big)$ for $\rho = 1$.
Conversely,
if there exists such solution ${T\!\mathstrut}_{\rho} \circ {V\!\!\mathstrut}_{\rho}$
minimizing the RHS,
then it is also a minimizing solution of the LHS and there is equality in (\ref{eqPartGeq}).}
\end{lemma}

\bigskip

\begin{proof}
Analogously to Lemma~\ref{lemLeq}:
\begin{align}
&
\;\;\;\;\;\;\;\;\;\;\;\;\;\;\;
\min_{\substack{\\T(y), \, V(x\,|\,y):\\
\\
D(T \,\circ \,V \, \| \, T \,\times\, Q)\;\geq\;
R
}}
\;\;\;\;\;\;\;\;\;\,
\Big\{D(T\circ V \, \| \, Q \circ P)
\, + \,
D(T \circ V \, \| \, T \times Q) \, - \, R
\Big\}
\nonumber \\
& \!\!
\overset{< \, \infty}{=}
\;\;\;\;\;\;\;\;\;\;\;\;\;\;\;\;\;\;\;\;\;\;\;\;\;\;
\;\;\;\;\;\;\;\;\;\;\;\;\;\;\;\;\;
D\big({T\!\mathstrut}_{R}\circ {V\!\!\mathstrut}_{R} \; \| \; Q \circ P\big)
\, + \,
D\big({T\!\mathstrut}_{R}\circ {V\!\!\mathstrut}_{R} \; \| \; T \times Q\big) \, - \, R
\nonumber \\
&
\!\!
\overset{\rho \, \leq \, 1}{\geq}
\;\;\;\;\;\;\;\;\;\;\;\;\;\;\;\;\;\;\;\;\;\;\;\;\;\;
\;\;\;\;\;\;\;\;\;\;\;\;\;\;\;
\;\,
D\big({T\!\mathstrut}_{R}\circ {V\!\!\mathstrut}_{R} \; \| \; Q \circ P\big)
\, + \,
\rho\underbrace{\big[D\big({T\!\mathstrut}_{R}\circ {V\!\!\mathstrut}_{R} \; \| \; T \times Q\big) \, - \, R\big]}_{\geq\,0}
\nonumber \\
&
\geq
\;\;\;\;\;\;\;\;\;\;\;\;\;\;\;\,
\min_{\substack{\\T(y), \, V(x\,|\,y)
}}
\;\;\;\;\;\;\;\;\;\;\;\;\;\;\,
\Big\{D(T\circ V \, \| \, Q \circ P)
\, + \,
\rho\,\big[D(T \circ V \, \| \, T \times Q) \, - \, R\big]
\Big\}
\nonumber \\
& =
\;\;\;\;\;\;\;\;\;\;\;\;\;\;\;\;\;\;\;\;\;\;\;\;\;\;
\;\;\;\;\;\;\;\;\;\;\;\;\;\;\;\;\;\;\;\,\,
D\big({T\!\mathstrut}_{\rho}\circ {V\!\!\mathstrut}_{\rho} \, \| \, Q \circ P\big)
\, + \,
\rho \,\big[D\big({T\!\mathstrut}_{\rho} \circ {V\!\!\mathstrut}_{\rho} \, \| \, {T\!\mathstrut}_{\rho} \times Q\big) \, - \, R\big]
\nonumber \\
& \geq \;
\;\;\;\;\;\;\;\;\;\;\;\;\;\;\;\;\;\;\;\;\;\;\;\;\;\;\;\;\;\;\;\;\;\;
\;\;\;\;\;\;\;\;\;\;\;\;\;\;\;\;\;\;\;\;\;\;\;\;\;\;\;\;\;\;\;\;\;\;
-
(1-\rho)\big[D\big({T\!\mathstrut}_{\rho} \circ {V\!\!\mathstrut}_{\rho} \, \| \, {T\!\mathstrut}_{\rho} \times Q\big) \, - \, R\big]
\, +
\nonumber \\
& \;\;\;\;\;\,
\min_{\substack{\\T(y), \, V(x\,|\,y):\\ \\ D(T \,\circ \,V \, \| \, T \,\times\, Q) \;\geq\;
D({T\!\mathstrut}_{\rho} \,\circ \,{V\!\!\mathstrut}_{\rho} \, \| \, {T\!\mathstrut}_{\rho} \,\times\, Q)}}
\Big\{
D(T\circ V \, \| \, Q \circ P)
\, + \,
D(T \circ V \, \| \, T \times Q) \, - \, R\Big\}
\nonumber
\end{align}
The argument is similar to the proof of Lemma~\ref{lemLeq}.
\end{proof}

\bigskip

Lemmas~\ref{lemLeq} and~\ref{lemGeq} can be combined into a single lemma:

\bigskip

\begin{lemma} [Supporting lines] \label{lemCombined}\newline
{\em For any $\rho \in [0, 1]$ the minimum (\ref{eqDefImplicit}) is lower-bounded as}
\begin{equation} \label{eqCombined}
{E\mathstrut}_{\!e}(R, Q)
\;\; \geq \;\;
\min_{\substack{\\T(y), \, V(x\,|\,y)}}
\Big\{
D(T\circ V \, \| \, Q \circ P)
\,
+ \,
\rho \big[D(T \circ V \, \| \, T \times Q) \,-\, R\big]
\Big\}.
\end{equation}
{\em In the case of equality in (\ref{eqCombined}), any
minimizing solution of the LHS
is also a minimizing solution of the RHS
$\,{T\!\mathstrut}_{\rho} \circ {V\!\!\mathstrut}_{\rho}$
such that
$R = D\big({T\!\mathstrut}_{\rho} \circ {V\!\!\mathstrut}_{\rho} \, \| \, {T\!\mathstrut}_{\rho} \times Q\big)$,
or $R \leq D\big({T\mathstrut}_{1} \circ {V\!\mathstrut}_{1} \, \| \, {T\mathstrut}_{1} \times Q\big)$ for $\rho = 1$,
or $R \geq D\big({T\mathstrut}_{0} \circ {V\!\mathstrut}_{0} \, \| \, {T\mathstrut}_{0} \times Q\big)$ for $\rho = 0$.
Conversely,
if there exists such solution ${T\!\mathstrut}_{\rho} \circ {V\!\!\mathstrut}_{\rho}$
minimizing the RHS,
then it is also a minimizing solution of the LHS
and there is equality in (\ref{eqCombined}).}
\end{lemma}
\begin{proof}
Follows since ${E\mathstrut}_{\!e}(R, Q)$ is the minimum between the LHS in (\ref{eqPartLeq}) and (\ref{eqPartGeq}).
\end{proof}

\bigskip

The minimum on the RHS of (\ref{eqCombined}) 
can be explicitly solved:

\bigskip

\begin{lemma} [Explicit solution] \label{lemExplicit}\newline
{\em For $\rho > -1$}
\begin{equation} \label{eqExplicitSolution}
\min_{\substack{\\T(y), \, V(x\,|\,y)}}
\Big\{
D(T\circ V \, \| \, Q \circ P)
\,
+ \,
\rho D(T \circ V \, \| \, T \times Q)
\Big\}
\;\; = \;\;
{E\mathstrut}_{0}(\rho, Q),
\end{equation}
{\em where}
\begin{equation} \label{eqE0}
{E\mathstrut}_{0}(\rho, Q)
\; \triangleq \;
-\log \sum_{y}\bigg[\sum_{x}Q(x)P^{\frac{1}{1\,+\,\rho}}(y\,|\,x)\bigg]^{1\,+\,\rho},
\end{equation}
{\em and the unique minimizing distribution of (\ref{eqExplicitSolution}) is given by ${T\!\mathstrut}_{\rho}\circ {V\!\!\mathstrut}_{\rho}$ with}
\begin{align}
{T\!\mathstrut}_{\rho}(y) \; & \propto \; \bigg[\sum_{a}Q(a)P^{\frac{1}{1\,+\,\rho}}(y\,|\,a)\bigg]^{1\,+\,\rho},
\label{eqTrho} \\
{V\!\!\mathstrut}_{\rho}(x \,|\,y) \; & \propto \;
\;\;\;\;\;\;\;\;
Q(x)P^{\frac{1}{1\,+\,\rho}}(y\,|\,x).
\label{eqVrho}
\end{align}
\end{lemma}
\begin{proof}
\begin{align}
& \;\;
\min_{\substack{\\T(y), \, V(x\,|\,y)}}
\Big\{
D(T\circ V \, \| \, Q \circ P)
\;\,
+ \;\,
\rho D(T \circ V \, \| \, T \times Q)
\Big\}
\nonumber \\
= &
\;\;
\min_{\substack{\\T(y), \, V(x\,|\,y)}}
\Big\{
\underbrace{D(T \, \| \, {T\!\mathstrut}_{\rho})}_{\geq\,0}
\; + \;
(1 + \rho)\underbrace{D(T \circ V \, \| \, T \circ {V\!\!\mathstrut}_{\rho})}_{\geq\,0} \Big\}
\; + \; {E\mathstrut}_{0}(\rho, Q)
\label{eqDivergences} \\
= & \;\;\;\;
{E\mathstrut}_{0}(\rho, Q),
\nonumber
\end{align}
where the unique joint distribution minimizing the divergences in (\ref{eqDivergences}) is ${T\!\mathstrut}_{\rho} \circ {V\!\!\mathstrut}_{\rho}$ given by (\ref{eqTrho})-(\ref{eqVrho}).
\end{proof}

\bigskip

Lemmas~\ref{lemCombined} and~\ref{lemExplicit} result in the explicit expression for (\ref{eqDefImplicit}):

\bigskip

\begin{thm}[Explicit formula]\label{thm2}
\begin{align}
{E\mathstrut}_{\!e}(R, Q)
\;\; & \equiv \;\;
\max_{\substack{\\0 \, \leq \, \rho \, \leq \, 1}}
\;\;
\big\{
{E\mathstrut}_{0}(\rho, Q) - \rho R
\big\},
\label{eqExplicit}
\end{align}
{\em where ${E\mathstrut}_{0}(\rho, Q)$ is defined in (\ref{eqE0}) and the unique minimizing distribution of (\ref{eqDefImplicit}) is given by (\ref{eqTrho})-(\ref{eqVrho})
for some 
$\rho \in [0, 1]$ maximizing (\ref{eqExplicit}).}
\end{thm}

\bigskip

\begin{proof}
For $\rho \in [0, 1]$ and $\,R = D\big({T\!\mathstrut}_{\rho} \circ {V\!\!\mathstrut}_{\rho} \, \| \, {T\!\mathstrut}_{\rho} \times Q\big)\equiv \frac{\partial {E\mathstrut}_{0}(\rho, \, Q)}{\partial \rho}\,$
Lemma~\ref{lemCombined} gives equality in (\ref{eqCombined})
with the same unique distribution ${T\!\mathstrut}_{\rho} \circ {V\!\!\mathstrut}_{\rho}$
given by Lemma~\ref{lemExplicit}
minimizing both sides.
Observe that two different slope parameters $0 \leq \alpha < \beta \leq 1$
of two lower supporting lines from Lemma~\ref{lemCombined}
necessarily satisfy
$\,D\big({T\!\mathstrut}_{\beta} \circ {V\!\!\mathstrut}_{\beta} \, \| \, {T\!\mathstrut}_{\beta} \times Q\big)
\leq
D\big({T\!\mathstrut}_{\alpha} \circ {V\!\!\mathstrut}_{\alpha} \, \| \, {T\!\mathstrut}_{\alpha} \times Q\big)$.
Since $\frac{\partial {E\mathstrut}_{0}(\rho, \, Q)}{\partial \rho}$ is a continuous function of $\rho$,
this covers all $R$ such that
\begin{displaymath}
D\big({T\mathstrut}_{1} \circ {V\!\mathstrut}_{1} \, \| \, {T\mathstrut}_{1} \times Q\big)
\; \leq \; R \; \leq \; D\big({T\mathstrut}_{0} \circ {V\!\mathstrut}_{0} \, \| \, {T\mathstrut}_{0} \times Q\big).
\end{displaymath}

For $\rho = 1$ and $R \leq D\big({T\mathstrut}_{1} \circ {V\!\mathstrut}_{1} \, \| \, {T\mathstrut}_{1} \times Q\big) \equiv \frac{\partial {E\mathstrut}_{0}(\rho, \, Q)}{\partial \rho}\Big|_{\rho\,=\,1}\,$ Lemma~\ref{lemCombined} gives equality in (\ref{eqCombined})
with the same unique minimizing distribution ${T\!\mathstrut}_{1} \circ {V\!\!\mathstrut}_{1}$ on both sides.

Likewise for $\rho = 0$ and $R \geq D\big({T\mathstrut}_{0} \circ {V\!\mathstrut}_{0} \, \| \, {T\mathstrut}_{0} \times Q\big) \equiv \frac{\partial {E\mathstrut}_{0}(\rho, \, Q)}{\partial \rho}\Big|_{\rho\,=\,0}\,$ Lemma~\ref{lemCombined} gives equality in (\ref{eqCombined})
with the unique minimizing distribution ${T\!\mathstrut}_{0} \circ {V\!\!\mathstrut}_{0}\equiv Q\circ P$ on both sides.
\end{proof}

\bigskip

It can be noticed from the solution of Theorem~\ref{thm2} and Lemmas~\ref{lemLeq} and~\ref{lemGeq},
that the minima on the LHS of (\ref{eqPartLeq}) and (\ref{eqPartGeq}) coincide for
\begin{displaymath}
\frac{\partial {E\mathstrut}_{0}(\rho, \, Q)}{\partial \rho}\Big|_{\rho\,=\,1} \; \leq \; R \; \leq \; \frac{\partial {E\mathstrut}_{0}(\rho, \, Q)}{\partial \rho}\Big|_{\rho\,=\,0}.
\end{displaymath}
For $R < \frac{\partial{E\mathstrut}_{0}(\rho, \, Q)}{\partial \rho}\Big|_{\rho\,=\,1}$ the minimum on the LHS of (\ref{eqPartGeq}) in Lemma~\ref{lemGeq}, which is for the part\newline $\,D(T\circ V \,\|\, T \times Q)\geq R$, is lower than the minimum on the LHS of (\ref{eqPartLeq}) in Lemma~\ref{lemLeq}.
In this case 
the rare conditional type
${V\!\!\mathstrut}_{1}$ satisfying $\,D\big({T\mathstrut}_{1} \circ {V\!\mathstrut}_{1} \, \| \, {T\mathstrut}_{1} \times Q\big)> R\,$ is responsible for the error event.

Theorem~\ref{thm2} shows that the metric $\,\log \frac{{V}_{m}(x \, | \, y)}{Q(x)}\,$ used by the decoder in place of the ML metric $\,\log P(y \, | \, x)\,$ produces the same optimal i.i.d. random coding error exponent (\ref{eqExplicit}) \cite{Gallager73}.
The metric $\,\log \frac{{V}_{m}(x \, | \, y)}{Q(x)}\,$ can be called natural for random codes, because
it is exponentially optimal and its average (\ref{eqMMIlikeDec}) itself
is the exponent of a conditional type
which is meaningfully compared to the rate $R$, as in (\ref{eqDefImplicit}).
As a result, a linear {\em split} of the metric in two by subtraction of a constant $\Delta > 0$ causes a linear shift by this $\Delta$ in the exponent as a function of $R$.

\section{Correct-decoding exponent}\label{CorDecExp}


\bigskip

Next we examine the exponential optimality of the decoder (\ref{eqMMIlikeDec}) when the correct-decoding event is exponentially rare.
We use the same derivation methods as in the preceding section. The results will allow us to formulate the computation algorithm in Section~\ref{CordecExpComp}
and the stochastic adaptation scheme in Section~\ref{ChInAd}.

\bigskip

\subsection{Implicit expression}\label{CorDecExpA}

\bigskip

Here we continue to assume the decoding procedure according to (\ref{eqDecRule}), using two generally different functions/metrics.
The decoding according to the maximum of (\ref{eqMMIlikeDec}) is a special case of (\ref{eqDecRule}).
Let ${P\mathstrut}_{\!c}$
denote the i.i.d. random code ensemble average probability of the correct-decoding event: when the {\em true} message $m$ satisfies (\ref{eqDecRule}).
In the case of the decoder (\ref{eqMMIlikeDec}) this is the {\em strict correct decoding} event excluding the possible tie-breaking success.
Or, alternatively, assume that in case of a tie between different messages the decoder (\ref{eqMMIlikeDec}) declares an error and doesn't try to guess the true message.
The next lemma and theorem give the exponent of the correct-decoding event, described by (\ref{eqDecRule}).

\bigskip

\begin{lemma}\label{lemEps}\newline
{\em Let $B(T\circ V) \equiv D(T\circ V \,\|\, T \times Q)$ and $A(T\circ V)\leq D(T\circ V \,\|\, T \times Q)$. Then
for any $\epsilon > 0$}
\begin{align}
\frac{\log{P}_{c}}{-n}
\;\;
\leq \;\;
& \min_{\substack{\\\text{types} \;\; T(y), \, V(x\,|\,y):\\
\\
A(T \,\circ \,V)\;\geq\;
R\,+\,\epsilon
}}
\Big\{D(T\circ V \, \| \, Q \circ P)\Big\} \; + \; o(1),
\label{eqUboundEps} \\
\liminf_{n\,\rightarrow\,\infty} \; \frac{\log{P}_{c}}{-n}
\;\;
\geq \;\;
&
\;\,
\min_{\substack{\\ T(y), \, V(x\,|\,y):\\
\\
A(T \,\circ \,V)\;\geq\;
R\,-\,\epsilon
}}
\,\,
\Big\{D(T\circ V \, \| \, Q \circ P)\Big\},
\label{eqLboundEps}
\end{align}
{\em where the minimization in (\ref{eqUboundEps}) is over types, corresponding to the blocklength $n$, and $o(1)$ depends on $\epsilon$.
Given that the joint type 
of the transmitted and the received words 
satisfies $\,A(T\circ V) \,\leq\, R - \epsilon\,$
the exponent of
the correct-decoding event is $+\infty$.}
\end{lemma}
The proof is given in Appendix.

\bigskip

\begin{thm}[Correct-decoding exponent of the natural decoder]\label{thmCorr}\newline
{\em Let $B(T\circ V) \equiv D(T\circ V \,\|\, T \times Q)$ and let $A(T\circ V)\leq D(T\circ V \,\|\, T \times Q)$ be a continuous function of $T\circ V$ in the support of $Q\circ P$.
Then}
\begin{equation} \label{eqCorrAB}
\lim_{n\,\rightarrow \,\infty}\,\frac{\log{P}_{c}}{-n}
\;\; = \;\;
\min_{\substack{\\T(y), \, V(x\,|\,y):\\
\\
A(T \,\circ \,V)\;\geq\;
R
}}
\Big\{D(T\circ V \, \| \, Q \circ P)\Big\}
\;\; \triangleq \;\;
{E\mathstrut}_{\!c}^{A}(R, Q),
\end{equation}
{\em 
provided that the RHS is continuous at $R$.}
\end{thm}

\bigskip

\begin{proof}
According to Lemma~\ref{lemEps} the $\,\liminf\,$ of the exponent is lower-bounded by
$\,\lim_{\,\epsilon\,\rightarrow\,0}{E\mathstrut}_{\!c}^{A}(R - \epsilon,\, Q)$.
For the upper bound, for any $\epsilon > 0$ consider the minimizing solution ${T\mathstrut}^{*}\circ {V\mathstrut}^{*}$ of $\,{E\mathstrut}_{\!c}^{A}(R + 2\epsilon,\, Q)$. Let ${T\mathstrut}_{n}^{*}\circ {V\mathstrut}_{n}^{*}$ denote a quantized version of ${T\mathstrut}^{*}\circ {V\mathstrut}^{*}$
with precision $\frac{1}{n}$. Then ${T\mathstrut}_{n}^{*}\circ {V\mathstrut}_{n}^{*}$ is a joint type with denominator $n$.
By continuity, for sufficiently large $n$ we obtain $\,A\big({T\mathstrut}_{n}^{*} \circ {V\mathstrut}_{n}^{*}\big) \geq R + \epsilon\,$
and also $\,D\big({T\mathstrut}_{n}^{*} \,\circ\, {V\mathstrut}_{n}^{*} \; \| \; Q \,\circ\, P\big) \,\leq\, {E\mathstrut}_{\!c}^{A}(R + 2\epsilon,\, Q) + \epsilon$.
Comparing to the upper bound (\ref{eqUboundEps}) of Lemma~\ref{lemEps} we conclude that the $\,\limsup\,$ of the exponent of the correct-decoding event is upper-bounded by $\,\lim_{\,\epsilon\,\rightarrow\,0}{E\mathstrut}_{\!c}^{A}(R + 2\epsilon,\, Q)$.
Provided that the RHS of (\ref{eqCorrAB}) is continuous at $R$,
the bounds coincide and give
$\,\lim_{\,\epsilon\,\rightarrow\,0}{E\mathstrut}_{\!c}^{A}(R - \epsilon,\, Q) = \lim_{\,\epsilon\,\rightarrow\,0}{E\mathstrut}_{\!c}^{A}(R + 2\epsilon,\, Q) = {E\mathstrut}_{\!c}^{A}(R,\, Q)$.
\end{proof}

\bigskip

We compare the strict correct decoding exponent of the decoder (\ref{eqMMIlikeDec}), given by Theorem~\ref{thmCorr}
for the choice $A(T\circ V) \equiv B(T\circ V) \equiv D(T\circ V \,\|\, T \times Q)\,$:
\begin{align}
{E\mathstrut}_{\!c}(R, Q)
\; & \triangleq \;
\min_{\substack{\\T(y), \, V(x\,|\,y):\\
\\
D(T\,\circ\, V \,\|\, T\, \times\, Q)\;\geq\;
R
}}
\Big\{D(T\circ V \, \| \, Q \circ P)\Big\},
\label{eqExpNoTB}
\end{align}
and the correct-decoding exponent of the ML decoder (\ref{eqML}) with tie breaking\footnote{Note that the correct-decoding exponent of the decoder (\ref{eqMMIlikeDec})
{\em with} tie breaking has to be\newline ${E\mathstrut}_{\!c}(R, Q)$ for $R \leq {R\mathstrut}_{\max}$ and
$\,{E\mathstrut}_{\!c}({R\mathstrut}_{\max}, \,Q) + R - {R\mathstrut}_{\max}\,$ for $\,R > {R\mathstrut}_{\max}\,$,
where $\,{R\mathstrut}_{\max} \,\triangleq\, \max_{\,{E\mathstrut}_{\!c}(R, Q) \, < \, +\infty}\,\{R\,\}$.\newline
On the other hand, the correct decoding exponent of the decoder (\ref{eqML}) {\em without} tie breaking is \cite[eq.~24]{TridenskiZamir17}\newline
$\displaystyle\text{}\;\;\;\;\;\;\;\;\;\;\;\;\;
\;\;\;\;\;\;\;\;\;\;\;\;\;\;\;\;\;\;\;\;\;
\min_{\substack{\\T(y),\,V(x \,|\, y):\;\; R(T\,\circ \,V, \,Q, \,0)\;\geq\;R}}
\;\;\;\;\;\;\;\;
\Big\{D(T\circ V \; \| \; Q \circ P) \Big\}, \;\;\, \text{where}$ $R(T\circ V, \,Q, \,0)$ is itself a minimum:
\newline $\displaystyle R(T\circ V, \,Q, \,0) \; = \;\min_{\substack{\\\widehat{V}(\widehat{x}\,|\,y):
\\
\sum_{x, \, y}T(y) V(x\,|\,y)\widehat{V}(\widehat{x}\,|\, y)\log \frac{P(y\,|\,x)}{P(y\,|\,\widehat{x})}
\;\leq\; 0
\\ \\}}
\;\;\,
\Big\{D\big(T \circ \widehat{V} \; \| \; T \times Q\big) \Big\}$\newline
Comparing to (\ref{eqExpNoTB}), this exponent has $R(T\circ V, \,Q, \,0)$ in place of $D(T\circ V \,\|\, T \times Q)$ and can be higher than (\ref{eqExpNoTB}) for some $R$.} \cite[eq.~32]{TridenskiZamir17}:
\begin{align}
{E\mathstrut}_{\!c}^{ML}(R, Q)
\; & \triangleq \;
\;\;\;\;\;
\min_{\substack{\\T(y), \, V(x\,|\,y)
}}
\;\;\;\;\,
\left\{
D(T\circ V \, \| \, Q \circ P) \, + \,
\big| R - D(T\circ V \,\|\, T \times Q)\big|^{+}
\right\}.
\label{eqDefImplicitCor}
\end{align}
In what follows, first we derive the expression (\ref{eqDefImplicitCor}) as the i.i.d. ML correct-decoding exponent.
Afterwards we derive an explicit expression for (\ref{eqExpNoTB}) and (\ref{eqDefImplicitCor}).
Let ${P\mathstrut}_{\!c}^{ML}$
denote the i.i.d. random code ensemble average probability of the correct-decoding event in the ML decoding with tie breaking.
Then

\bigskip

\begin{thm}[Correct-decoding exponent of the ML decoder]\label{thmCorrML}
\begin{equation} \label{eqCorrML}
\lim_{n\,\rightarrow \,\infty}\,\frac{\log{P}_{c}^{ML}}{-n}
\;\; = \;\;
{E\mathstrut}_{\!c}^{ML}(R, Q).
\end{equation}
\end{thm}
The proof is given in Appendix.

\bigskip

The implicit expression (\ref{eqDefImplicitCor}) for the i.i.d. random code correct-decoding exponent of the ML decoder ${E\mathstrut}_{\!c}^{ML}(R, Q)$, can be easily compared to the constant-composition correct-decoding exponent:
\begin{align}
{E\mathstrut}_{\!c}^{ML}(R, Q) \;\;\; = \;\;\;\;
& \min_{\substack{\\T(y), \, V(x\,|\,y)
}}
\,
\Big\{
D(T\circ V \, \| \, Q \circ P) \; + \;
\big|R - \underbrace{D(T\circ V \,\|\, T \times Q)}_{= \; I(U\,\circ\, W) \, + \, D(U \,\|\, Q)}\big|^{+}
\Big\}
\nonumber \\
\equiv \;\;\;\; & \!\min_{\substack{\\U(x), \, W(y\,|\,x)
}}
\left\{
D(U\circ W \, \| \, Q \circ P) \, + \,
\big|R - I(U\circ W) - D(U \,\|\, Q)\big|^{+}
\right\}
\nonumber \\
\overset{U=\,Q}{\leq} \;\; &
\;\;\; \min_{\substack{\\W(y\,|\,x)
}}
\;\;\;
\left\{
D(Q\circ W \, \| \, Q \circ P) \, + \,
\big|R - I(Q\circ W)\big|^{+}
\right\}.
\label{eqCCCorDec}
\end{align}
After minimization over the codebook distribution $Q$, the expression (\ref{eqCCCorDec})
can be recognized as the constant-composition correct-decoding exponent \cite{DueckKorner79},
which is optimal.
Therefore the minimum of its achievable lower bound (\ref{eqDefImplicitCor}) over $Q$ is also optimal.
The proof of Theorem~\ref{thmCorrML} can be repeated for the constant composition codes,
in which case
the 
ML metric $\log P(y \, | \, x)$ with tie breaking leads to the exponent (\ref{eqCCCorDec}).

For completeness, the i.i.d. strict correct decoding exponent (\ref{eqExpNoTB}) of the decoder (\ref{eqMMIlikeDec}) can also be compared to its constant-composition counterpart:
\begin{align}
{E\mathstrut}_{\!c}(R, Q)
\;\;\; = \;\; & \;
\min_{\substack{\\T(y), \, V(x\,|\,y):\\
\\
D(T\,\circ\, V \,\|\, T\, \times\, Q)\;\geq\;
R
}}
\Big\{D(T\circ V \, \| \, Q \circ P)\Big\}
\nonumber \\
\leq \;\; & \;
\;\;\;\;
\min_{\substack{\\U(x), \, W(y\,|\,x):\\
\\
I(U\,\circ\, W)\;\geq\;
R
}}
\;\;\;\;
\Big\{D(U\circ W \, \| \, Q \circ P)\Big\}
\nonumber \\
\overset{U = \, Q}{\leq} & \;
\;\;\;\;\,
\min_{\substack{\\W(y\,|\,x):\\
\\
I(Q\,\circ\, W)\;\geq\;
R
}}
\;\;\;\,\,
\Big\{D(Q\circ W \, \| \, Q \circ P)\Big\}.
\label{eqCCCounterpart}
\end{align}
The proof of Theorem~\ref{thmCorr} can be repeated for the constant composition codes.
In that case the same metric
$\,\log \frac{{V}_{m}(x \, | \, y)}{Q(x)}\,$ from (\ref{eqMMIlikeDec}) becomes the pointwise mutual information, resulting in the maximum mutual information (MMI) decoder with the strict correct decoding exponent (\ref{eqCCCounterpart}).
As we shall see, the minimum of (\ref{eqExpNoTB}) over $Q$ achieves the optimum for all $R$. This is not the case with its constant-composition upper bound (\ref{eqCCCounterpart}). Since the mutual information $I(Q\,\circ\, W)$ is upper-bounded by the entropy of $Q$, for sufficiently large $R$
the minimum (\ref{eqCCCounterpart}) inevitably becomes $+\infty$.

\bigskip

\subsection{Explicit expression}\label{CorDecExpB}

\bigskip

We next proceed to obtain an explicit expression for 
(\ref{eqExpNoTB}) and (\ref{eqDefImplicitCor}).
The following parallels the corresponding lemmas for the error exponent case.

\bigskip

\begin{lemma} [For $\,D(T\circ V \,\|\, T \times Q)\geq R\,$ part] \label{lemGeqCD}\newline
{\em For any $\rho \leq 0$ the minimum (\ref{eqExpNoTB}) is lower-bounded as}
\begin{equation} \label{eqPartGeqCD}
\min_{\substack{\\T(y), \, V(x\,|\,y):\\
\\
D(T \,\circ \,V \, \| \, T \,\times\, Q)\;\geq\;
R
}}
D(T\circ V \, \| \, Q \circ P)
\;\; \geq \;\;
\min_{\substack{\\T(y), \, V(x\,|\,y):\\ \\\text{supp}(V)\;\subseteq\;\text{supp}(Q)}}
\Big\{
D(T\circ V \, \| \, Q \circ P)
\,
- \,
\rho \big[R \, - \, D(T \circ V \, \| \, T \times Q)\big]
\Big\}.
\end{equation}
{\em
In the case of equality in (\ref{eqPartGeqCD}), any
minimizing solution of the LHS
is also a minimizing solution of the RHS
$\,{T\!\mathstrut}_{\rho} \circ {V\!\!\mathstrut}_{\rho}$ such that
$R = D\big({T\!\mathstrut}_{\rho} \circ {V\!\!\mathstrut}_{\rho} \, \| \, {T\!\mathstrut}_{\rho} \times Q\big)$,
or $R \leq D\big({T\mathstrut}_{0} \circ {V\!\mathstrut}_{0} \, \| \, {T\mathstrut}_{0} \times Q\big)$ for $\rho = 0$.
Conversely,
if there exists such solution ${T\!\mathstrut}_{\rho} \circ {V\!\!\mathstrut}_{\rho}$
minimizing the RHS,
then it is also a minimizing solution of the LHS
and there is equality in (\ref{eqPartGeqCD}).}
\end{lemma}
\begin{proof}
Analogous to the proof of Lemma~\ref{lemLeq}.
\end{proof}

\bigskip

\begin{lemma} [For $\,D(T\circ V \,\|\, T \times Q)\leq R\,$ part] \label{lemLeqCD}\newline
{\em For any $\rho \geq -1$}
\begin{align}
&
\;\;\;\;
\min_{\substack{\\T(y), \, V(x\,|\,y):\\
\\
D(T \,\circ \,V \, \| \, T \,\times\, Q)\;\leq\;
R
}}
\Big\{D(T\circ V \, \| \, Q \circ P)
\, + \,
R \, - \,
D(T \circ V \, \| \, T \times Q)
\Big\}
\nonumber \\
&
\geq
\;\,
\min_{\substack{\\T(y), \, V(x\,|\,y):\\ \\\text{supp}(V)\;\subseteq\;\text{supp}(Q)
}}
\;\;
\Big\{D(T\circ V \, \| \, Q \circ P)
\, - \,
\rho\big[R \, - \, D(T \circ V \, \| \, T \times Q)\big]
\Big\}.
\label{eqPartLeqCD}
\end{align}
{\em
In the case of equality in (\ref{eqPartLeqCD}), any
minimizing solution of the LHS
is also a minimizing solution of the RHS
$\,{T\!\mathstrut}_{\rho} \circ {V\!\!\mathstrut}_{\rho}$
such that
$R = D\big({T\!\mathstrut}_{\rho} \circ {V\!\!\mathstrut}_{\rho} \, \| \, {T\!\mathstrut}_{\rho} \times Q\big)$,
or $R \geq D\big({T\!\mathstrut}_{-1} \circ {V\!\mathstrut}_{-1} \, \| \, {T\!\mathstrut}_{-1} \times Q\big)$ for $\rho = -1$.
Conversely,
if there exists such solution ${T\!\mathstrut}_{\rho} \circ {V\!\!\mathstrut}_{\rho}$
minimizing the RHS,
then it is also a minimizing solution of the LHS
and there is equality in (\ref{eqPartLeqCD}).}
\end{lemma}
\begin{proof}
Analogous to the proof of Lemma~\ref{lemGeq}.
\end{proof}

\bigskip

Now Lemmas~\ref{lemGeqCD} and~\ref{lemLeqCD} are combined to lower-bound (\ref{eqDefImplicitCor}):

\bigskip

\begin{lemma} [Supporting lines] \label{lemCombinedCD}\newline
{\em For any $\rho \in [-1, 0]$ the minimum (\ref{eqDefImplicitCor}) is lower-bounded as}
\begin{equation} \label{eqCombinedCD}
{E\mathstrut}_{\!c}^{ML}(R, Q)
\;\; \geq \;\;
\min_{\substack{\\T(y), \, V(x\,|\,y):\\ \\\text{supp}(V)\;\subseteq\;\text{supp}(Q)}}
\Big\{
D(T\circ V \, \| \, Q \circ P)
\,
- \,
\rho \big[R \, - \, D(T \circ V \, \| \, T \times Q)\big]
\Big\}.
\end{equation}
{\em
In the case of equality in (\ref{eqCombinedCD}), any
minimizing solution of the LHS
is also a minimizing solution of the RHS
$\,{T\!\mathstrut}_{\rho} \circ {V\!\!\mathstrut}_{\rho}$
such that
$R = D\big({T\!\mathstrut}_{\rho} \circ {V\!\!\mathstrut}_{\rho} \, \| \, {T\!\mathstrut}_{\rho} \times Q\big)$,
or $R \leq D\big({T\mathstrut}_{0} \circ {V\!\mathstrut}_{0} \, \| \, {T\mathstrut}_{0} \times Q\big)$ for $\rho = 0$,
or $R \geq D\big({T\!\mathstrut}_{-1} \circ {V\!\mathstrut}_{-1} \, \| \, {T\!\mathstrut}_{-1} \times Q\big)$ for $\rho = -1$.
Conversely,
if there exists such solution ${T\!\mathstrut}_{\rho} \circ {V\!\!\mathstrut}_{\rho}$
minimizing the RHS,
then it is also a minimizing solution of the LHS and there is equality in (\ref{eqCombinedCD}).}
\end{lemma}
\begin{proof}
By Lemmas~\ref{lemGeqCD} and~\ref{lemLeqCD} since ${E\mathstrut}_{\!c}^{ML}(R, Q)$ is the minimum between the LHS in (\ref{eqPartGeqCD}) and (\ref{eqPartLeqCD}).
\end{proof}

\bigskip

Explicit solution of the RHS of (\ref{eqCombinedCD}) gives

\bigskip

\begin{lemma} [Explicit solution] \label{lemExplicitSolutionCD}\newline
{\em For\footnote{A solution can be obtained also for $\rho < -1$ by maximization of a convex ($\cup$) function.} $\rho \geq -1$}
\begin{equation} \label{eqExplicitSolCD}
\min_{\substack{\\T(y), \, V(x\,|\,y):\\ \\\text{supp}(V)\;\subseteq\;\text{supp}(Q)}}
\Big\{
D(T\circ V \, \| \, Q \circ P)
\,
+ \,
\rho D(T \circ V \, \| \, T \times Q)
\Big\}
\;\; = \;\;
{E\mathstrut}_{0}(\rho, Q),
\end{equation}
{\em where ${E\mathstrut}_{0}(\rho, Q)$ is given by (\ref{eqE0})
for $\rho > -1$ and}
\begin{equation} \label{eqE0Minus1}
{E\mathstrut}_{0}(-1, Q) \; \triangleq \;
\lim_{\substack{\\\rho \; \searrow \; -1}}{E\mathstrut}_{0}(\rho, Q)
\; = \;
-\log \,\sum_{y} \,\max_{\substack{\\ \\x: \; Q(x)\,>\, 0}}P(y\,|\,x).
\end{equation}
{\em If $\rho > -1$,
then the unique minimizing solution of (\ref{eqExplicitSolCD}) is given by (\ref{eqTrho})-(\ref{eqVrho}).\newline
If $\rho = - 1$,
then all minimizing solutions of (\ref{eqExplicitSolCD}) are any ${T\!\mathstrut}_{-1}\circ {V\!\mathstrut}_{-1}$ such that}
\begin{align}
{T\!\mathstrut}_{-1}(y) \;\; & \propto \;\;
\;\;\;\;\;\;\;\;\;\;\;\;\;\;\;\;\;\;\;\;\;\;\;\;\,
\max_{\substack{\\a: \; Q(a)\,>\, 0}}P(y\,|\,a),
\label{eqTminus1} \\
{V\!\mathstrut}_{-1}(x \,|\,y) \;\; & = \;\;
0,
\;\;\;\;\;
\forall \; x \, \notin \, {\displaystyle \arg\max_{\substack{\\a: \; Q(a)\,>\, 0}}P(y\,|\,a)}.
\label{eqVminus1}
\end{align}
\end{lemma}
\begin{proof}
The case $\rho > -1$ follows by Lemma~\ref{lemExplicit}.
For $\rho = -1\,$:
\begin{align}
& \;\;
\min_{\substack{\\T(y), \, V(x\,|\,y):\\ \\\text{supp}(V)\;\subseteq\;\text{supp}(Q)}}
\Big\{
D(T\circ V \, \| \, Q \circ P)
\,
- \,
D(T \circ V \, \| \, T \times Q)
\Big\}
\nonumber \\
= &
\;\;
\min_{\substack{\\T(y), \, V(x\,|\,y):\\ \\\text{supp}(V)\;\subseteq\;\text{supp}(Q)}}
\bigg\{
\sum_{x,\,y}T(y)V(x\, | \,y)\log\frac{T(y)}{P(y\, | \, x)} \bigg\}
\, + \, R
\nonumber \\
= &
\;\;\;\;\;\;\;\;\,\,
\min_{\substack{\\T(y)}}
\;\;\;\;\;\;\;\,
\Big\{
\underbrace{D(T \, \| \, {T\!\mathstrut}_{-1})}_{\geq\,0} \Big\}
\; + \; {E\mathstrut}_{0}(-1, Q) \, + \, R
\nonumber \\
= & \;\;\;\;
{E\mathstrut}_{0}(-1, Q) \; + \; R,
\nonumber
\end{align}
where the minimum is achieved with any $V(x\, | \,y)$ satisfying (\ref{eqVminus1}) and the unique ${T\!\mathstrut}_{-1}(y)$
given by (\ref{eqTminus1}).
\end{proof}

\bigskip

Since the minimizing distribution ${V\!\mathstrut}_{-1}(x \,|\,y)$ of (\ref{eqExplicitSolCD}) with $\rho = -1$
is allowed to be arbitrary inside its support which is restricted according to (\ref{eqVminus1}),
for convenience let us define
\begin{align}
{R\mathstrut}_{\,-1}^{\,-}(Q) \;\; & \triangleq \;\; \min_{\substack{\\{\;\;\;V\!\mathstrut}_{-1}(x\,|\,y)}}\, D\big({T\!\mathstrut}_{-1} \circ {V\!\mathstrut}_{-1} \, \| \, {T\!\mathstrut}_{-1} \times Q\big),
\label{eqRMinus} \\
{R\mathstrut}_{\,-1}^{\,+}(Q) \;\; & \triangleq \;\; \max_{\substack{\\{\;\;\;V\!\mathstrut}_{-1}(x\,|\,y)}}\, D\big({T\!\mathstrut}_{-1} \circ {V\!\mathstrut}_{-1} \, \| \, {T\!\mathstrut}_{-1} \times Q\big),
\label{eqRPlus}
\end{align}
where the $\min$ and $\max$ are over ${V\!\mathstrut}_{-1}(x \,|\,y)$ as in (\ref{eqVminus1}).

Lemmas~\ref{lemGeqCD}, ~\ref{lemCombinedCD}, and~\ref{lemExplicitSolutionCD} result in the explicit expression for (\ref{eqExpNoTB}) and (\ref{eqDefImplicitCor}):

\bigskip

\begin{thm}[Explicit formula]\label{thm3}
\begin{align}
{E\mathstrut}_{\!c}(R, Q)
\;\; & = \;\;
\max_{\substack{\\\!\!\!-1 \, \leq \, \rho \, \leq \, 0}}
\;\;
\big\{
{E\mathstrut}_{0}(\rho, Q) - \rho R
\big\},
\;\;\;\; \forall \; R \; \leq \; {R\mathstrut}_{\,-1}^{\,+}(Q),
\label{eqExplicitSCD} \\
{E\mathstrut}_{\!c}^{ML}(R, Q)
\;\; & \equiv \;\;
\max_{\substack{\\\!\!\!-1 \, \leq \, \rho \, \leq \, 0}}
\;\;
\big\{
{E\mathstrut}_{0}(\rho, Q) - \rho R
\big\},
\label{eqExplicitCD}
\end{align}
{\em where ${E\mathstrut}_{0}(\rho, Q)$ is defined by (\ref{eqE0}) and (\ref{eqE0Minus1}),
${R\mathstrut}_{\,-1}^{\,+}(Q)$
defined in (\ref{eqRPlus}),
and ${T\!\mathstrut}_{-1}\circ {V\!\mathstrut}_{-1}$ is a family of distributions defined by (\ref{eqTminus1})-(\ref{eqVminus1}).\newline
If the RHS of (\ref{eqExplicitCD}) is maximized by $\rho = - 1$,
then all minimizing solutions of the LHS 
are given by ${T\!\mathstrut}_{-1}\circ {V\!\mathstrut}_{-1}$ as in (\ref{eqTminus1})-(\ref{eqVminus1}) such that $D\big({T\!\mathstrut}_{-1} \circ {V\!\mathstrut}_{-1} \, \| \, {T\!\mathstrut}_{-1} \times Q\big)\leq R$.\newline
If the RHS of (\ref{eqExplicitSCD}) is maximized by $\rho = - 1$,
then all minimizing solutions of the LHS 
are given by ${T\!\mathstrut}_{-1}\circ {V\!\mathstrut}_{-1}$ as in (\ref{eqTminus1})-(\ref{eqVminus1}) such that $D\big({T\!\mathstrut}_{-1} \circ {V\!\mathstrut}_{-1} \, \| \, {T\!\mathstrut}_{-1} \times Q\big)= R$.\newline
If the RHS of (\ref{eqExplicitCD}) is maximized by $\rho \in (-1, 0]$,
then the unique minimizing solution of the LHS 
is given by (\ref{eqTrho})-(\ref{eqVrho}).
Same for (\ref{eqExplicitSCD}).}   
\end{thm}

\bigskip

\begin{proof}
For $\rho = -1$ and $R \,\geq\, {R\mathstrut}_{\,-1}^{\,-}(Q)\,$
Lemma~\ref{lemCombinedCD} gives equality in (\ref{eqCombinedCD}) with all possible solutions of the minimum on the LHS 
of (\ref{eqCombinedCD}) as given by Lemma~\ref{lemExplicitSolutionCD}
in (\ref{eqTminus1})-(\ref{eqVminus1}) and such that $D\big({T\!\mathstrut}_{-1} \circ {V\!\mathstrut}_{-1} \, \| \, {T\!\mathstrut}_{-1} \times Q\big)\leq R$.
Likewise, Lemma~\ref{lemGeqCD} gives equality in (\ref{eqPartGeqCD}) with all possible solutions of the minimum on the LHS 
of (\ref{eqPartGeqCD}) as given by Lemma~\ref{lemExplicitSolutionCD}
in (\ref{eqTminus1})-(\ref{eqVminus1}) and such that $D\big({T\!\mathstrut}_{-1} \circ {V\!\mathstrut}_{-1} \, \| \, {T\!\mathstrut}_{-1} \times Q\big)= R$.
The minimum $\,{R\mathstrut}_{\,-1}^{\,-}(Q)\,$
is achieved by
\begin{align}
{V\!\mathstrut}_{-1}(x \,|\,y)  \;\; & \propto \;\;
\left\{
\begin{array}{l l}
Q(x), &  \;\;\;\;\;\; x \, \in \, {\displaystyle \arg\max_{\substack{\\a: \; Q(a)\,>\, 0}}P(y\,|\,a)}, \\
0, & \;\;\;\;\;\; \text{else}.
\end{array}
\right.
\nonumber
\end{align}
This allows for the continuity $\,\lim_{\,\rho \, \searrow \, -1}
D\big({T\!\mathstrut}_{\rho} \circ {V\!\!\mathstrut}_{\rho} \, \| \, {T\!\mathstrut}_{\rho} \times Q\big) \,=\, {R\mathstrut}_{\,-1}^{\,-}(Q)$.
The rest of the proof for $R \,\leq\, {R\mathstrut}_{\,-1}^{\,-}(Q)\,$
also follows by Lemmas~\ref{lemGeqCD}, ~\ref{lemCombinedCD}, and~\ref{lemExplicitSolutionCD} and
is similar to the proof of Theorem~\ref{thm2}.
\end{proof}

\bigskip

It follows from Theorem~\ref{thm3}
that the minima (\ref{eqDefImplicitCor}) and (\ref{eqExpNoTB}) (the latter appears also on the LHS of (\ref{eqPartGeqCD}) in Lemma~\ref{lemGeqCD})
coincide for
$
R \, \leq \, {R\mathstrut}_{\,-1}^{\,+}(Q).
$
For greater $R$ the ML exponent (\ref{eqDefImplicitCor}) continues to grow with the increase of $R$ with constant slope $1$, according to (\ref{eqExplicitCD}).
For such $R$ the achieving types ${V\!\mathstrut}_{-1}$ with high probability are found in the codebook.
Both (\ref{eqExpNoTB}) and (\ref{eqDefImplicitCor}) have a supporting line $\,E(R) = {E\mathstrut}_{0}(-1, Q) \, + \, R\,$ of slope $1$.
This supporting line
is invariant in a sense that it depends only on the support of the distribution $Q(x)$ according to the expression for ${E\mathstrut}_{0}(-1, Q)$ in (\ref{eqE0Minus1}).
Since the rate $R = D\big({T\!\mathstrut}_{-1} \circ {V\!\mathstrut}_{-1} \, \| \, {T\!\mathstrut}_{-1} \times Q\big)$ where both exponents meet this supporting line can be made arbitrarily large by reducing $Q(x)$ on some letter $x$,
the minimum of (\ref{eqExpNoTB}) can always achieve the minimum of (\ref{eqDefImplicitCor}):
\begin{equation} \label{eqOptimum}
\min_{\substack{\\Q}}\, {E\mathstrut}_{\!c}(R, Q) \;\; = \;\;
\min_{\substack{\\Q}}\, {E\mathstrut}_{\!c}^{ML}(R, Q), \;\;\;\;\;\; \forall \; R.
\end{equation}
Therefore both exponents (\ref{eqExpNoTB}) and (\ref{eqDefImplicitCor}) achieve the optimum.
For comparison, this is not possible with (\ref{eqCCCounterpart}) where the mutual information is bounded by the entropy in the support of $Q$.

\bigskip

\section{Iterative minimization of the correct-decoding exponent}\label{CordecExpComp}

\bigskip

We propose a procedure of iterative minimization with (\ref{eqDefImplicitCor}) or (\ref{eqExpNoTB}) at fixed rate $R$ which leads to (\ref{eqOptimum}).
This can be termed also as a fixed-rate computation of the correct-decoding exponent and is different than the
algorithm of Arimoto \cite{Arimoto76} for computation of the exponent function $\min_{\,Q}{E\mathstrut}_{0}(\rho, Q)$.
The difference is both in the fact that the Arimoto algorithm is a fixed-slope computation but also the computation itself is different.
The advantage of the fixed-rate computation over the fixed-slope is that we know how to translate it to a stochastic procedure.

The next lemma presents and characterizes the iterative minimization procedure for the ML exponent (\ref{eqDefImplicitCor}).

\bigskip

\begin{lemma} [Monotonicity for ML] \label{lemMonotML}\newline
{\em An iterative update of the parameter $Q$ in (\ref{eqDefImplicitCor}) by its minimizing solution $\breve{T\mathstrut}(y)\breve{V\mathstrut}(x\,|\,y)$:}
\begin{equation} \label{eqUpdate}
{Q\mathstrut}_{\ell\,+\,1}(x) \;\; \leftarrow \;\; \sum_{y}\breve{T\mathstrut}_{\ell}(y)\breve{V\mathstrut}_{\!\ell}(x\,|\,y)
\end{equation}
{\em results in a monotonically non-increasing sequence $\,{\left\{{E\mathstrut}_{\!c}^{ML}(R, {Q\mathstrut}_{\ell})\right\}\mathstrut}_{\ell \, = \, 0}^{+\infty}\,$ of (\ref{eqDefImplicitCor}).\newline
At each step, the sequence decreases at least by an amount  $\,(1 \, + \, {\hat{\rho}\mathstrut}_{\ell\,+\,1})D({Q\mathstrut}_{\ell\,+\,1} \, \| \, {Q\mathstrut}_{\ell})$, where $\,{\hat{\rho}\mathstrut}_{\ell\,+\,1} \in [-1, 0]\,$ is a parameter of some supporting line (\ref{eqCombinedCD}) touching
the graph of
$\,{E\mathstrut}_{\!c}^{ML}(R, {Q\mathstrut}_{\ell\,+\,1})\,$ at $R$.}
\end{lemma}

\bigskip

\begin{proof}
By Theorem~\ref{thm3}~/~Lemma~\ref{lemExplicitSolutionCD} the graph of ${E\mathstrut}_{\!c}^{ML}(R, Q)$ touches at $R$ some supporting line of the form (\ref{eqCombinedCD})
with some slope parameter $\hat{\rho} \in [-1, 0]$.
The solution $\breve{T\mathstrut}\circ\breve{V\mathstrut}$ of (\ref{eqDefImplicitCor})
according to Lemma~\ref{lemCombinedCD} is also a solution of (\ref{eqCombinedCD}) with $\hat{\rho}$ and we can write:
\begin{align}
{E\mathstrut}_{\!c}^{ML}(R, Q)
\;\;\;\; & \!\!\overset{\text{touch}}{=} \;\;\;\;
\;\;\;\;\;\;\;\;\;\;\;\;\;\;\;\;\;\;\;\;\;\;\;\;\;\;\;\;\;
\,
\;\;\;\;\;\;
D(\breve{T\mathstrut}\circ \breve{V\mathstrut} \, \| \, Q \circ P)
\,
- \,
\hat{\rho} \big[R \, - \, D(\breve{T\mathstrut} \circ \breve{V\mathstrut} \, \| \, \breve{T\mathstrut} \times Q)\big]
\nonumber \\
& \overset{(a)}{=} \;\;\;\;
\max_{\substack{\\\!\!\!-1 \, \leq \, \rho \, \leq \, 0}}
\;\;\;\;\;\;\;\;\;\;\;\;\;\;\;\;\;\;\;\;\;\;\,\,
\Big\{
D(\breve{T\mathstrut}\circ \breve{V\mathstrut} \, \| \, Q \circ P)
\,
- \,
\rho \big[R \, - \, D(\breve{T\mathstrut} \circ \breve{V\mathstrut} \, \| \, \breve{T\mathstrut} \times Q)\big]
\Big\}
\nonumber \\
& \overset{(b)}{\geq} \;\;\;\,\,
\max_{\substack{\\\!\!\!-1 \, \leq \, \rho \, \leq \, 0}}
\;\;\;\;\;\;\;\;\;\;\;\;\;\;\;\;\;\;\;\;\;\;\,\,
\Big\{
D(\breve{T\mathstrut}\circ \breve{V\mathstrut} \, \| \, \breve{Q\mathstrut} \circ P)
\,
- \,
\rho \big[R \, - \, D(\breve{T\mathstrut} \circ \breve{V\mathstrut} \, \| \, \breve{T\mathstrut} \times \breve{Q\mathstrut})\big]
\Big\}
\nonumber \\
& \geq \;\;\;\;\,
\max_{\substack{\\\!\!\!-1 \, \leq \, \rho \, \leq \, 0}}
\;\;
\min_{\substack{\\T(y), \, V(x\,|\,y):\\ \\\text{supp}(V)\;\subseteq\;\text{supp}(\breve{Q\mathstrut})}}
\Big\{
D(T\circ V \, \| \, \breve{Q\mathstrut} \circ P)
\,
- \,
\rho \big[R \, - \, D(T \circ V \, \| \, T \times \breve{Q\mathstrut})\big]
\Big\}
\nonumber \\
& \overset{(c)}{=} \;\;\;\;
{E\mathstrut}_{\!c}^{ML}(R, \breve{Q\mathstrut}),
\nonumber
\end{align}
where

($a$) holds because for $\hat{\rho} \in (-1, 0)$ Lemma~\ref{lemCombinedCD}
gives $R = D(\breve{T\mathstrut} \circ \breve{V\mathstrut} \, \| \, \breve{T\mathstrut} \times Q)$
and the brackets are zero.
For $\hat{\rho} = 0$ Lemma~\ref{lemCombinedCD} gives $R \leq D(\breve{T\mathstrut} \circ \breve{V\mathstrut} \, \| \, \breve{T\mathstrut} \times Q)$,
so that the brackets are non-positive and the maximum is at $\rho = 0$.
In the case $\hat{\rho} = -1$ Lemma~\ref{lemCombinedCD} gives $R \geq D(\breve{T\mathstrut} \circ \breve{V\mathstrut} \, \| \, \breve{T\mathstrut} \times Q)$,
so that the brackets are non-negative and the maximum is at $\rho = -1$.

($b$) holds because by replacing $Q(x)$ with $\breve{Q\mathstrut}(x) = \sum_{y}\breve{T\mathstrut}(y) \breve{V\mathstrut}(x \, | \, y)$
we obtain in the expression
\begin{align}
& D(\breve{T\mathstrut}\circ \breve{V\mathstrut} \, \| \, Q \circ P)
\, + \,
\rho D(\breve{T\mathstrut} \circ \breve{V\mathstrut} \, \| \, \breve{T\mathstrut} \times Q)
\nonumber \\
= \;\;\;
& \sum_{x, \, y}
\breve{T\mathstrut}(y) \breve{V\mathstrut}(x \, | \, y)\log\frac{\breve{T\mathstrut}(y)}{P(y\,|\,x)}
\, + \,
(1+\rho) \sum_{x, \, y}\breve{T\mathstrut}(y) \breve{V\mathstrut}(x \, | \, y)\log\frac{\breve{V\mathstrut}(x \, | \, y)}{Q(x)}
\nonumber \\
= \;\;\;
& \sum_{x, \, y}
\breve{T\mathstrut}(y) \breve{V\mathstrut}(x \, | \, y)\log\frac{\breve{T\mathstrut}(y)}{P(y\,|\,x)}
\, + \,
(1+\rho) \sum_{x, \, y}\breve{T\mathstrut}(y) \breve{V\mathstrut}(x \, | \, y)\log\frac{\breve{V\mathstrut}(x \, | \, y)}{\breve{Q\mathstrut}(x)}
\, + \, \underbrace{(1 + \rho)D(\breve{Q\mathstrut} \, \| \, Q)}_{\geq \; 0}
\label{eqSqueeze} \\
\geq \;\;\;
& \sum_{x, \, y}
\breve{T\mathstrut}(y) \breve{V\mathstrut}(x \, | \, y)\log\frac{\breve{T\mathstrut}(y)}{P(y\,|\,x)}
\, + \,
(1+\rho) \sum_{x, \, y}\breve{T\mathstrut}(y) \breve{V\mathstrut}(x \, | \, y)\log\frac{\breve{V\mathstrut}(x \, | \, y)}{\breve{Q\mathstrut}(x)}
\nonumber \\
= \;\;\;
& D(\breve{T\mathstrut}\circ \breve{V\mathstrut} \, \| \, \breve{Q\mathstrut} \circ P)
\, + \,
\rho D(\breve{T\mathstrut} \circ \breve{V\mathstrut} \, \| \, \breve{T\mathstrut} \times \breve{Q\mathstrut}).
\nonumber
\end{align}

($c$) holds by (\ref{eqExplicitCD}) and (\ref{eqExplicitSolCD}).\newline
The bound on the amount of decrease follows from (\ref{eqSqueeze}).
\end{proof}

\bigskip

Note that whenever $\,{\hat{\rho}\mathstrut}_{\ell} \in (-1, 0]\,$
the computation in (\ref{eqUpdate}) goes along (\ref{eqTrho})-(\ref{eqVrho}) with ${\hat{\rho}\mathstrut}_{\ell}\,$,
which gives the ratio ${Q\mathstrut}_{\ell\,+\,1}(x)/{Q\mathstrut}_{\ell}(x)$ different than in the Arimoto computation \cite[eq.~24-25]{Arimoto76} of ${Q\mathstrut}_{\ell\,+\,1}$
from ${Q\mathstrut}_{\ell}$ for the same ${\hat{\rho}\mathstrut}_{\ell}$.
Besides, the slope parameter ${\hat{\rho}\mathstrut}_{\ell}$ itself is changing here in each iteration.
The following theorem tries to characterize the convergence of the above minimization procedure.

\bigskip

\begin{thm}[Convergence of iterations for ML]\label{thmConvergence}\newline
{\em Let $\,{\left\{\big(\breve{T\mathstrut}_{\ell}, \breve{V\mathstrut}_{\!\ell}\big)\right\}\mathstrut}_{\ell\,=\,0}^{+\infty}$
be a sequence of iterative solutions of (\ref{eqDefImplicitCor}) with $Q = {Q\mathstrut}_{\ell}$ obtained by (\ref{eqUpdate}).
Then}
\begin{equation} \label{eqZMin}
{E\mathstrut}_{\!c}^{ML}(R, {Q\mathstrut}_{\ell}) \;\;\;\; \overset{\ell \, \rightarrow\,\infty}{\searrow} \;\;\;\; \min_{\substack{\\Q:\;\;\text{supp}(Q)\,\subseteq \,{\cal Z}}}\, {E\mathstrut}_{\!c}^{ML}(R, Q),
\end{equation}
{\em for some} ${\cal Z}\subseteq \text{supp}({Q\mathstrut}_{0})$.
\end{thm}

\bigskip

\begin{proof}
By Theorem~\ref{thm3}~/~Lemma~\ref{lemExplicitSolutionCD} the graph of ${E\mathstrut}_{\!c}^{ML}(\,\,\widetilde{\!\!R}, {Q\mathstrut}_{\ell})$ touches at $\,\,\widetilde{\!\!R} = R$ some supporting line of the form (\ref{eqCombinedCD}),
not necessarily unique.
Let's choose a slope parameter of one such line ${\hat{\rho}\mathstrut}_{\ell} \in [-1, 0]$ for each index $\,\ell$.
Then we have a sequence of pairs $\,{\Big\{\big({Q\mathstrut}_{\ell}, {\hat{\rho}\mathstrut}_{\ell}\big)\Big\}\mathstrut}_{\ell\,=\,0}^{+\infty}$.
By Theorem~\ref{thm3} the distribution ${Q\mathstrut}_{\ell}$ is updated for the next time $\ell + 1$ according to
either
(\ref{eqTrho})-(\ref{eqVrho}) with ${\hat{\rho}\mathstrut}_{\ell} \in (-1, 0]$ or (\ref{eqTminus1})-(\ref{eqVminus1}) if ${\hat{\rho}\mathstrut}_{\ell} = -1$.
In both cases the support of the distribution ${Q\mathstrut}_{\ell}$ cannot increase. It can decrease by (\ref{eqVminus1}),
or the distribution ${Q\mathstrut}_{\ell}$ can approach arbitrarily close to zero on some letters of the channel input alphabet where the initial value of ${Q\mathstrut}_{0}$
is positive.
Consider a convergent subsequence of adjacent pairs
$\,{\left\{\big({Q\mathstrut}_{{\ell\mathstrut}_{i}}, \, {\hat{\rho}\mathstrut}_{{\ell\mathstrut}_{i}}, \,
{Q\mathstrut}_{{\ell\mathstrut}_{i}\,+\,1}, \,
{\hat{\rho}\mathstrut}_{{\ell\mathstrut}_{i}\,+\,1}\big)\right\}\mathstrut}_{i\,=\,1}^{+\infty}\,$:
\begin{align}
{Q\mathstrut}_{{\ell\mathstrut}_{i}} \;\;\;\; & \underset{i\,\rightarrow\,\infty}{\longrightarrow} \;\;\;\; \,{\overline{\!Q\mathstrut}}_{1},
\;\;\;\;\;\;\;\;\;\;\;\;
\;\;\;\;
{\hat{\rho}\mathstrut}_{{\ell\mathstrut}_{i}} \;\;\;\; \underset{i\,\rightarrow\,\infty}{\longrightarrow} \;\;\;\; {\bar{\rho}\mathstrut}_{1} \, \in \, [-1, 0],
\nonumber \\
{Q\mathstrut}_{{\ell\mathstrut}_{i}\,+\,1} \;\;\;\; & \underset{i\,\rightarrow\,\infty}{\longrightarrow} \;\;\;\; \,{\overline{\!Q\mathstrut}}_{2},
\;\;\;\;\;\;\;\;\;\;\;\;
{\hat{\rho}\mathstrut}_{{\ell\mathstrut}_{i}\,+\,1} \;\;\;\; \underset{i\,\rightarrow\,\infty}{\longrightarrow} \;\;\;\; {\bar{\rho}\mathstrut}_{2} \, \in \, [-1, 0].
\nonumber
\end{align}
We have $\,\text{supp}(\,{\overline{\!Q\mathstrut}}_{j})\subseteq \text{supp}({Q\mathstrut}_{0})$, $j = 1, 2$.




Let us first examine the limit of the graph of $\,{E\mathstrut}_{\!c}^{ML}\big(\,\,\widetilde{\!\!R}, \,{Q\mathstrut}_{{\ell\mathstrut}_{i}}\big)\,$
as a function of $\,\,\widetilde{\!\!R}$.
For any $\beta \in (0, 1)$ arbitrarily close to $1$ and large enough index $i$ we can write according to (\ref{eqExplicitCD}) and (\ref{eqE0Minus1}):
\begin{align}
\sup_{\substack{\\\!\!\!-1 \, < \, \rho \, \leq \, 0}}
\Bigg\{
-\log \sum_{y}\bigg[\sum_{x}\underbrace{\beta\,\overline{\!Q\mathstrut}_{1}(x)}_{\leq \; {Q\mathstrut}_{{\ell\mathstrut}_{i}}(x)}P^{\frac{1}{1\,+\,\rho}}(y\,|\,x)\bigg]^{1\,+\,\rho}
\, - \, \rho \,\,\widetilde{\!\!R}\,\Bigg\}
\;\;\; & \geq \;\;\;
{E\mathstrut}_{\!c}^{ML}\big(\,\,\widetilde{\!\!R}, \,{Q\mathstrut}_{{\ell\mathstrut}_{i}}\big)
\nonumber \\
&
\geq \;\;\;
\max_{\substack{\\\!\!\!-\beta \, \leq \, \rho \, \leq \, 0}}
\;\;
\Big\{
{E\mathstrut}_{0}\big(\rho, \,{Q\mathstrut}_{{\ell\mathstrut}_{i}}\big) - \rho \,\,\widetilde{\!\!R}\,
\Big\}.
\nonumber
\end{align}
Now it is convenient to take $i$ to $+\infty$ in the lower bound.
From which we obtain for any $\beta \in (0, 1)$:
\begin{align}
\max_{\substack{\\\!\!\!-1 \, \leq \, \rho \, \leq \, 0}}
\;\;
\Big\{
{E\mathstrut}_{0}(\rho, \,\overline{\!Q\mathstrut}_{1}) \, - \, (1+\rho) \log \beta \, - \, \rho \,\,\widetilde{\!\!R}\,
\Big\}
\;\;\; & \geq \;\;\;
\limsup_{i \, \rightarrow \,\infty}
\;
{E\mathstrut}_{\!c}^{ML}\big(\,\,\widetilde{\!\!R}, \,{Q\mathstrut}_{{\ell\mathstrut}_{i}}\big),
\nonumber \\
\max_{\substack{\\\!\!\!-\beta \, \leq \, \rho \, \leq \, 0}}
\;\;
\Big\{
{E\mathstrut}_{0}(\rho, \,\overline{\!Q\mathstrut}_{1}) \, - \, \rho \,\,\widetilde{\!\!R}\,
\Big\}
\;\;\;
& \leq \;\;\;
\liminf_{i \, \rightarrow \,\infty}
\;\;
{E\mathstrut}_{\!c}^{ML}\big(\,\,\widetilde{\!\!R}, \,{Q\mathstrut}_{{\ell\mathstrut}_{i}}\big).
\nonumber
\end{align}
Then by continuity of ${E\mathstrut}_{0}(\rho, \,\overline{\!Q\mathstrut}_{1})$ as a function of $\rho$ (\ref{eqE0Minus1}) and (\ref{eqExplicitCD}) we obtain
\begin{displaymath}
\lim_{i \, \rightarrow \,\infty}
{E\mathstrut}_{\!c}^{ML}\big(\,\,\widetilde{\!\!R}, \,{Q\mathstrut}_{{\ell\mathstrut}_{i}}\big)
\;\; = \;\;
{E\mathstrut}_{\!c}^{ML}(\,\,\widetilde{\!\!R}, \,\,\overline{\!Q\mathstrut}_{1}), \;\;\;\;\;\; \forall \; \,\,\widetilde{\!\!R}.
\end{displaymath}
In particular, supporting lines $\,{E\mathstrut}_{0}\big({\hat{\rho}\mathstrut}_{{\ell\mathstrut}_{i}}, \,{Q\mathstrut}_{{\ell\mathstrut}_{i}}\big) - {\hat{\rho}\mathstrut}_{{\ell\mathstrut}_{i}} \,\,\widetilde{\!\!R}\,$ of ${E\mathstrut}_{\!c}^{ML}\big(\,\,\widetilde{\!\!R}, \,{Q\mathstrut}_{{\ell\mathstrut}_{i}}\big)$ converge to the supporting line of ${E\mathstrut}_{\!c}^{ML}(\,\,\widetilde{\!\!R}, \,\,\overline{\!Q\mathstrut}_{1})$ 
with slope parameter ${\bar{\rho}\mathstrut}_{1}$,
which is given by $\,{E\mathstrut}_{0}({\bar{\rho}\mathstrut}_{1}, \,\overline{\!Q\mathstrut}_{1}) \, - \, {\bar{\rho}\mathstrut}_{1}  \,\,\widetilde{\!\!R}$.
At $\,\,\widetilde{\!\!R} = R$ this gives
$\,{E\mathstrut}_{\!c}^{ML}\big(R, \,{Q\mathstrut}_{{\ell\mathstrut}_{i}}\big) \rightarrow {E\mathstrut}_{0}({\bar{\rho}\mathstrut}_{1}, \,\overline{\!Q\mathstrut}_{1}) \, - \, {\bar{\rho}\mathstrut}_{1}  R$.
Similarly we obtain $\,{E\mathstrut}_{\!c}^{ML}\big(R, \,{Q\mathstrut}_{{\ell\mathstrut}_{i}\,+\,1}\big) \,\rightarrow\, {E\mathstrut}_{0}({\bar{\rho}\mathstrut}_{2}, \,\overline{\!Q\mathstrut}_{2}) \, - \, {\bar{\rho}\mathstrut}_{2}  R$.

If ${\bar{\rho}\mathstrut}_{1} = 0$ or ${\bar{\rho}\mathstrut}_{2} = 0$, then 
$\,{E\mathstrut}_{\!c}^{ML}\big(R, \,{Q\mathstrut}_{{\ell\mathstrut}_{i}}\big) \searrow  0\,$ by the above result
and monotonicity of Lemma~\ref{lemMonotML}.

If ${\bar{\rho}\mathstrut}_{1} = -1$ then
$\,{E\mathstrut}_{\!c}^{ML}\big(R, \,{Q\mathstrut}_{{\ell\mathstrut}_{i}}\big) \,\searrow\, {E\mathstrut}_{0}(-1, \,\overline{\!Q\mathstrut}_{1}) + R$.
By (\ref{eqExplicitCD}) and (\ref{eqE0Minus1})
we conclude that this is the minimum of ${E\mathstrut}_{\!c}^{ML}(R, Q)$ over all $Q$ with $\,\text{supp}(Q) \subseteq \text{supp}(\,\overline{\!Q\mathstrut}_{1})$.
Similarly if ${\bar{\rho}\mathstrut}_{2} = -1$ then
$\,{E\mathstrut}_{\!c}^{ML}\big(R, \,{Q\mathstrut}_{{\ell\mathstrut}_{i}}\big) \,\searrow\, {E\mathstrut}_{0}(-1, \,\overline{\!Q\mathstrut}_{2}) + R$,
which is the minimum of ${E\mathstrut}_{\!c}^{ML}(R, Q)$ over all $Q$ with $\,\text{supp}(Q) \subseteq \text{supp}(\,\overline{\!Q\mathstrut}_{2})$.

Suppose now that ${\bar{\rho}\mathstrut}_{1}, {\bar{\rho}\mathstrut}_{2} \in (-1, 0)$.
Since ${\bar{\rho}\mathstrut}_{1} \in (-1, 0)$, also ${\hat{\rho}\mathstrut}_{{\ell\mathstrut}_{i}} \in (-1, 0)$ for large enough index $i$ and
the distribution ${Q\mathstrut}_{{\ell\mathstrut}_{i}}$ is updated for the next time ${\ell\mathstrut}_{i}+1$ according to (\ref{eqTrho})-(\ref{eqVrho}) as:
\begin{equation} \label{eqNexQ}
{Q\mathstrut}_{{\ell\mathstrut}_{i}\,+\,1}(x)
\;\; = \;\;
{Q\mathstrut}_{{\ell\mathstrut}_{i}}(x)\cdot
\frac{1}{{K\mathstrut}_{i}}\,
\sum_{\substack{\\y:\; P(y \, | \, x) \, > \, 0}}P^{\,{\gamma\mathstrut}_{i}}(y\,|\,x)
\bigg[\sum_{a}{Q\mathstrut}_{{\ell\mathstrut}_{i}}(a)P^{\,{\gamma\mathstrut}_{i}}(y\,|\,a)\bigg]^{{\hat{\rho}\mathstrut}_{{\ell\mathstrut}_{i}}},
\;\;\;\;\;\; {\gamma\mathstrut}_{i} \, \triangleq \, \frac{1}{1\,+\,{\hat{\rho}\mathstrut}_{{\ell\mathstrut}_{i}}}.
\end{equation}
In the limit where $\,\overline{\!Q\mathstrut}_{1}(x)$ is positive (\ref{eqNexQ}) becomes
\begin{equation} \label{eqNextQLim}
\,\overline{\!Q\mathstrut}_{2}(x)
\;\; = \;\;
\,\overline{\!Q\mathstrut}_{1}(x)\cdot
\frac{1}{K}
\sum_{\substack{\\y: \; P(y \, | \, x) \, > \, 0}}P^{\,\gamma}(y\,|\,x)
\bigg[\sum_{a}\,\overline{\!Q\mathstrut}_{1}(a)P^{\,\gamma}(y\,|\,a)\bigg]^{{\bar{\rho}\mathstrut}_{1}},
\;\;\;\;\;\; \gamma \, \triangleq \, \frac{1}{1\,+\,{\bar{\rho}\mathstrut}_{1}}.
\end{equation}
Since also ${\bar{\rho}\mathstrut}_{2} \in (-1, 0)$ then $1 + {\hat{\rho}\mathstrut}_{{\ell\mathstrut}_{i}\,+\,1}$ converges to a positive number and
by Lemma~\ref{lemMonotML} necessarily $D\big({Q\mathstrut}_{{\ell\mathstrut}_{i}\,+\,1} \, \| \, {Q\mathstrut}_{{\ell\mathstrut}_{i}}\big) \, \rightarrow \, 0$.
In this case also $\,{Q\mathstrut}_{{\ell\mathstrut}_{i}\,+\,1} \rightarrow \,\overline{\!Q\mathstrut}_{1}$,
i.e. necessarily $\,\overline{\!Q\mathstrut}_{1} = \,\overline{\!Q\mathstrut}_{2}$.
Dividing both sides of (\ref{eqNextQLim}) by $\,\overline{\!Q\mathstrut}_{1}(x)$ where it is positive,
for all such $x$
we obtain:
\begin{equation} \label{eqNextQLimQPos}
\sum_{\substack{\\y:\; P(y \, | \, x) \, > \, 0}}P^{\,\gamma}(y\,|\,x)
\bigg[\sum_{a}\,\overline{\!Q\mathstrut}_{1}(a)P^{\,\gamma}(y\,|\,a)\bigg]^{{\bar{\rho}\mathstrut}_{1}} \;\; = \;\; K,
\;\;\;\;\;\; \,\overline{\!Q\mathstrut}_{1}(x) > 0.
\end{equation}
This can be recognized as a sufficient condition for $\,\overline{\!Q\mathstrut}_{1}$ to minimize ${E\mathstrut}_{0}({\bar{\rho}\mathstrut}_{1}, Q)$ over all $Q$ with $\,\text{supp}(Q) \subseteq \text{supp}(\,\overline{\!Q\mathstrut}_{1})$, the same as \cite[eq.~22]{Arimoto76}.
By (\ref{eqExplicitCD}) we conclude that the limit of $\,{E\mathstrut}_{\!c}^{ML}\big(R, \,{Q\mathstrut}_{{\ell\mathstrut}_{i}}\big)$ which is given by $\,{E\mathstrut}_{0}({\bar{\rho}\mathstrut}_{1}, \,\overline{\!Q\mathstrut}_{1}) \, - \, {\bar{\rho}\mathstrut}_{1}  R\,$
is the minimum of $\,{E\mathstrut}_{\!c}^{ML}(R, \,Q)\,$ over such $Q$.
\end{proof}

\bigskip

Let $C({\cal Z})$ denote the capacity of the channel with an input alphabet ${\cal Z} \subseteq {\cal X}$.
Observe that 
for any $R > 0$ holds
\begin{displaymath}
\min_{\substack{\\{\cal Z}:\;\;C\left({\cal Z}\right) \, < \, R}}\;\;\;
\min_{\substack{\\Q:\;\;\text{supp}(Q)\,\subseteq \,{\cal Z}}}\, {E\mathstrut}_{\!c}^{ML}(R, Q)
\;\;\; = \;\;\;
\min_{\substack{\\Q:\;\;C\left(\text{supp}(Q)\right) \; < \; R}}\, {E\mathstrut}_{\!c}^{ML}(R, Q) \;\;\; > \;\;\; 0.
\end{displaymath}
This observation conveniently allows us to grasp and
write
one sufficient condition for the convergence
of the iterative minimization using (\ref{eqDefImplicitCor}) described by Lemma~\ref{lemMonotML} and also of the analogous procedure for (\ref{eqExpNoTB})
all the way to the minimum over $Q$ (\ref{eqOptimum}),
when this minimum is zero.

\bigskip

\begin{lemma} [Convergence to zero for ML] \label{lemConvergenceReg}\newline
{\em Let $\,{\left\{\big(\breve{T\mathstrut}_{\ell}, \breve{V\mathstrut}_{\!\ell}\big)\right\}\mathstrut}_{\ell\,=\,0}^{+\infty}$
be a sequence of iterative solutions of (\ref{eqDefImplicitCor}) with $Q = {Q\mathstrut}_{\ell}$ obtained by (\ref{eqUpdate}).\newline
If 
the initial distribution ${Q\mathstrut}_{0}$ satisfies the strict inequality:}
\begin{equation} \label{eqLowerThan}
{E\mathstrut}_{\!c}^{ML}(R, {Q\mathstrut}_{0}) \;\;\;\; < \;\;\;\; \min_{\substack{\\Q:\;\;C\left(\text{supp}(Q)\right) \; < \; R}}\, {E\mathstrut}_{\!c}^{ML}(R, Q),
\end{equation}
{\em then}
\begin{displaymath}
{E\mathstrut}_{\!c}^{ML}(R, {Q\mathstrut}_{\ell})
\;\;\;\; \overset{\ell \, \rightarrow\,\infty}{\searrow} \;\;\;\; 0.
\end{displaymath}
\end{lemma}

\bigskip

\begin{proof}
By Theorem~\ref{thmConvergence} the resultant sequence of ${E\mathstrut}_{\!c}^{ML}(R, {Q\mathstrut}_{\ell})$
must monotonically converge to a minimum of ${E\mathstrut}_{\!c}^{ML}(R, Q)$ over $Q$ with $\,\text{supp}(Q) \subseteq {\cal Z}\,$
for some subset of the channel input alphabet $\,{\cal Z} \subseteq {\cal X}$.
Suppose that $\,C({\cal Z}) < R$.
Then also for every subset $\,\text{supp}(Q) \subseteq {\cal Z}\,$ we have $C\big(\text{supp}(Q)\big) < R$.
Then the limit of the sequence of ${E\mathstrut}_{\!c}^{ML}(R, {Q\mathstrut}_{\ell})$ must be lower-bounded by the minimum on the RHS of (\ref{eqLowerThan}).
This is a contradiction, since the monotonically non-increasing sequence must be upper-bounded by its first element
${E\mathstrut}_{\!c}^{ML}(R, {Q\mathstrut}_{0})$ on the LHS of the strict inequality (\ref{eqLowerThan}).
We conclude that necessarily $\,C({\cal Z}) \geq R$.
In particular, there exists some $Q$ with $\,\text{supp}(Q) \subseteq {\cal Z}\,$
such that $I(Q\circ P) \geq R$. This gives ${E\mathstrut}_{\!c}^{ML}(R, Q) = 0$ by (\ref{eqDefImplicitCor}) for this $Q$.
Consequently the minimum in (\ref{eqZMin}) is zero.
\end{proof}

\bigskip

Note that for each $\,0 < R \leq C({\cal X})\,$ there exist such initial input distributions ${Q\mathstrut}_{0}$ that satisfy the condition
(\ref{eqLowerThan}) of Lemma~\ref{lemConvergenceReg}.
Therefore (\ref{eqLowerThan}) guarantees a region of convergence of (\ref{eqDefImplicitCor}) to (\ref{eqOptimum}) as a result of the iterative procedure (\ref{eqUpdate}) for $0 < R \leq C({\cal X})$.
Next, we extend the above result from (\ref{eqDefImplicitCor}) to (\ref{eqExpNoTB}).

\bigskip

\begin{lemma} \label{lemFullCap}
$\,{E\mathstrut}_{0}(-1, Q) + C\big(\text{supp}(Q)\big) \, \leq \, 0$.
\end{lemma}

\bigskip

\begin{proof}
Suppose on the contrary that $\,{E\mathstrut}_{0}(-1, Q) + C\big(\text{supp}(Q)\big) \, > \, 0$.
Then
\begin{align}
\min_{\substack{\\\,\widetilde{\!Q}:\;\;\text{supp}(\,\widetilde{\!Q})\,\subseteq \,\text{supp}(Q)}}\, {E\mathstrut}_{\!c}^{ML}\big(C\big(\text{supp}(Q)\big), \, \,\widetilde{\!Q}\big)
\;\; & \overset{(\ref{eqExplicitCD})}{\geq} \;\;
\min_{\substack{\\\,\widetilde{\!Q}:\;\;\text{supp}(\,\widetilde{\!Q})\,\subseteq \,\text{supp}(Q)}}
\Big\{{E\mathstrut}_{0}(-1, \,\widetilde{\!Q}) + C\big(\text{supp}(Q)\big)\Big\}
\nonumber \\
& \overset{(\ref{eqE0Minus1})}{\geq} \;\;
\min_{\substack{\\\,\widetilde{\!Q}:\;\;\text{supp}(\,\widetilde{\!Q})\,\subseteq \,\text{supp}(Q)}}
\Big\{{E\mathstrut}_{0}(-1, Q) + C\big(\text{supp}(Q)\big)\Big\}
\;\; > \;\; 0,
\nonumber
\end{align}
i.e. the minimal correct-decoding exponent $\,\min_{\,\,\widetilde{\!Q}}{E\mathstrut}_{\!c}^{ML}\big(R, \, \,\widetilde{\!Q}\big)\,$ of the channel with the input alphabet $\,\text{supp}(Q)\,$ appears to be positive at $R = C\big(\text{supp}(Q)\big)$, which is a contradiction.
\end{proof}

\bigskip

\begin{lemma}[One iteration]\label{lemOneStep}\newline
{\em If} $\,C\big(\text{supp}(Q)\big) \geq R\,$
{\em then the following holds:}\newline
{\em (1)} $\, C\big(\text{supp}(Q)\big) \, \leq \, {R\mathstrut}_{\,-1}^{\,-}(Q)${\em, as defined in (\ref{eqRMinus}),} \newline
{\em (2)} {\em $\,{E\mathstrut}_{\!c}(R, Q) = {E\mathstrut}_{\!c}^{ML}(R, Q)$,
as defined in (\ref{eqExpNoTB}) and (\ref{eqDefImplicitCor}),
sharing the same solution
$(\breve{T\mathstrut}, \breve{V\mathstrut})$,}\newline
{\em (3)} $\,\breve{Q\mathstrut}(x) = \sum_{y}\breve{T\mathstrut}(y) \breve{V\mathstrut}(x \, | \, y)$
{\em satisfies} $\,\text{supp}(\breve{Q\mathstrut}) = \text{supp}(Q)${\em .}
\end{lemma}

\bigskip

\begin{proof}
By Lemma~\ref{lemFullCap}
\begin{equation} \label{eqNonPos}
{E\mathstrut}_{0}(-1, Q) \, + \, C\big(\text{supp}(Q)\big) \; \leq \; 0.
\end{equation}
This means that by Lemmas~\ref{lemCombinedCD},~\ref{lemExplicitSolutionCD} and definition (\ref{eqRMinus}) the supporting line
$\,E(\,\,\widetilde{\!\!R}) = {E\mathstrut}_{0}(-1, Q) + \,\,\widetilde{\!\!R}\,$ of slope $1$
touches the graph of $\,{E\mathstrut}_{\!c}^{ML}(\,\,\widetilde{\!\!R}, \,Q)\,$
at
\begin{displaymath} 
\,\,\widetilde{\!\!R} \;\; = \;\; {R\mathstrut}_{\,-1}^{\,-}(Q) \;\; \geq \;\; C\big(\text{supp}(Q)\big),
\end{displaymath}
i.e. to the right of $\, \,\,\widetilde{\!\!R} \, = \, C\big(\text{supp}(Q)\big)$.
Then by Theorem~\ref{thm3} we conclude that two things hold ((a) and (b)):\newline
a) 
the graphs of
$\,{E\mathstrut}_{\!c}^{ML}(\,\,\widetilde{\!\!R}, \,Q)\,$
and $\,{E\mathstrut}_{\!c}(\,\,\widetilde{\!\!R}, \,Q)\,$
coincide for all $\, \,\,\widetilde{\!\!R}\,\leq\, C\big(\text{supp}(Q)\big)\,$ and\newline
b) the corresponding minima (\ref{eqDefImplicitCor}) and (\ref{eqExpNoTB})
share the same minimizing solutions there.\newline
In particular this holds at $\, \,\,\widetilde{\!\!R} \,=\, R \,\leq\, C\big(\text{supp}(Q)\big)$.
Then
$\,{E\mathstrut}_{\!c}^{ML}(R, Q) = {E\mathstrut}_{\!c}(R, Q)$, sharing the same solutions
$(\breve{T\mathstrut}, \breve{V\mathstrut})$.\newline
Since $\,C\big(\text{supp}(Q)\big) \geq R\,$ the inequality (\ref{eqNonPos}) with Theorem~\ref{thm3} imply also that the graph of $\,{E\mathstrut}_{\!c}^{ML}(\,\,\widetilde{\!\!R}, \,Q)\,$
touches at $\, \,\,\widetilde{\!\!R} \, = \, R$
some supporting line with slope parameter $\rho \in (-1, 0]$.
Then by Theorem~\ref{thm3} the unique solution $(\breve{T\mathstrut}, \breve{V\mathstrut})$ is determined according to (\ref{eqTrho})-(\ref{eqVrho}).
Then $\,\text{supp}(\breve{Q\mathstrut}) = \text{supp}(Q)$.
\end{proof}

\bigskip

\begin{lemma} [Convergence of iterations] \label{lemMonot}\newline
{\em Let $\,{\left\{\big(\breve{T\mathstrut}_{\ell}, \breve{V\mathstrut}_{\!\ell}\big)\right\}\mathstrut}_{\ell\,=\,0}^{+\infty}$
be a sequence of iterative solutions of (\ref{eqExpNoTB}) with $Q = {Q\mathstrut}_{\ell}$
at each iteration obtained from the previous solution as in (\ref{eqUpdate}).
If the initial distribution satisfies} $\,C\big(\text{supp}({Q\mathstrut}_{0})\big) \geq R\,$
{\em then
for each $\ell$ holds 
$\,{E\mathstrut}_{\!c}(R, {Q\mathstrut}_{\ell}) = {E\mathstrut}_{\!c}^{ML}(R, {Q\mathstrut}_{\ell})$, (\ref{eqDefImplicitCor}),
sharing the same solution $\big(\breve{T\mathstrut}_{\ell}, \breve{V\mathstrut}_{\!\ell}\big)$.}
\end{lemma}

\bigskip

\begin{proof}
Follows from Lemma~\ref{lemOneStep} by induction.
\end{proof}

\bigskip

\begin{thm}[Convergence to zero]\label{thmConvAdapt}\newline
{\em Let $\,{\left\{\big(\breve{T\mathstrut}_{\ell}, \breve{V\mathstrut}_{\!\ell}\big)\right\}\mathstrut}_{\ell\,=\,0}^{+\infty}$
be a sequence of iterative solutions of (\ref{eqExpNoTB}) with $Q = {Q\mathstrut}_{\ell}$
at each iteration obtained from the previous solution as in (\ref{eqUpdate}).
If 
the initial distribution ${Q\mathstrut}_{0}$ satisfies the strict inequality (\ref{eqLowerThan}) for (\ref{eqDefImplicitCor})
then}
\begin{displaymath}
{E\mathstrut}_{\!c}(R, {Q\mathstrut}_{\ell})
\;\; = \;\;
{E\mathstrut}_{\!c}^{ML}(R, {Q\mathstrut}_{\ell})
\;\;\;\; \overset{\ell \, \rightarrow\,\infty}{\searrow} \;\;\;\; 0.
\end{displaymath}
\end{thm}

\bigskip

\begin{proof}
From (\ref{eqLowerThan}) follows $\,C\big(\text{supp}({Q\mathstrut}_{0})\big) \geq R$.
Then Lemma~\ref{lemMonot} applies and the claim follows by Lemma~\ref{lemConvergenceReg}.
\end{proof}

\bigskip

\section{Comparison to the Blahut algorithm}\label{ComptoBlahut}

\bigskip

The fixed-rate iterative computation of the correct-decoding exponent according to (\ref{eqExpNoTB}) and (\ref{eqUpdate}) can be compared to the fixed-distortion version of the Blahut algorithm \cite{Blahut72}, \cite{CsiszarTusnady84}
for the rate-distortion function computation.
As we have seen from (\ref{eqExpNoTB})-(\ref{eqDefImplicitCor}) and (\ref{eqOptimum}), the correct-decoding exponent can be written as a double minimum,
quite similarly to the rate-distortion function:

\begin{displaymath}
\begin{array}{c | c}
\text{double minimization} & \text{iterative computation} \\
\hline
\displaystyle{ \;\;\;\;\;\;\, 0 \;\;\; = \;\;\; \min_{\substack{\\Q(x)}}
\;\;
\min_{\substack{\\T(y), \, V(x \,|\, y): \\ \\\sum_{x, \, y} T(y)V(x\,|\,y)\log\frac{V(x\,|\,y)}{Q(x)} \; \geq \; R}}
D(T \circ V \, \| \, Q \circ P)} \; &
\;
{\displaystyle
{Q\mathstrut}_{\ell\,+\,1}(x) \; = \; \sum_{y}\breve{T\mathstrut}_{\ell}(y)\breve{V\mathstrut}_{\!\ell}(x\, | \, y)}
\\
\hline
\displaystyle{ R(D) \;\;\; = \;\;\; \min_{\substack{\\Q(x)}} \;\,\,
\;\;
\min_{\substack{\\V(x\, | \, y): \\ \\\sum_{x, \, y} T(y)V(x \, | \,y)\,d(y, \,x) \; \leq \; D}}
\;\,\, D(T\circ V\,\|\, Q \times T)} \; &
\;
{\displaystyle {Q\mathstrut}_{\ell\,+\,1}(x) \; = \; \sum_{y}\,T(y)\,\breve{V\mathstrut}_{\!\ell}(x\,|\,y)}
\end{array}
\end{displaymath}

\bigskip

The rate-distortion function $R(D)$ has also a meaning of an optimal probability exponent \cite{ZamirRose01}.
Here in the second row the discrete memoryless source is denoted as $T(y)$. The i.i.d. source reproduction distribution is denoted as $Q(x)$. The distortion measure is $d(y, x)$.

The algorithm for $R(D)$ is an alternating minimization procedure \cite{CsiszarTusnady84}, i.e.
$\breve{V\mathstrut}_{\!\ell}(x\,|\,y)$ solves the inner minimum of $D(T\circ V\,\|\, {Q\mathstrut}_{\ell} \times T)$
and then ${Q\mathstrut}_{\ell\,+\,1}(x)$ in turn minimizes $D(T\circ \breve{V\mathstrut}_{\!\ell}\,\|\, Q \times T)$.
On the other hand, the proposed algorithm for the correct-decoding exponent is {\em not}
exactly an alternating minimization procedure. Specifically, observe that ${Q\mathstrut}_{\ell\,+\,1}(x)$
minimizes simultaneously both $D\big(\breve{T\mathstrut}_{\ell} \circ \breve{V\mathstrut}_{\!\ell} \, \| \, Q \circ P\big)$ and
$D\big(\breve{T\mathstrut}_{\ell} \circ \breve{V\mathstrut}_{\!\ell} \, \| \, \breve{T\mathstrut}_{\ell} \times Q\big)$ thus violating the condition under the inner minimum with the same $\breve{T\mathstrut}_{\ell} \circ \breve{V\mathstrut}_{\!\ell}$.
Nonetheless, this results in a monotonically non-increasing sequence of the inner minima over $T\circ V$
at least given the condition on ${Q\mathstrut}_{0}$ of Lemma~\ref{lemMonot}.
The sequence converges all the way down to zero at least under the initial condition (\ref{eqLowerThan}) according to Theorem~\ref{thmConvAdapt}.

Alternative metrics from the family (\ref{eqFamily}) can also be used for construction of similar algorithms \cite{TridenskiZamir18}.
Remarkably, the last metric in (\ref{eqFamily}) (with $P(y \, | \, x)$)
results in a very similar algorithm which is an alternating minimization procedure of \cite{CsiszarTusnady84}.

\bigskip

\section{Computation of the exponent function}\label{FixedSlope}

\bigskip

In the iterative minimization procedure at fixed rate $R$ described in the preceding sections
the minimization itself is implicit and the slope parameter $\rho$ is different in each iteration.
Here, the slope parameter is fixed and the computations acquire an explicit form.
The fixed-slope version of the iterative procedure (\ref{eqUpdate}) is similar to the fixed-slope version of the Blahut algorithm \cite{Blahut72},
\cite{Csiszar74}
for $R(D)$ computation
and presents an alternative for the Arimoto algorithm \cite{Arimoto76} for computation of $\,\min_{\,Q}{E\mathstrut}_{0}(\rho, Q)$, $\,\rho \in (-1, 0)$.
Similarly as the Blahut and the Arimoto algorithms, it does not require any special conditions for convergence.
We do not use this version in the rest of the paper and only desribe it in this section.

On the basis of (\ref{eqExplicitSolution}) let us define for $\rho \in (-1, 0)$
\begin{displaymath}
{F\mathstrut}_{\!\rho}(T\circ V, \,Q) \;\; \triangleq \;\;
D(T\circ V \, \| \, Q \circ P)
\,
+ \,
\rho D(T \circ V \, \| \, T \times Q).
\end{displaymath}

\bigskip

\begin{lemma} [Fixed slope] \label{lemFixSlope}\newline
{\em An iterative update of the parameters $\,({T\mathstrut}_{\ell}, {V\mathstrut}_{\!\ell})\,$ and $\,{Q\mathstrut}_{\ell}\,$ for $\ell = 0$, $1$, $2$, ..., starting from ${Q\mathstrut}_{0}\,$:}
\begin{align}
{T\mathstrut}_{\ell}(y) \;\; & = \;\;
\frac{1}{{K\mathstrut}_{\ell}}\,
\bigg[\sum_{a}{Q\mathstrut}_{\ell}(a)P^{\frac{1}{1\,+\,\rho}}(y\,|\,a)\bigg]^{1\,+\,\rho},
\nonumber \\
{V\mathstrut}_{\!\ell}(x \,|\,y) \;\; & = \;\;
\;\;\;
\frac{1}{{K\mathstrut}_{\ell}(y)}\,\,
{Q\mathstrut}_{\ell}(x)P^{\frac{1}{1\,+\,\rho}}(y\,|\,x),
\;\;\;\;\;\; {T\mathstrut}_{\ell}(y) > 0,
\nonumber \\
{Q\mathstrut}_{\ell\,+\,1}(x) \;\; & = \;\;
\sum_{y}{T\mathstrut}_{\ell}(y){V\mathstrut}_{\!\ell}(x\,|\,y),
\nonumber
\end{align}
{\em results in a monotonically non-increasing sequence}
\begin{displaymath}
{F\mathstrut}_{\!\rho}({T\mathstrut}_{0}\circ {V\mathstrut}_{\!0}, \,{Q\mathstrut}_{0})
\; \geq \;
{F\mathstrut}_{\!\rho}({T\mathstrut}_{0}\circ {V\mathstrut}_{\!0}, \,{Q\mathstrut}_{1})
\; \geq \;
{F\mathstrut}_{\!\rho}({T\mathstrut}_{1}\circ {V\mathstrut}_{\!1}, \,{Q\mathstrut}_{1})
\; \geq \;
\cdots.
\end{displaymath}
\end{lemma}

\bigskip

\begin{proof}
The inequality $\,{F\mathstrut}_{\!\rho}({T\mathstrut}_{\ell\,-\,1}\circ {V\mathstrut}_{\!\ell\,-\,1}, \,{Q\mathstrut}_{\ell})
\geq
{F\mathstrut}_{\!\rho}({T\mathstrut}_{\ell}\circ {V\mathstrut}_{\!\ell}, \,{Q\mathstrut}_{\ell})\,$
follows by Lemma~\ref{lemExplicit}.
The inequality $\,{F\mathstrut}_{\!\rho}({T\mathstrut}_{\ell}\circ {V\mathstrut}_{\!\ell}, \,{Q\mathstrut}_{\ell})
\geq
{F\mathstrut}_{\!\rho}({T\mathstrut}_{\ell}\circ {V\mathstrut}_{\!\ell}, \,{Q\mathstrut}_{\ell\,+\,1})\,$
follows by (\ref{eqSqueeze}).
\end{proof}

\bigskip

\begin{thm}[Fixed slope convergence]\label{thmFixSlope}
\begin{equation} \label{eqE0min}
{F\mathstrut}_{\!\rho}({T\mathstrut}_{\ell}\circ {V\mathstrut}_{\!\ell}, \,{Q\mathstrut}_{\ell})
\;\;\;\; \overset{\ell \, \rightarrow\,\infty}{\searrow} \;\;\;\;
\min_{\substack{\\Q:\;\;\text{supp}(Q)\,\subseteq \,\text{supp}({Q\mathstrut}_{0})}}\,
{E\mathstrut}_{0}(\rho, Q).
\end{equation}
\end{thm}

\bigskip

\begin{proof}
Due to similarity to the fixed-slope rate-distortion function computation,
the proof of Csisz\'ar \cite{Csiszar74} is used here as a blueprint.
The iterations start from ${Q\mathstrut}_{0}$.
Consider some other distribution ${Q\mathstrut}^{*}$.
Suppose a joint distribution $\breve{T\mathstrut}\circ\breve{V\mathstrut}$ is obtained from ${Q\mathstrut}^{*}$
according to (\ref{eqTrho})-(\ref{eqVrho}).
From $\breve{T\mathstrut}\circ\breve{V\mathstrut}$
the marginal distribution is obtained as $\breve{Q\mathstrut}(x) = \sum_{y}\breve{T\mathstrut}(y)\breve{V\mathstrut}(x \, | \, y)$.
Suppose $\text{supp}({Q\mathstrut}^{*}) \subseteq \text{supp}({Q\mathstrut}_{0})$.
Then also $\text{supp}(\breve{Q\mathstrut}) \subseteq \text{supp}({Q\mathstrut}_{0})$.
With the help of the identity
\begin{displaymath}
{T\mathstrut}_{\ell}^{1\,+\,\rho}(y){V\mathstrut}_{\!\ell}^{1\,+\,\rho}(x \, | \, y) \; = \;
{Q\mathstrut}_{\ell}^{1\,+\,\rho}(x)P(y\,|\,x){T\mathstrut}_{\ell}^{\rho}(y)\exp \big\{{E\mathstrut}_{0}(\rho, {Q\mathstrut}_{\ell})\big\},
\end{displaymath}
the validity of the following equality can be verified:
\begin{align}
& \underbrace{{F\mathstrut}_{\!\rho}({T\mathstrut}_{\ell}\circ {V\mathstrut}_{\!\ell}, \,{Q\mathstrut}_{\ell})}_{=\; {E\mathstrut}_{0}(\rho, \, {Q\mathstrut}_{\ell})} \; - \;
{F\mathstrut}_{\!\rho}(\breve{T\mathstrut}\circ\breve{V\mathstrut}, \,\breve{Q\mathstrut})
\; + \;
(1 + \rho)\underbrace{\sum_{x, \, y}\breve{T\mathstrut}(y)\breve{V\mathstrut}(x \, | \, y)\log\frac{\breve{T\mathstrut}(y)\breve{V\mathstrut}(x \, | \, y){Q\mathstrut}_{\ell\,+\,1}(x)}{\breve{Q\mathstrut}(x){T\mathstrut}_{\ell}(y){V\mathstrut}_{\!\ell}(x\,|\,y)}}_{\geq \; 0}
\; - \;
\rho D(\breve{T\mathstrut} \, \| \, {T\mathstrut}_{\ell})
\nonumber \\
& = \;
(1 + \rho) \sum_{x}\breve{Q\mathstrut}(x)\log\frac{{Q\mathstrut}_{\ell\,+\,1}(x)}{{Q\mathstrut}_{\ell}(x)}.
\label{eqEight}
\end{align}
This plays the role analogous to \cite[eq.~8]{Csiszar74}. From (\ref{eqEight}) we have the upper bound:
\begin{equation}\label{eqBothSides}
{F\mathstrut}_{\!\rho}({T\mathstrut}_{\ell}\circ {V\mathstrut}_{\!\ell}, \,{Q\mathstrut}_{\ell}) \; - \;
{F\mathstrut}_{\!\rho}(\breve{T\mathstrut}\circ\breve{V\mathstrut}, \,\breve{Q\mathstrut})
\; \leq \;
(1 + \rho) \sum_{x}\breve{Q\mathstrut}(x)\log\frac{{Q\mathstrut}_{\ell\,+\,1}(x)}{{Q\mathstrut}_{\ell}(x)}.
\end{equation}
Similarly as in \cite{Csiszar74}, summing both sides of (\ref{eqBothSides}) over $0 \leq \ell \leq N$
we obtain a telescoping sum on the upper side of the inequality:
\begin{align}
\sum_{\ell \, = \, 0}^{N}\Big[\underbrace{{F\mathstrut}_{\!\rho}({T\mathstrut}_{\ell}\circ {V\mathstrut}_{\!\ell}, \,{Q\mathstrut}_{\ell})}_{=\; {E\mathstrut}_{0}(\rho, \, {Q\mathstrut}_{\ell})} \; - \;
{F\mathstrut}_{\!\rho}(\breve{T\mathstrut}\circ\breve{V\mathstrut}, \,\breve{Q\mathstrut})\Big]
\; & \leq \;
(1 + \rho) \sum_{x}\breve{Q\mathstrut}(x)\log\frac{{Q\mathstrut}_{N\,+\,1}(x)}{{Q\mathstrut}_{0}(x)}
\nonumber \\
& \leq \;
(1 + \rho) \sum_{x}\breve{Q\mathstrut}(x)\log\frac{\breve{Q\mathstrut}(x)}{{Q\mathstrut}_{0}(x)}.
\label{eqTelescoping}
\end{align}
Suppose now that ${Q\mathstrut}^{*}$ is some distribution achieving the minimum in (\ref{eqE0min}).
Then ${E\mathstrut}_{0}(\rho, {Q\mathstrut}_{\ell}) \geq {E\mathstrut}_{0}(\rho, {Q\mathstrut}^{*}) = {F\mathstrut}_{\!\rho}(\breve{T\mathstrut}\circ\breve{V\mathstrut}, \,\breve{Q\mathstrut}) = {E\mathstrut}_{0}(\rho, \breve{Q\mathstrut})\,$
and all the differences on the LHS of (\ref{eqTelescoping}) are non-negative.
On the other hand, since $\text{supp}(\breve{Q\mathstrut}) \subseteq \text{supp}({Q\mathstrut}_{0})$,
the upper bound on the RHS of (\ref{eqTelescoping}) is bounded. Since $N$ on the left can be taken to $+\infty$ the claim follows.
\end{proof}

Other similar fixed-slope algorithms, different from \cite{Arimoto76}, can be developed using the second and the third metrics in the family (\ref{eqFamily}) (the ones with $\Phi(x \, | \, y)$).

\bigskip

\section{Channel input adaptation}\label{ChInAd}

\bigskip

As we have seen, the correct-decoding exponent for channels exhibits properties reminiscent of the rate-distortion function for sources.
In \cite{ZamirRose01} the phenomenon of natural type selection in lossy source-{\em encoding} was found to be a stochastic counterpart
of the Blahut algorithm.
In this section we describe an analogous phenomenon in noisy-channel {\em decoding}
as a stochastic counterpart of the fixed-rate iterative minimization of the correct-decoding exponent
presented in Section~\ref{CordecExpComp}.

\bigskip

\subsection{Adaptation scheme}\label{Adscheme}

\bigskip

The decoder looks for the message $\widehat{m}$ such that
\begin{equation} \label{eqFirstDec}
D(T\circ {V\!\mathstrut}_{\widehat{m}} \,\|\, T \times Q) \; > \; D(T\circ {V\!\mathstrut}_{m} \,\|\, T \times Q), \;\;\;\;\;\; \forall \;\; m \, \neq \, \widehat{m},
\end{equation}
where $T(y)$ is the type of the received block ${\bf y}$ and ${V\!\mathstrut}_{m}(x \, | \, y)$ represents the conditional type of each codeword ${\bf x\mathstrut}_{m}$ given the received block.
This is equivalent to the decoding rule (\ref{eqDecRule}) with
$A(T\circ V) \equiv B(T\circ V) \equiv D(T\circ V \,\|\, T \times Q)$,
where $D(T\circ V \,\|\, T \times Q)$ is the metric average (\ref{eqMMIlikeDec}).
The error exponent given for this case by Corollary~\ref{cor1} is ${E\mathstrut}_{\!e}(R, Q)$ (\ref{eqDefImplicit}).
This is the same as the error exponent of the ML decoder according to Theorem~\ref{thm2}.
We assume that the communication rate $R$ is lower than the mutual information $I(Q\circ P)$
so that
\begin{displaymath}
{E\mathstrut}_{\!e}(R, Q) \; > \; 0,
\end{displaymath}
and the communication is reliable.

In addition to the comparison (\ref{eqFirstDec}), the decoder keeps track of the distance of the highest metric average $D(T\circ {V\!\mathstrut}_{m} \,\|\, T \times Q)$
to the second highest one and compares this distance to some constant parameter $\Delta > 0$.
The decoder then sends reliably a bit of feedback, $F = 0$ or $1$,
to the transmitter
according to the following rule:
\begin{equation} \label{eqSecondDec}
D(T\circ {V\!\mathstrut}_{\widehat{m}} \,\|\, T \times Q) \; > \; D(T\circ {V\!\mathstrut}_{m} \,\|\, T \times Q) \, + \, \Delta, \;\;\;\;\;\; \forall \;\; m \, \neq \, \widehat{m},
\;\;\;\;\;\;\;\;\;
\Longleftrightarrow
\;\;\;\;\;\;\;\;\;
F \, = \, 1.
\end{equation}
In case $\widehat{m}$ does not satisfy (\ref{eqSecondDec}) or there does not even exist a unique $\widehat{m}$, i.e.
an $\widehat{m}$ strictly satisfying (\ref{eqFirstDec}), the decoder sends $\,F = 0\,$ (Fig.~\ref{fig1}).
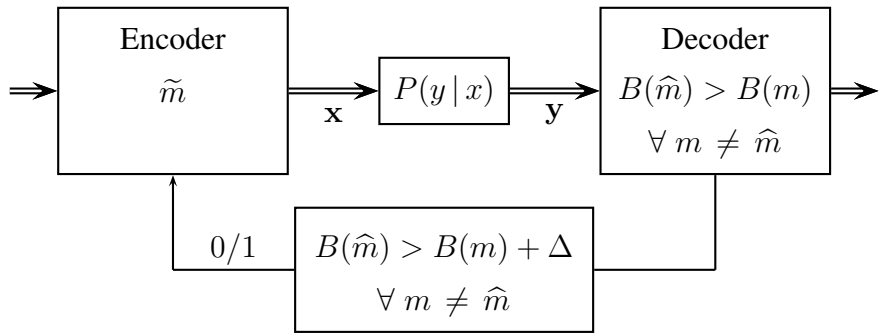
\begin{figure}[!t]
\psset{unit=.7mm}
\begin{center}
\begin{pspicture}(-20, 23)(145, 85)

\psframe(-11, 53)(33, 85)
\rput(11, 79){Encoder}
\rput(11, 69){$\widetilde{m}$}
\psline[doubleline = true]{->}(-20, 69)(-11, 69)
\psline[doubleline = true]{->}(33, 69)(50, 69)

\rput(41.5, 64.5){${\bf x}$}

\psframe(50, 62.5)(75, 75.5) \rput(62.5, 69){$P(y\,|\,x)$}

\rput(83.5, 64.5){${\bf y}$}

\psline[doubleline = true]{->}(136, 69)(145, 69)
\psline[doubleline = true]{->}(75, 69)(92, 69)

\psframe(92, 53)(136, 85)
\rput(114, 79){Decoder}
\rput(114, 69){$B(\widehat{m}) > B(m)$}
\rput(114, 59){$\forall \; m \, \neq \, \widehat{m}$}

\psframe(34, 23)(91, 47)

\rput(62.5, 39){$B(\widehat{m}) > B(m) + \Delta$}
\rput(62.5, 29){$\forall \; m \, \neq \, \widehat{m}$}

\psline{-}(114, 53)(114, 35)
\psline{-}(91, 35)(114, 35)

\psline{-}(11, 35)(34, 35)
\psline{<-}(11, 53)(11, 35)

\rput(22.5, 39){$0 / 1$}

\end{pspicture}
\end{center}
\caption{Channel with a $1$-bit feedback per block. $B(m) \triangleq D(T\circ {V}_{m} \,\|\, T \times Q)$,
where $T$ is the type of the received block, ${V}_{m}$ is the conditional type of the codeword for the message $m$, and $Q$ is the i.i.d. codebook generating distribution.}
\label{fig1}
\end{figure}
The procedure (\ref{eqSecondDec})
is equivalent to the decoding rule (\ref{eqDecRule}) with
$B(T\circ V) \equiv D(T\circ V \,\|\, T \times Q)$ and $A(T\circ V) \equiv D(T\circ V \,\|\, T \times Q) - \Delta$.
Let ${\cal E\mathstrut}^{c}$ denote the correct-decoding event in the random code ensemble~/~the channel in this scenario,
i.e. when $\widehat{m}$ is indeed the correct message and it does satisfy (\ref{eqSecondDec}).

\bigskip

\begin{lemma} [Natural selection exponent] \label{lemNaturalSelectExp}\newline
{\em If $\,R + \Delta \, < \, {R\mathstrut}_{\,-1}^{\,-}(Q)$, as defined in (\ref{eqRMinus}), then}
\begin{equation} \label{eqNatSelectExp}
\lim_{n\,\rightarrow \,\infty}\,\frac{\Pr\big\{{\cal E\mathstrut}^{c}\big\}}{-n}
\;\; = \;\;
{E\mathstrut}_{\!c}(R + \Delta,\, Q).
\end{equation}
\end{lemma}

\bigskip

\begin{proof}
Comparing the definition in (\ref{eqCorrAB}) with (\ref{eqExpNoTB}) we see that $\,{E\mathstrut}_{\!c}^{A}(R, Q) \equiv {E\mathstrut}_{\!c}(R + \Delta,\, Q)$.
Given that $\,R + \Delta \, < \, {R\mathstrut}_{\,-1}^{\,-}(Q)\,$ by the explicit expression for ${E\mathstrut}_{\!c}(R, Q)$ (\ref{eqExplicitSCD}) of Theorem~\ref{thm3}
we conclude that ${E\mathstrut}_{\!c}(R + \Delta,\, Q)$ is continuous at $R$ as a convex ($\cup$) function.
Provided this continuity, Theorem~\ref{thmCorr} then asserts that the exponent in the probability of ${\cal E\mathstrut}^{c}$ is given by ${E\mathstrut}_{\!c}(R + \Delta,\, Q)$.
\end{proof}

\bigskip

Note that the events $\{F = 1\}$ and ${\cal E\mathstrut}^{c}$ are not the same.
It can happen that an incorrect message $\widehat{m}$ satisfies (\ref{eqSecondDec}).
Then $\{F = 1\}$ will hold but ${\cal E\mathstrut}^{c}$ will not.
In order to ensure that the two events are the same with high probability,
we further assume that $\Delta$ is small enough so that
\begin{equation} \label{eqErrorgrtCorrect}
{E\mathstrut}_{\!e}(R, Q) \; > \; {E\mathstrut}_{\!c}(R + \Delta,\, Q).
\end{equation}
Under this condition\footnote{A better bound than ${E\mathstrut}_{\!e}(R, Q)$ would be the exponent of undetected error for the decoder (\ref{eqDecRule}) with $B(T\circ V) \equiv D(T\circ V \,\|\, T \times Q)$ and $A(T\circ V) \equiv D(T\circ V \,\|\, T \times Q) - \Delta$.}
and the condition of Lemma~\ref{lemNaturalSelectExp}, given $\{F = 1\}$ with high probability holds also ${\cal E\mathstrut}^{c}$.
Given the condition of Lemma~\ref{lemNaturalSelectExp}, ${E\mathstrut}_{\!c}(R + \Delta,\, Q)$ is the same as the correct-decoding exponent of the ML decoder ${E\mathstrut}_{\!c}^{ML}(R + \Delta,\, Q)$ according to Theorem~\ref{thm3}.
This situation is depicted in Fig.~\ref{graph12}. There on the left graph
$\,{E\mathstrut}_{\!e}(R, Q) > {E\mathstrut}_{\!c}(R + \Delta,\, Q) = 0\,$
while on the right graph $\,{E\mathstrut}_{\!e}(R, Q) > {E\mathstrut}_{\!c}(R + \Delta,\, Q) > 0$.

On the other hand, the probability exponent of the event $\{F = 0\}$ is given by ${E\mathstrut}_{\!e}(R + \Delta, \,Q)$ according to Corollary~\ref{cor1}.
Indeed, since ${E\mathstrut}_{\!e}(R, Q) > 0$ we have also $\,{E\mathstrut}_{\!e}(R, Q) > {E\mathstrut}_{\!e}(R + \Delta, \,Q)$,
then the exponent in the probability of undetected error in (\ref{eqSecondDec}) is also higher than ${E\mathstrut}_{\!e}(R + \Delta, \,Q)$.

In case $F = 1$, which is a rare event when ${E\mathstrut}_{\!c}(R + \Delta,\, Q) > 0$,
the system parameter $Q$ is updated.
A new codebook is adopted by both the encoder and the decoder according
to the type of the transmitted codeword ${\bf x\mathstrut}_{\widetilde{m}}\,$:
\begin{displaymath}
Q'(x) \; = \; {T\!\mathstrut}_{\widetilde{m}}(x) \; = \; {T\!\mathstrut}_{\widehat{m}}(x),
\end{displaymath}
where ${T\!\mathstrut}_{m}(x) = \sum_{y}T(y){V\!\mathstrut}_{m}(x \, | \, y)$.
Under the condition (\ref{eqErrorgrtCorrect}) and the condition of Lemma~\ref{lemNaturalSelectExp}, the type of the transmitted codeword is
known at the decoder with high probability also given the event $\{F = 1\}$.
In case the feedback $F = 0$, the codebook distribution $Q$ remains unchanged.
To summarize:
\begin{displaymath}
\begin{array}{c|c|c}
\text{Feedback} & \text{Encoder} & \text{Decoder} \\
& & \\[-1em]
\hline
& & \\[-1em]
F\mathstrut\,=\,1 & Q(x)\,\leftarrow \, {T\!\mathstrut}_{\widetilde{m}}(x) & Q(x)\,\leftarrow \, {T\!\mathstrut}_{\widehat{m}}(x) \\
& & \\[-1em]
\hline
& & \\[-1em]
F\,=\,0 & - & -
\end{array}
\end{displaymath}

\bigskip

\subsection{Natural type selection}\label{NatTS}

\bigskip

The joint type $T(y){V\!\mathstrut}_{\widetilde{m}}(x \, | \, y)$ of the transmitted and the received blocks
given the event $\{F = 1\}$ or ${\cal E\mathstrut}^{c}$
is related to the probability exponent of this event $\,{E\mathstrut}_{\!c}(R + \Delta,\, Q)$.

\bigskip

\begin{thm}[Natural type selection]\label{thmConvType}\newline
{\em If $\,R + \Delta \, < \, {R\mathstrut}_{\,-1}^{\,-}(Q)$, as defined in (\ref{eqRMinus}), then given the event ${\cal E\mathstrut}^{c}$ the joint type of the transmitted and the received words $({\bf X}, {\bf Y})$ converges in probability to the minimizing distribution of $\,{E\mathstrut}_{\!c}(R + \Delta,\, Q)$ (\ref{eqExpNoTB}).}
\end{thm}

\bigskip

\begin{proof}
By the preceding Lemma~\ref{lemNaturalSelectExp}
the exponent of ${\cal E\mathstrut}^{c}$ is given by ${E\mathstrut}_{\!c}(R + \Delta,\, Q)$.
Therefore by (\ref{eqExplicitSCD}) it is finite.
By the last assertion of Lemma~\ref{lemEps} for any $\epsilon > 0$ given the event
\begin{displaymath}
{\cal G} \;\;\; \triangleq \;\;\; \Big\{
({\bf X}, {\bf Y})
\;\;\; \text{of any type} \;\;\; T\circ V
\;\;\; \text{s.t.} \;\;\;
 D(T \circ V \, \| \, T \times Q) \, - \, \Delta
\; \leq \;  R \, - \, \epsilon\Big\}
\end{displaymath}
the exponent of ${\cal E\mathstrut}^{c}$ is $+\infty$.
Therefore, given ${\cal E\mathstrut}^{c}$ with high probability holds also ${\cal E\mathstrut}^{c} \cap\, {\cal G\mathstrut}^{c}$.

On the other hand, the exponent in the probability of the event
\begin{displaymath}
{\cal H} \;\;\; \triangleq \;\;\; \Big\{
({\bf X}, {\bf Y})
\;\;\; \text{of any type} \;\;\; T\circ V
\;\;\; \text{s.t.} \;\;\;
 D(T \circ V \, \| \, Q \circ P)
\; > \;  {E\mathstrut}_{\!c}(R + \Delta,\, Q) \, + \, \epsilon\Big\},
\end{displaymath}
is obviously lower-bounded by $\,{E\mathstrut}_{\!c}(R + \Delta,\, Q) + \epsilon$.
Then, given ${\cal E\mathstrut}^{c}$ with high probability holds also ${\cal E\mathstrut}^{c} \cap\, {\cal G\mathstrut}^{c} \,\cap\, {\cal H\mathstrut}^{c}$.

Now consider the joint type $T\circ V$ of $({\bf X}, {\bf Y})$ given $\,{\cal E\mathstrut}^{c} \cap\, {\cal G\mathstrut}^{c} \,\cap\, {\cal H\mathstrut}^{c}$.
Since $\,R + \Delta \, < \, {R\mathstrut}_{\,-1}^{\,-}(Q)\,$ by Theorem~\ref{thm3} there exists some $\beta \in (-1, 0]$ such that
\begin{displaymath}
{E\mathstrut}_{\!c}(R + \Delta,\, Q) \; = \; {E\mathstrut}_{0}(\beta, Q) - \beta (R + \Delta).
\end{displaymath}
We can use this $\beta$ to write
\begin{align}
{E\mathstrut}_{\!c}(R + \Delta,\, Q) \, + \, \epsilon
\;\; & \overset{{\cal H\mathstrut}^{c}}{\geq} \;\;
D(T \circ V \, \| \, Q \circ P)
\nonumber \\
& \overset{{\cal G\mathstrut}^{c}}{\geq} \;\;
D(T \circ V \, \| \, Q \circ P) \, - \, \beta
\underbrace{\big[
R + \Delta - \epsilon - D(T \circ V \, \| \, T \times Q)
\big]}_{\leq\,0}
\label{eqAddedEps} \\
& = \;\;
D(T \, \| \, {T\!\mathstrut}_{\beta})
\; + \;
(1 + \beta)D(T \circ V \, \| \, T \circ {V\!\!\mathstrut}_{\beta})
\; + \; \beta \epsilon
\; + \;
\underbrace{{E\mathstrut}_{0}(\beta, Q) \; - \; \beta (R + \Delta)}_{{E\mathstrut}_{\!c}(R \,+ \,\Delta,\; Q)}
\nonumber \\
& = \;\;
D(T \, \| \, {T\!\mathstrut}_{\beta})
\; + \;
(1 + \beta)D(T \circ V \, \| \, T \circ {V\!\!\mathstrut}_{\beta})
\; + \; \beta \epsilon
\; + \; {E\mathstrut}_{\!c}(R + \Delta,\, Q)
\nonumber \\
(1-\beta)\,\epsilon
\;\; & \geq \;\;
D(T \, \| \, {T\!\mathstrut}_{\beta})
\; + \;
(1 + \beta)D(T \circ V \, \| \, T \circ {V\!\!\mathstrut}_{\beta}),
\label{eqConverge}
\end{align}
where ${T\!\mathstrut}_{\beta}\circ{V\!\!\mathstrut}_{\beta}$
is the minimizing distribution of $\,{E\mathstrut}_{\!c}(R + \Delta,\, Q)$
determined according to Theorem~\ref{thm3} by (\ref{eqTrho})-(\ref{eqVrho}) with $\beta$.
The inequality (\ref{eqConverge}) implies that the type $T\circ V$ and the solution ${T\!\mathstrut}_{\beta}\circ{V\!\!\mathstrut}_{\beta}$
are close in ${\cal L}_{1}$ norm. And all this given $\,{\cal E\mathstrut}^{c} \cap\, {\cal G\mathstrut}^{c} \,\cap\, {\cal H\mathstrut}^{c}$, i.e. with high probability.
\end{proof}

\bigskip

In the subsequent analysis we assume that the blocklength $n$ is large and neglect the difference
between the random joint type of the transmitted and the received blocks $\,T(y){V\!\mathstrut}_{\widetilde{m}}(x \, | \, y)\,$
given $\{F = 1\}$
and the respective solution $\,\breve{T\mathstrut}(y) \breve{V\mathstrut}(x \, | \, y)\,$ to the minimization problem $\,{E\mathstrut}_{\!c}(R + \Delta,\, Q)\,$ (\ref{eqExpNoTB}) or $\,{E\mathstrut}_{\!c}^{ML}(R + \Delta,\, Q)\,$ (\ref{eqDefImplicitCor}).
We also assume that the inequality (\ref{eqErrorgrtCorrect}) between the error exponent and the correct-decoding exponent is never violated,
so that $\{F = 1\}$ is always exponentially equivalent in probability to ${\cal E\mathstrut}^{c}$.

Let ${Q\mathstrut}_{0}$ be the initial codebook distribution
and consider the consecutive events $\big\{{\cal E\mathstrut}^{c}_{\ell}\big\}_{\ell \, = \, 0}^{+\infty}$.
They result in the sequence of codebook distributions $\big\{{Q\mathstrut}_{\ell}\big\}_{\ell \, = \, 1}^{+\infty}$.
Suppose
\begin{displaymath}
C\big(\text{supp}({Q\mathstrut}_{0})\big) \; > \; R + \Delta.
\end{displaymath}
Then by Lemma~\ref{lemOneStep} holds also
$\,R + \Delta < {R\mathstrut}_{\,-1}^{\,-}({Q\mathstrut}_{0})$,
which is the condition of both Lemma~\ref{lemNaturalSelectExp}
and Theorem~\ref{thmConvType}.
As a result, given (\ref{eqErrorgrtCorrect}) for ${Q\mathstrut}_{0}$ the events $\{F = 1\}$ and ${\cal E\mathstrut}^{c}$
are equivalent and given these events the joint type of the transmitted and the received blocks
(approximately, with high probability) achieves the minimum $\,{E\mathstrut}_{\!c}(R + \Delta,\, {Q\mathstrut}_{0}) = {E\mathstrut}_{\!c}^{ML}(R + \Delta,\, {Q\mathstrut}_{0})$.
Therefore the next distribution ${Q\mathstrut}_{1}$ is obtained according to (\ref{eqUpdate}).
Finally the same Lemma~\ref{lemOneStep} gives $\,\text{supp}({Q\mathstrut}_{1}) = \text{supp}({Q\mathstrut}_{0})$.
Then, provided that (\ref{eqErrorgrtCorrect}) continues to hold for each ${Q\mathstrut}_{\ell}$, by induction we obtain that at each iteration $\ell$ the codebook distribution ${Q\mathstrut}_{\ell\,+\,1}$
evolves according to (\ref{eqUpdate}).
This results in convergence of $\,{E\mathstrut}_{\!c}^{ML}(R + \Delta,\, {Q\mathstrut}_{\ell})$.
Suppose the initial distribution ${Q\mathstrut}_{0}$ satisfies further the strict inequality (\ref{eqLowerThan})
with $R + \Delta$:
\begin{displaymath}
{E\mathstrut}_{\!c}^{ML}(R + \Delta, \,{Q\mathstrut}_{0}) \;\;\;\; < \;\;\;\; \min_{\substack{\\Q:\;\;C\left(\text{supp}(Q)\right) \; < \; R\,+\,\Delta}}\, {E\mathstrut}_{\!c}^{ML}(R + \Delta, \,Q).
\end{displaymath}
This is always possible given $\,C\big(\text{supp}({Q\mathstrut}_{0})\big) > R + \Delta$.
Then the sequence of $\,{E\mathstrut}_{\!c}^{ML}(R + \Delta,\, {Q\mathstrut}_{\ell})\,$ converges to zero by Lemma~\ref{lemConvergenceReg},
achieving our goal. In the limit of convergence of the codebook distribution for a given channel -- reliable communication is guaranteed for all rates below $R + \Delta$.

An example is shown in Fig.~\ref{graph12},~\ref{graph3}.
\begin{figure}
\centering
\begin{minipage}{.5\textwidth}
  \centering
  \includegraphics[width=1\textwidth]{errcorr1.eps}
\end{minipage}%
\begin{minipage}{.5\textwidth}
  \centering
  \includegraphics[width=1\textwidth]{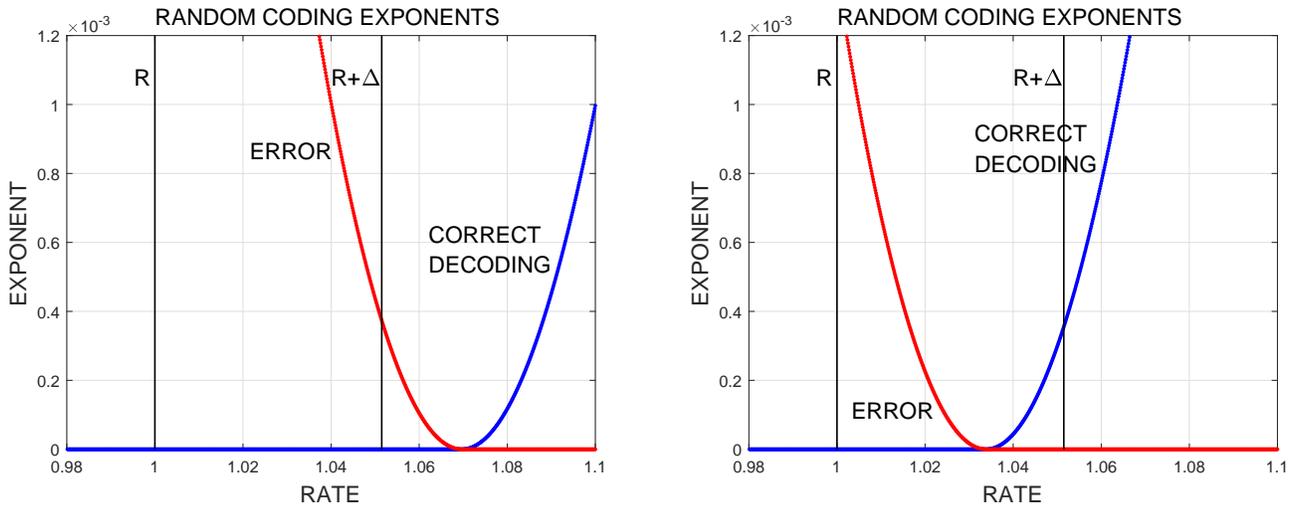}
\end{minipage}
\caption{The decreasing curve is the error exponent ${E}_{e}(\,\,\widetilde{\!\!R}, Q)$. The increasing curve is the correct-decoding exponent
${E}_{c}(\,\,\widetilde{\!\!R}, Q)$. Both graphs are for the same $Q(x)$. The channel $P(y\,|\,x)$ is different between the left graph and the right graph. In both cases ${E}_{e}(R, Q) > {E}_{c}(R + \Delta, \,Q)$.}
\label{graph12}
\end{figure}
In Fig.~\ref{graph12} on the left graph the correct-decoding exponent is zero at $R+\Delta$.
The rate of communication is lower -- at $R$.
Then the channel $P(y \, | \, x)$ changes abruptly and both the error exponent curve and the correct-decoding exponent curve for the same $Q(x)$ move to the left, as shown on the right graph of Fig.~\ref{graph12}.
Now the correct-decoding exponent becomes positive at $R + \Delta$,
but is still lower than the error exponent at $R$, so that the strict inequality (\ref{eqErrorgrtCorrect})
still holds.
The reliable communication continues at $R$.
The new channel $P(y \, | \, x)$ is assumed to remain the same during the ensuing iterations,
shown in Fig.~\ref{graph3}.
\begin{figure}[t]
\centering
\includegraphics[width=.50\textwidth]{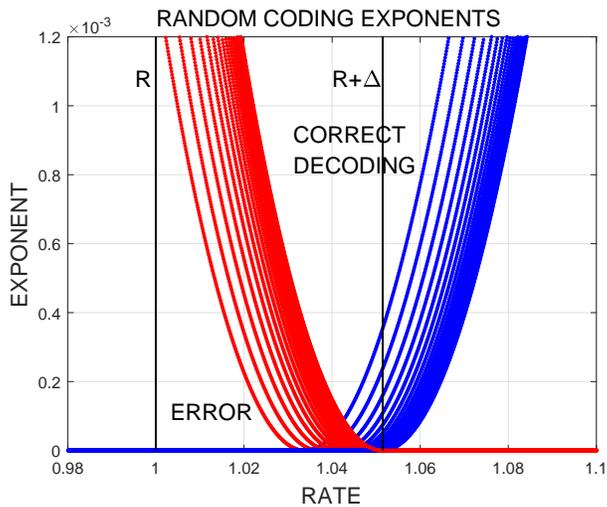}  
\caption{The decreasing curves are the error exponents ${E}_{e}(\,\,\widetilde{\!\!R}, {Q}_{\ell})$.
The increasing curves are the correct-decoding exponents ${E}_{c}(\,\,\widetilde{\!\!R}, {Q}_{\ell})$.
All the curves are for the same channel $P(y \, | \, x)$.
For each $\ell = 0$, $1$, $2$, ... , the respective pair of curves meets zero at the same point $\,\,\widetilde{\!\!R} = I({Q}_{\ell}\circ P)$.
For each $\ell$ holds $\,{E}_{e}(R, {Q}_{\ell}) > {E}_{c}(R + \Delta, \,{Q}_{\ell})\,$ (\ref{eqErrorgrtCorrect}).
The correct-decoding exponent ${E}_{c}(R + \Delta, \,{Q}_{\ell})$ converges to zero as $\ell$ grows.
At the same time the zero point of the error exponent at $\,\,\widetilde{\!\!R} = I({Q}_{\ell}\circ P)$ moves to the right towards $R + \Delta$.}
\label{graph3}
\end{figure}
During the iterations the codebook distribution adapts to the new channel.
In the limit of the iterations the correct-decoding exponent returns to zero at $R+\Delta$
with respect to the new channel.
In this way the adaptation scheme will safeguard the reliable communication mode
at $R$ for as long as the DMC capacity of the block doesn't deteriorate below $R+\Delta$.

In the presented example the change in the channel is abrupt relatively to the number of iterations required to adapt to the change.
In practice the change in the channel should be slow and the correct-decoding exponent at $R+\Delta$ should be near zero,
in order for the scheme to be able to follow the changes in the channel successfully.

The presented adaptation scheme tries to parallel the natural type selection scheme in lossy source-encoding \cite{ZamirRose01}
in that it is embedded in the structure of a specific channel-decoding procedure itself through the confidence parameter $\Delta$.
Maintaining or restoring this confidence means adaptation.
A closely related alternative scheme is presented next.

\bigskip

\subsection{Alternative scheme}\label{AlternativeScheme}

\bigskip

Here we assume any sequence (in the blocklength $n$) of reliable decoders
with asymptotic error exponent higher than ${E\mathstrut}_{\!c}(R + \Delta,\, Q)$
as in (\ref{eqErrorgrtCorrect}).
The decoder determines its estimate $\widehat{m}$ of the transmitted message
and compares the metric average (\ref{eqMMIlikeDec}) of the corresponding codeword
${\bf x\mathstrut}_{\widehat{m}}$ to $R + \Delta$:
\begin{equation} \label{eqUpdateRule}
D(T\circ {V\!\mathstrut}_{\widehat{m}} \,\|\, T \times Q) \; > \; R \, + \, \Delta,
\;\;\;\;\;\;\;\;\;
\Longleftrightarrow
\;\;\;\;\;\;\;\;\;
F \, = \, 1.
\end{equation}
This single comparison is performed at the receiver side after the decoding is over (Fig.~\ref{fig2}) and it replaces the condition (\ref{eqSecondDec}).
\begin{figure}[!t]
\psset{unit=.7mm}
\begin{center}
\begin{pspicture}(-20, 23)(145, 85)

\psframe(-11, 53)(33, 85)
\rput(11, 79){Encoder}
\rput(11, 69){$\widetilde{m}$}
\psline[doubleline = true]{->}(-20, 69)(-11, 69)
\psline[doubleline = true]{->}(33, 69)(50, 69)

\rput(41.5, 64.5){${\bf x}$}

\psframe(50, 62.5)(75, 75.5) \rput(62.5, 69){$P(y\,|\,x)$}

\rput(83.5, 64.5){${\bf y}$}

\psline[doubleline = true]{->}(136, 69)(145, 69)
\psline[doubleline = true]{->}(75, 69)(92, 69)

\psframe(92, 53)(136, 85)
\rput(114, 79){Decoder}
\rput(114, 69){$\widehat{m}$}

\psframe(34, 23)(91, 47)

\rput(62.5, 35){$B(\widehat{m}) 
>
R + \Delta$}

\psline{-}(114, 53)(114, 35)
\psline{-}(91, 35)(114, 35)

\psline{-}(11, 35)(34, 35)
\psline{<-}(11, 53)(11, 35)

\rput(22.5, 39){$0 / 1$}

\end{pspicture}
\end{center}
\caption{An alternative scheme. The decoder providing $\widehat{m}$ is not specified. $B(\widehat{m}) \triangleq D(T\circ {V}_{\widehat{m}} \,\|\, T \times Q)$,
where $T$ is the type of the received block, ${V}_{\widehat{m}}$ is the conditional type of the codeword for the estimated message $\widehat{m}$, and $Q$ is the i.i.d. codebook generating distribution.}
\label{fig2}
\end{figure}
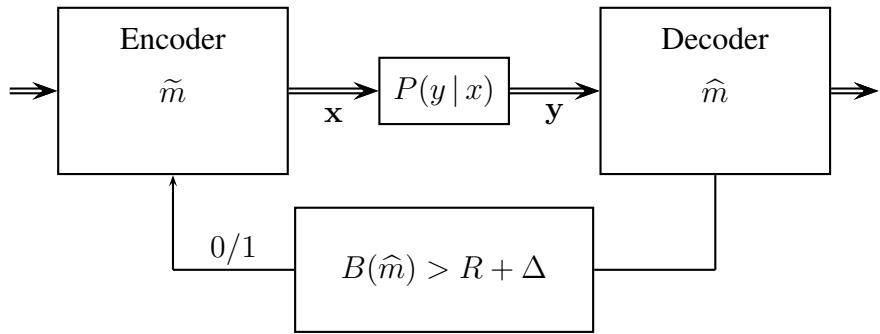

In this case also, the exponent of the event $\{F = 1\}$ is given by $\,{E\mathstrut}_{\!c}(R + \Delta,\, Q)$.
Let $({\bf X}, {\bf Y})$ 
denote
the pair of the transmitted and the received blocks.
Define an event:
\begin{equation} \label{eqSelEvent}
{\cal S} \;\; \triangleq \;\; \Big\{({\bf X}, {\bf Y})
\;\;\; \text{of any type} \;\;\; T\circ V
\;\;\; \text{s.t.} \;\;\;
D(T \circ V \, \| \, T \times Q)
\; > \; R + \Delta \Big\}.
\end{equation}
This is similar to ${\cal E\mathstrut}^{c}$ in the previous setting.

\bigskip

\begin{lemma} [Selection exponent] \label{lemSelectExp}\newline
{\em If $\,R + \Delta \, < \, {R\mathstrut}_{\,-1}^{\,-}(Q)$, as defined in (\ref{eqRMinus}), then}
\begin{equation} \label{eqSelectExp}
\lim_{n\,\rightarrow \,\infty}\,\frac{\Pr\,\{{\cal S}\}}{-n}
\;\; = \;\;
{E\mathstrut}_{\!c}(R + \Delta,\, Q).
\end{equation}
\end{lemma}

\begin{proof}
Similar to the proof of Theorem~\ref{thmCorr}, with the lower bound obtained directly without $\epsilon$, as in Sanov's theorem.
\end{proof}

\bigskip

If the error exponent of the sequence of employed decoders is higher than the exponent of the event ${\cal S}$
given by ${E\mathstrut}_{\!c}(R + \Delta,\, Q)$,
then given ${\cal S}$ with high probability holds also $\{F = 1\}$.
And vice versa, given $\{F = 1\}$ with high probability holds ${\cal S}$.

In case $F = 1$ the system parameter $Q$ is updated as before and
similarly evolves according to (\ref{eqUpdate})
for sufficiently large $n$ when the joint type has converged.

\bigskip

\begin{thm}[Convergence of a type]\label{thmConvType2}\newline
{\em If $\,R + \Delta \, < \, {R\mathstrut}_{\,-1}^{\,-}(Q)$, as defined in (\ref{eqRMinus}), then given the event ${\cal S}$ (\ref{eqSelEvent}) the joint type of the transmitted and the received words $({\bf X}, {\bf Y})$ converges in probability to the minimizing distribution of $\,{E\mathstrut}_{\!c}(R + \Delta,\, Q)$ (\ref{eqExpNoTB}).}
\end{thm}

\begin{proof}
Analogous to the proof of Theorem~\ref{thmConvType}, without $\epsilon$ in (\ref{eqAddedEps}).
\end{proof}

\bigskip

All the rest is the same. The advantage of the alternative scheme is that arbitrary sufficiently reliable decoder is allowed.

\bigskip

\subsection{Other metrics}\label{Other}

\bigskip

Other suboptimal variants of the adaptation scheme are possible with the metrics in (\ref{eqFamily})
used for the metric average $A$
in a pair with (\ref{eqMMIlikeDec}) which is $B$.
Then the feedback bit is determined using $A$ and $B$
in the scheme of Fig.~\ref{fig1}:
\begin{displaymath}
A(T\circ {V\!\mathstrut}_{\widehat{m}}) \; > \; B(T\circ {V\!\mathstrut}_{m}) \, + \, \Delta, \;\;\;\;\;\; \forall \;\; m \, \neq \, \widehat{m},
\;\;\;\;\;\;\;\;\;
\Longleftrightarrow
\;\;\;\;\;\;\;\;\;
F \, = \, 1,
\end{displaymath}
or using only $A$ in the alternative scheme of Fig.~\ref{fig2} with an unspecified decoder \cite{TridenskiZamir18}:
\begin{displaymath}
A(T\circ {V\!\mathstrut}_{\widehat{m}}) \; > \; R \, + \, \Delta,
\;\;\;\;\;\;\;\;\;
\Longleftrightarrow
\;\;\;\;\;\;\;\;\;
F \, = \, 1.
\end{displaymath}
In particular, the variant with the last metric in (\ref{eqFamily}), which uses the knowledge of the channel,
doesn't seem to require a condition for convergence like (\ref{eqLowerThan}) since it falls into the family of alternating minimization
procedures of Csisz\'ar and Tusn\'ady \cite{CsiszarTusnady84}.

\bigskip

\section{Conclusion}\label{Conc}

\bigskip

In this work we introduce two different expressions for the optimal correct-decoding exponent (\ref{eqOptimum}):
\begin{align}
&
\min_{\substack{\\Q(x)}}
\;\;\;\;\;\;
\min_{\substack{\\T(y), \, V(x\,|\,y)
}}
\;\;\;\;\,
\left\{
D(T\circ V \, \| \, Q \circ P) \, + \,
\big| R - D(T\circ V \,\|\, T \times Q)\big|^{+}
\right\}
\label{eqExpressionA} \\
\equiv \;\;
&
\min_{\substack{\\Q(x)}}
\;
\min_{\substack{\\T(y), \, V(x\,|\,y):\\
\\
D(T\,\circ\, V \,\|\, T\, \times\, Q)\;\geq\;
R
}}
\Big\{D(T\circ V \, \| \, Q \circ P)\Big\},
\label{eqExpressionB}
\end{align}
as alternatives to the Dueck-K{\"o}rner expression \cite{DueckKorner79}.
We show that the inner minimum in (\ref{eqExpressionA}) has a meaning of the 
correct-decoding exponent of the ML decoder for a given i.i.d. codebook distribution $Q$.
We propose a minimization procedure over $Q$ at constant $R$ which uses the inner minimum in (\ref{eqExpressionA}) iteratively.
It is shown that this procedure results in a sequence of distributions ${Q\mathstrut}_{\ell}$ with a monotonically non-increasing sequence of the corresponding inner minima in (\ref{eqExpressionA}).
This sequence of minima eventually converges to the double minimum (\ref{eqExpressionA}) over some subset of the channel input alphabet,
and more precisely -- over some subset of the support $\text{supp}({Q\mathstrut}_{0})$ of the initial distribution ${Q\mathstrut}_{0}$.
In general, it remains unclear whether the minimization procedure at constant $R$ does always achieve the global minimum (\ref{eqExpressionA}) over
the initial channel input support $\text{supp}({Q\mathstrut}_{0})$ itself.

From the practical standpoint, it is interesting when the correct-decoding exponent is zero. This is when the reliable communication begins.
For any rate $R$ below the capacity,
we provide a minimal and quite obvious sufficient condition (\ref{eqLowerThan}) on the initial distribution ${Q\mathstrut}_{0}$ which guarantees convergence of the minimization procedure to zero.
Still, this meager sufficient condition presents an inner bound on the region of convergence of the fixed-rate computation algorithm in terms of ${Q\mathstrut}_{0}$ for each rate below the capacity.
This ``computation of zero'' is interesting of course only because of the unknown set of the zero-exponent achieving distributions $Q$.

We show that the inner minimum in (\ref{eqExpressionB}) in turn has a meaning of the exponent in the strict correct decoding with
the channel-independent metric (\ref{eqNewMetric})
for a given i.i.d. codebook distribution $Q$.
The inner minima in (\ref{eqExpressionA}) and (\ref{eqExpressionB})
coincide as increasing functions of $R$ for slopes less than $1$.
This coincidence allows us to give a stochastic interpretation to the fixed-rate minimization procedure
and to propose a scheme for the channel input adaptation (Fig.~\ref{fig1}).
The scheme does not rely on the knowledge of the channel.
In this scheme the communication occurs at a rate $R$ and is assumed sufficiently reliable from the get-go.
Then, in the limit of large block length the adaptation falls exactly into the steps of the iterative minimization procedure.
As a result, under the initial condition (\ref{eqLowerThan}) for $R + \Delta$ the correct-decoding exponent of the decoder gradually descents to zero at $R + \Delta$,
thereby securing the reliable communication mode at $R$.

The adaptation scheme uses a single bit of feedback per transmitted block.
According to this bit the system decides whether to replace the codebook distribution $Q$ with the empirical distribution of the last sent codeword or not.
In practice, a less interesting case would be when the feedback bit has entropy zero,
i.e. when the feedback bit is $1$ or $0$ with high probability.
The first situation occurs when the ML correct-decoding exponent for a given $Q$ meets zero at a rate higher than $R+\Delta$ (Fig.~\ref{graph12}, left).
Then it follows from Corollary~\ref{cor1} that the feedback bit is $1$ with high probability.
In this case
there is no clear advantage of the selected empirical distribution over $Q$
and its constant replacement is not vital.
The second situation happens when the correct decoding exponent is substantially positive at $R+ \Delta$ (Fig.~\ref{graph12}, right).
In this case it naturally takes an exponentially large number of blocks to obtain a single adaptation step.

Therefore, the promising case seems to be in the transition zone, when the feedback bit has a non-zero entropy.
This is the situation when the correct-decoding exponent meets zero at $R+\Delta$
and fluctuates there, at a finite block length, heaving upwards following the changes in the channel and falling back to zero in the result of the adaptation process.
For such fluctuations the sufficient condition (\ref{eqLowerThan}) is adequate and enough,
because it stays satisfied.
The question however remains -- how slow and how large, respectively, the change in the channel
and the size of the block have to be in order for the adaptation mechanism to follow the channel from block to block.

We presented also a fixed-slope version of the algorithm (Section~\ref{FixedSlope}) which always converges in the support of ${Q\mathstrut}_{0}$
without any additional conditions,
just like the Arimoto algorithm \cite{Arimoto76}.
It remains a question for further research -- if this fixed-slope algorithm can also be translated into some adaptation scheme,
as we have done here for the fixed-rate version.

\bigskip

\section*{Appendix}\label{App}

\bigskip

{\em Proof of Lemma~\ref{lemTypes}:}\newline
Let $T(y)V(x\,|\,y)$ be the joint type of the received and the transmitted blocks.
The exponent in the probability (after $-n$) of this joint type is
\begin{equation} \label{eqJoint}
D(T\circ V \, \| \, Q \circ P) \, + \, o(1),
\end{equation}
where the diminishing term $o(1)$ is uniform with respect to $T\circ V$.

Suppose a message $m$ is sent and consider a different message $m'\neq m$ in the codebook.
Consider the event that the random codeword ${\bf X\mathstrut}_{m'}$, corresponding to the message $m'$, has a conditional type $\widehat{V}(x\,|\,y)$
given the received vector ${\bf y}$ of the type $T(y)$:
\begin{equation} \label{eqEventDef}
{\cal T\mathstrut}_{m'}\big(\widehat{V}\,|\, {\bf y}\big) \;\; \triangleq \;\; \left\{\;{\bf X\mathstrut}_{m'} \;\;\text{of type}\;\; \widehat{V} \;\; \text{w.r.t.} \;\; {\bf y} \;\right\}.
\end{equation}
The exponent in the probability of this event is given by
\begin{displaymath}
{D\mathstrut}_{n} \; = \; D(T\circ \widehat{V} \,\|\, T \times Q) \, + \, o(1).
\end{displaymath}
For convenience, we briefly denote this exponent as ${D\mathstrut}_{n}$ and the exponent in the codebook size (after $n$) as ${R\mathstrut}_{n} \triangleq \frac{1}{n}\log \Big\lceil {e\mathstrut}^{nR}\Big\rceil$.
The larger the blocklength $n$, the closer these quantities are to $D(T\circ \widehat{V} \,\|\, T \times Q)$ and $R$, respectively, uniformly with respect to the joint type $T\circ \widehat{V}$.
Consider the event when the conditional type $\widehat{V}(x\,|\,y)$ appears somewhere among the $\Big\lceil {e\mathstrut}^{nR}\Big\rceil - 1$
incorrect codewords:
\begin{displaymath}
\bigcup_{m'\,\neq \, m}{\cal T\mathstrut}_{m'}\big(\widehat{V}\,|\, {\bf y}\big) \;\; = \;\; \left\{\;\exists\, m'\neq m : \;\; {\bf X\mathstrut}_{m'} \;\;\text{of type}\;\; \widehat{V} \;\; \text{w.r.t.} \;\; {\bf y} \;\right\}.
\end{displaymath}
Using the union bound, we can upper-bound the probability of this event as
\begin{equation} \label{eqUB}
\Pr\left\{\,\bigcup_{m'\,\neq \, m}{\cal T\mathstrut}_{m'}\big(\widehat{V}\,|\, {\bf y}\big)\right\}
\;\leq\;
\min\Big\{1, \, \underbrace{{e\mathstrut}^{-n({D}_{n}\, - \, {R}_{n})}}_{\text{UB}}\Big\}
\; = \; {e\mathstrut}^{-n\big|{D}_{n}\, - \, {R}_{n}\big|^{+}}.
\end{equation}
For the lower bound, we prepare two alternative bounds:
\begin{align}
\Pr\left\{\,\bigcup_{m'\,\neq \, m}{\cal T\mathstrut}_{m'}\big(\widehat{V}\,|\, {\bf y}\big)\right\}
\;\; & \geq \;\;
\sum_{m'\,\neq \, m}\Pr\left\{{\cal T\mathstrut}_{m'}\big(\widehat{V}\,|\, {\bf y}\big)\right\}\,\cdot
\prod_{\substack{m''\,\neq\,m\\\,\,m''\,\neq\,m'}}\Pr\Big\{\underbrace{{\cal T\mathstrut}_{m''}^{\,c}\big(\widehat{V}\,|\, {\bf y}\big)}_{\substack{\text{complementary}\\\text{event}}}\Big\}
\nonumber \\
& = \;\;
\Big({e\mathstrut}^{n {R}_{n}}-1\Big){e\mathstrut}^{-n {D}_{n}}\Big(1 - {e\mathstrut}^{-n {D}_{n}}\Big)^{{e\mathstrut}^{n {R}_{n}} \, - \, 2}
\nonumber \\
& = \;\;
{e\mathstrut}^{-n ({D}_{n}\,-\,{R}_{n})} \;\;\Big(1 - {e\mathstrut}^{-n {R}_{n}}\Big)\Big(1 - {e\mathstrut}^{-n {D}_{n}}\Big)^{{e\mathstrut}^{n {R}_{n}} \, - \, 2}
\nonumber \\
& \geq \;\;
{e\mathstrut}^{-n \big|{D}_{n}\,-\,{R}_{n}\big|^{+}} \Big(1 - {e\mathstrut}^{-n {R}_{n}}\Big)\Big(1 - {e\mathstrut}^{-n {D}_{n}}\Big)^{{e\mathstrut}^{n {R}_{n}} \, - \, 2},
\label{eqAltOne} \\
\Pr\left\{\,\bigcup_{m'\,\neq \, m}{\cal T\mathstrut}_{m'}\big(\widehat{V}\,|\, {\bf y}\big)\right\}
\;\; & = \;\;
1 \, - \, \prod_{m'\,\neq\,m}\Pr\left\{{\cal T\mathstrut}_{m'}^{\,c}\big(\widehat{V}\,|\, {\bf y}\big)\right\}
\nonumber \\
\;\; & = \;\;
1 \, - \, \Big(1 - {e\mathstrut}^{-n {D}_{n}}\Big)^{{e\mathstrut}^{n {R}_{n}} \, - \, 1}
\nonumber \\
\;\; & \geq \;\;
\underbrace{{e\mathstrut}^{-n \big|{D}_{n}\,-\,{R}_{n}\big|^{+}}}_{\leq \; 1}
\left[1 \, - \, \Big(1 - {e\mathstrut}^{-n {D}_{n}}\Big)^{{e\mathstrut}^{n {R}_{n}}\, - \, 1}\right].
\label{eqAltTwo}
\end{align}
Combining (\ref{eqAltOne}) and (\ref{eqAltTwo}) together, we obtain an exponentially tight lower bound:
\begin{align}
\Pr\left\{\,\bigcup_{m'\,\neq \, m}{\cal T\mathstrut}_{m'}\big(\widehat{V}\,|\, {\bf y}\big)\right\}
\;\; & \geq \;\;
{e\mathstrut}^{-n \big|{D}_{n}\,-\,{R}_{n}\big|^{+}}\cdot
\left\{
\begin{array}{l l}
\Big(1 - {e\mathstrut}^{-n {R}_{n}}\Big)\Big(1 - {e\mathstrut}^{-n {D}_{n}}\Big)^{{e\mathstrut}^{n {R}_{n}} \, - \, 2}, & {D\mathstrut}_{n} \, \geq \, {R\mathstrut}_{n} \\
\;\;\;\;\;\;\;\;\;\;\;\;
1 \, - \, \Big(1 - {e\mathstrut}^{-n {D}_{n}}\Big)^{{e\mathstrut}^{n {R}_{n}}\, - \, 1}, & {D\mathstrut}_{n} \, < \, {R\mathstrut}_{n}
\end{array}
\right.
\nonumber \\
& \geq \;\;
{e\mathstrut}^{-n \big|{D}_{n}\,-\,{R}_{n}\big|^{+}}\cdot
\left\{
\begin{array}{l l}
\Big(1 - {e\mathstrut}^{-n {R}_{n}}\Big)\Big(1 - {e\mathstrut}^{-n {R}_{n}}\Big)^{{e\mathstrut}^{n {R}_{n}} \, - \, 2}, & {D\mathstrut}_{n} \, \geq \, {R\mathstrut}_{n} \\
\;\;\;\;\;\;\;\;\;\;\;\;
1 \, - \, \Big(1 - {e\mathstrut}^{-n {R}_{n}}\Big)^{{e\mathstrut}^{n {R}_{n}}\, - \, 1}, & {D\mathstrut}_{n} \, < \, {R\mathstrut}_{n}
\end{array}
\right.
\label{eqReplace} \\
& = \;\;
{e\mathstrut}^{-n \big|{D}_{n}\,-\,{R}_{n}\big|^{+}}\cdot
\left\{
\begin{array}{l l}
\;\;\;\;\;\;\;\;\;\;\;\;\;\;\;\;\;\;\;
\Big(1 - {e\mathstrut}^{-n {R}_{n}}\Big)^{{e\mathstrut}^{n {R}_{n}} \, - \, 1}, & {D\mathstrut}_{n} \, \geq \, {R\mathstrut}_{n} \\
\;\;\;\;\;\;\;\;\;\;\;\;
1 \, - \, \Big(1 - {e\mathstrut}^{-n {R}_{n}}\Big)^{{e\mathstrut}^{n {R}_{n}}\, - \, 1}, & {D\mathstrut}_{n} \, < \, {R\mathstrut}_{n}
\end{array}
\right.
\nonumber \\
& \geq \;\;
{e\mathstrut}^{-n \big|{D}_{n}\,-\,{R}_{n}\big|^{+}}\cdot
\,\min\bigg\{\underbrace{\Big(1 - {e\mathstrut}^{-n {R}_{n}}\Big)^{{e\mathstrut}^{n {R}_{n}} \, - \, 1}}_{\rightarrow \; 1/e}, \,
\underbrace{1 \, - \, \Big(1 - {e\mathstrut}^{-n {R}_{n}}\Big)^{{e\mathstrut}^{n {R}_{n}}\, - \, 1}}_{\rightarrow\; 1 \, - \, 1/e}
\bigg\},
\label{eqLB}
\end{align}
where ${D\mathstrut}_{n}$ is replaced with ${R\mathstrut}_{n}$ in (\ref{eqReplace}).
From all this, comparing the upper and lower bounds (\ref{eqUB}) and (\ref{eqLB}),
we only conclude that the exponent in the probability of the appearance of $\widehat{V}(x\,|\,y)$ among incorrect codewords
is given by
\begin{equation} \label{eqCombinedExp}
\big|D(T\circ \widehat{V} \,\|\, T \times Q) \, - \, R\,\big|^{+} \, + \, o(1),
\end{equation}
where the diminishing term $o(1)$ is uniform with respect to $T\circ \widehat{V}$.

Consider now the condition 
(\ref{eqDecRule}).
Given the joint type of the received and the transmitted blocks $\,T \circ V$,
if there exists an incorrect codeword of a conditional type $\widehat{V}$ such that
\begin{equation} \label{eqConstraint}
A(T\circ V) \; \leq \; B(T\circ \widehat{V}), 
\end{equation}
then the sent message $m$ does not satisfy (\ref{eqDecRule}).
Since there is only a polynomial number (in $n$) of different possible types,
adding together the exponents of $\,T \circ V$, (\ref{eqJoint}), and of $\widehat{V}$, (\ref{eqCombinedExp}), and minimizing
their sum subject to the constraint (\ref{eqConstraint}),
we obtain the exponent in the probability of this event
as given by (\ref{eqLemma}).
The additive diminishing term $o(1)$ can be put conveniently outside the minimum,
because of the uniformity of the corresponding term in (\ref{eqJoint}) and (\ref{eqCombinedExp}) with respect to various types $T$, $V$, and $\widehat{V}$. $\square$

\bigskip

{\em Proof of Lemma~\ref{lemEps}:}\newline
Suppose the message $m$ is sent and the pair of the transmitted and the received words $({\bf x}, {\bf y})$
has a joint type $T(y)V(x \, | \, y)$. Let ${\cal T\mathstrut}_{m'}\big(\widehat{V}\,|\, {\bf y}\big)$
denote the event (\ref{eqEventDef})
that the random codeword ${\bf X\mathstrut}_{m'}$, corresponding to the message $m'\neq m$, has a conditional type $\widehat{V}(x\,|\,y)$
given the received vector ${\bf y}$.
Define also ${R\mathstrut}_{n} \triangleq \frac{1}{n}\log \left(\Big\lceil {e\mathstrut}^{nR}\Big\rceil - 1\right)$ and
a quantity close to $A(T \circ V)$:
\begin{align}
{A\mathstrut}_{n} \;\; & \triangleq \;\; \min_{\substack{\\\text{types}\;\widehat{V}(x\,|\,y):\\ \\A(T \,\circ\, V)\; \leq \; D(T \,\circ \,\widehat{V} \, \| \, T \,\times\, Q)}}
D(T \circ \widehat{V} \, \| \, T \times Q) \; + \;
\frac{|{\cal X}||{\cal Y}|\log(n + 1)}{n}.
\nonumber 
\end{align}
Both quantities ${R\mathstrut}_{n}$ and ${A\mathstrut}_{n}$ converge respectively to $R$ and $A(T \circ V)$,
as $n$ grows, uniformly in $T\circ V$.
We can set an upper and a lower bounds on the conditional probability of correct decoding given $({\bf x}, {\bf y})$:
\begin{align}
\Pr \Big\{\,\text{correct decoding} \; \Big| \; ({\bf X\mathstrut}_{m}, {\bf Y}) = ({\bf x}, {\bf y})\,\Big\}
\;\; & \leq \;
\min_{\substack{\\\text{types}\;\widehat{V}(x\,|\,y):\\ \\A(T \,\circ\, V)\; \leq \; D(T \,\circ \,\widehat{V} \, \| \, T \,\times\, Q)}}
\prod_{m'\,\neq \, m}\left(1 - \Pr\Big\{{\cal T\mathstrut}_{m'}\big(\widehat{V}\,|\, {\bf y}\big)\Big\}\right)
\nonumber \\
& \leq \;\;
\left(1 - {e\mathstrut}^{-n {A}_{n}}\right)^{{e\mathstrut}^{n R_n}}
\nonumber \\
& = \;\;
\bigg[\underbrace{\left(1 - {e\mathstrut}^{-n {A}_{n}}\right)^{-{e\mathstrut}^{n {A}_{n}}}}_{>\; e}\bigg]^{-{e\mathstrut}^{-n {A}_{n}}\cdot\,{e\mathstrut}^{n R_n}}
\nonumber \\
& \overset{(*)}{<} \;\;
\exp\Big\{-{e\mathstrut}^{n (R_n \, - \, {A}_{n})}\Big\},
\label{eqUpperBCorDec}
\end{align}
where
$(*)$ holds because $\;(1 - x)^{-1/x} \, > \, e\;$ for $\;0\,<\,x\,<\,1$.
\begin{align}
& \Pr \Big\{\,\text{correct decoding} \; \Big| \; ({\bf X\mathstrut}_{m}, {\bf Y}) = ({\bf x}, {\bf y})\,\Big\}
\nonumber \\
& \geq \;\;
1 \; - \; \underbrace{\underbrace{(n + 1)^{|{\cal X}| \, |{\cal Y}|}}_{\geq \; \text{\# of types}}\,\cdot
\max_{\substack{\\\text{types}\;\widehat{V}(x\,|\,y):\\ \\A(T \,\circ\, V)\; \leq \; D(T \,\circ \,\widehat{V} \, \| \, T \,\times\, Q)}}
\sum_{m'\,\neq\,m}\Pr\Big\{{\cal T\mathstrut}_{m'}\big(\widehat{V}\,|\, {\bf y}\big)\Big\}}_{\text{UB on the probability that}\;
A(T \,\circ\, V)\; \leq \; D(T \,\circ \,\widehat{V}_{m'} \, \| \, T \,\times\, Q)\;\text{for some} \; m'\,\neq\,m}
\nonumber \\
& \geq \;\;1 \; - \; (n + 1)^{|{\cal X}| \, |{\cal Y}|}\cdot {e\mathstrut}^{-n (A(T\,\circ \,V) \, - \, R_n)}.
\label{eqLBCorDec}
\end{align}
From the upper bound (\ref{eqUpperBCorDec})
we see that if the joint type $T\circ V$ is such that $A(T\circ V) < R$,
then the conditional probability of correct decoding tends to zero super-exponentially as $n$ increases.
Consequently, those types drop out of the asymptotic exponent.
In particular, if $R > \max_{\,T\,\circ\, V}A(T\circ V)$,
then the exponent is $+\infty$
and for any $\epsilon > 0$
the last assertion of the lemma holds.
By the same token the lower bound (\ref{eqLboundEps}) follows.
On the other side, if $A(T\circ V) > R$, then
the lower bound (\ref{eqLBCorDec})
shows that the conditional probability of correct decoding tends to $1$.
This gives (\ref{eqUboundEps}). $\square$

\bigskip

{\em Proof of Theorem~\ref{thmCorrML}:}\newline
The idea is to compare the correct-decoding exponents of two different decoders.
One is an optimal decoder with a helper,
which must give the exponent at least as good (low) as the ML decoder,
and the second one is a suboptimal decoder.

Suppose a genie tells the decoder what is the joint type $T(y)V(x \, | \, y)$
of the received and the transmitted blocks.
Since the ML metric average (\ref{eqML}) depends only on the joint type of ${\bf y}$ and ${\bf x\mathstrut}_{m}$,
the best the decoder can do with this information is to choose at random one of the indices of the codewords
with the same conditional type $V$ with respect to the received word ${\bf y}$.
Without loss of generality we assume that the decoder chooses a codeword
according to the uniform distribution over all codewords of the same conditional type $V$ with respect to ${\bf y}$ in the codebook.
This will result in at least as good the correct-decoding exponent as the optimal (ML) decoder without a genie, or better.

For convenience, let us denote the exponent in the probability (after $-n$) of an independently generated codeword being of the conditional type $V$ with respect to ${\bf y}$ as
${D\mathstrut}_{n} = D(T\circ V \,\|\, T \times Q) + o(1)$,
and the exponent in the number of $\Big\lceil {e\mathstrut}^{nR}\Big\rceil - 1$ incorrect codewords (after $n$) as ${R\mathstrut}_{n} \triangleq \frac{1}{n}\log \left(\Big\lceil {e\mathstrut}^{nR}\Big\rceil - 1\right)$.
Let $N$ be the random number of incorrect codewords of the conditional type $V$ with respect to ${\bf y}$ in the codebook.
Then the conditional probability of correct decoding is given by
\begin{displaymath}
\mathbb{E} \left[\frac{1}{N + 1}\right].
\end{displaymath}
We use Chebyshev's inequality to upper-bound this probability:
\begin{align}
\mathbb{E} \left[\frac{1}{N + 1}\right] \; & \leq \;
\underbrace{\Pr\left\{N \geq \tfrac{1}{2}\mathbb{E}[N]\right\}}_{\leq \; 1}\,\cdot\, \frac{1}{\frac{1}{2}\mathbb{E}[N] + 1}
\, + \,
\Pr\left\{N < \tfrac{1}{2}\mathbb{E}[N]\right\}\cdot 1
\nonumber \\
& = \;
\;\;\;\;\;\;\;\;\;\;\;\;\;\;\;\;\;\;\;\;\;\;\;\;\;\;\;\;\;
\frac{1}{\frac{1}{2}\mathbb{E}[N] + 1}
\, + \,
\Pr\left\{N < \tfrac{1}{2}\mathbb{E}[N]\right\}
\nonumber \\
& \leq \;
\;\;\;\;\;\;\;\;\;\;\;\;\;\;\;\;\;\;\;\;\;\;\;\;\;\;\;\;\;
\frac{1}{\frac{1}{2}\mathbb{E}[N] + 1}
\, + \,
\underbrace{\frac{\mathbb{E}\big[(N - \mathbb{E}[N])^2\big]}{\frac{1}{4}\mathbb{E}^{2}[N]}}_{\text{Chebyshev}}
\nonumber \\
& = \;
\;\;\;\;\;\;\;\;\;\;\;\;\;\;
\frac{1}{\frac{1}{2}{e\mathstrut}^{n(R_n \, - \, D_n)} + 1}
\; + \;
\frac{{e\mathstrut}^{n(R_n \, - \, D_n)}\big(1 - {e\mathstrut}^{-n D_n}\big)}{\frac{1}{4}{e\mathstrut}^{2n(R_n \, - \, D_n)}}
\nonumber \\
& \leq \;
\;\;\;\;\;\;\;\;\;\;\;\;\;\;\;\;\;\;\;
2{e\mathstrut}^{-n(R_n \, - \, D_n)}
\; + \;
\;\;\;
4{e\mathstrut}^{-n(R_n \, - \, D_n)}
\nonumber \\
\mathbb{E} \left[\frac{1}{N + 1}\right] \; & \leq \;
\min \Big\{1, 6{e\mathstrut}^{-n(R_n \, - \, D_n)}\Big\}
\; \leq \;
6{e\mathstrut}^{-n\big|R_n \, - \, D_n\big|^{+}}.
\label{eqUBChebyshev}
\end{align}
We use Jensen's inequality for a lower bound:
\begin{align}
\mathbb{E} \left[\frac{1}{N + 1}\right] \; \overset{\text{Jensen}}{\geq} \; \frac{1}{\mathbb{E}[N] + 1}
\; = \; & \frac{1}{{e\mathstrut}^{n(R_n \, - \, D_n)} + 1}
\nonumber \\
\geq \; & \frac{1}{{e\mathstrut}^{n\big|R_n \, - \, D_n\big|^{+}} + 1}
\; = \;
\frac{{e\mathstrut}^{-n\big|R_n \, - \, D_n\big|^{+}}}{1 + {e\mathstrut}^{-n\big|R_n \, - \, D_n\big|^{+}}}
\; \geq \;
\frac{1}{2}\cdot{e\mathstrut}^{-n\big|R_n \, - \, D_n\big|^{+}}.
\label{eqLBJensen}
\end{align}

Comparing the upper and lower bounds (\ref{eqUBChebyshev}) and (\ref{eqLBJensen}),
we conclude that the exponent in the conditional probability of correct decoding with the knowledge of the joint type $T(y)V(x \, | \, y)$
of the received and the transmitted blocks at the receiver
is given by
\begin{equation} \label{eqIncorrectExp}
\big| R\, - \, D(T\circ V \,\|\, T \times Q)\big|^{+} \, + \, o(1),
\end{equation}
where $o(1)$ is uniform with respect to $T\circ V$.
The exponent in the overall probability of correct decoding of this decoder therefore is given by
\begin{equation} \label{eqObjective}
\min_{\substack{\\\text{types} \;\;  T(y), \, V(x\,|\,y)
}}
\left\{
D(T\circ V \, \| \, Q \circ P) \, + \,
\big|R \, - \, D(T\circ V \,\|\, T \times Q)\big|^{+}
\right\} \; + \; o(1).
\end{equation}
In the limit, as $n\,\rightarrow\,\infty$, the term $o(1)$ disappears and the minimization is performed over all rational distributions $\,T \circ V$.
Since the objective function in (\ref{eqObjective}) is a continuous function of $\,T \circ V$, the infimum over rational distributions equals the minimum over all distributions as intended in the definition (\ref{eqDefImplicitCor}) of the RHS of (\ref{eqCorrML}).

Consider now a suboptimal decoder. The decoder fixes some joint type $\widetilde{T}\circ \widetilde{V}$.
If the type of the received block ${\bf y}$ is not $\widetilde{T}$, the decoder declares an error.
Otherwise, in case the type of the received block is indeed $\widetilde{T}$, the decoder looks for the indices of the codewords with the conditional type $\widetilde{V}$
with respect to ${\bf y}$ and chooses one of them as its estimate $\widehat{m}$ of the transmitted message with uniform probability.
If there are no codewords of the conditional type $\widetilde{V}$ with respect to ${\bf y}$ in the codebook,
then the decoder declares an error.
It follows from the same analysis as before, that the exponent in the probability of correct decoding in this case is given by
\begin{displaymath}
D(\widetilde{T}\circ\widetilde{V} \, \| \, Q \circ P) \, + \,
\big|R \, - \, D(\widetilde{T}\circ \widetilde{V} \,\|\, \widetilde{T} \times Q)\big|^{+} \; + \; o(1).
\end{displaymath}
Consequently, the best possible choice of the fixed type $\widetilde{T}\circ \widetilde{V}$
will result in the exponent of correct decoding given by (\ref{eqObjective}).
We conclude that the best decoder from the above family of suboptimal decoders will produce the same correct-decoding exponent
as the optimal decoder with the genie.
Therefore both result in the correct-decoding exponent of the optimal decoder (ML with tie breaking). $\square$

\bibliographystyle{IEEEtran}

\end{document}